%% file: main.tex
\documentclass[11pt]{article}

\usepackage{agmacros}
\usepackage{andrews_bib}
\DeclareMathOperator{\initialRoot}{init}
\DeclareMathOperator{\scaledSystem}{scaled}
\DeclareMathOperator{\powerseriesRoot}{root}
\DeclareMathOperator{\composedResultant}{composed}

\title{Hilbert's Nullstellensatz is in the Counting Hierarchy}
\author{Robert Andrews\thanks{Cheriton School of Computer Science, University of Waterloo. Part of this work was supported by the Simons Institute for the Theory of Computing, and was conducted when the author was visiting the Institute. Email: randrews@uwaterloo.ca.} \and Abhibhav Garg\thanks{Cheriton School of Computer Science, University of Waterloo. Email: abhibhav.garg@uwaterloo.ca.} \and \'{E}ric Schost\thanks{Cheriton School of Computer Science, University of Waterloo. Email: eschost@uwaterloo.ca.}}
\date{}

\begin{document}

\pagenumbering{roman}

\maketitle

\begin{abstract}
  We show that Hilbert's Nullstellensatz, the problem of deciding if a system of multivariate polynomial equations has a solution in the algebraic closure of the underlying field, lies in the counting hierarchy.
  More generally, we show that the number of solutions to a system of equations can be computed in polynomial time with oracle access to the counting hierarchy.
  Our results hold in particular for polynomials with coefficients in either the rational numbers or a finite field.
  Previously, the best-known bounds on the complexities of these problems were $\PSPACE$ and $\FPSPACE$, respectively.
  Our main technical contribution is the construction of a uniform family of constant-depth arithmetic circuits that compute the multivariate resultant.
\end{abstract}

\tableofcontents

\newpage

\listoftodos
\newpage

\pagenumbering{arabic}

\section{Introduction} \label{section: introduction}
\input{sections/introduction}

\section{Preliminaries} \label{section: preliminaries}
\input{sections/preliminaries}

\section{Uniformity of basic operations on arithmetic circuits} \label{section: arithmetic circuits}
\input{sections/uniformalgebraic}

\section{The multivariate resultant} \label{section: resultant basics}
\input{sections/resultantbasics}

\section{Constant-depth circuits for the resultant} \label{section: resultant constant depth}
\input{sections/resultant}

\section{Deciding the Nullstellensatz in the counting hierarchy} \label{section: counting hierarchy}
\input{sections/circuits_to_counting}

\section{Applications of the Nullstellensatz} \label{section: applications}
\input{sections/applications}

\printbibliography

\end{document}

%% file: sections/introduction.tex
\subsection{Background}

Polynomial equations are ubiquitous throughout mathematics and the sciences.
They capture useful geometric and physical relationships, such as the Pythagorean theorem or the altitude of a projectile over time, and so have been the subject of investigation for centuries.
The quadratic formula, Gaussian elimination, and the Euclidean algorithm, all important tools for solving equations, were known to ancient civilizations and are among the oldest algorithms to survive to the present day.
Renaissance mathematicians later discovered formulas for the solution of cubic and quartic equations, and a great deal of effort was expended on the search for the quintic formula before Abel proved no such formula exists.
The parallel development of numerical techniques, such as Newton's method, provided a means of solving high-degree equations even in the absence of an explicit formula for their solution, and they have since become an essential part of modern computational science.

Interest in solving polynomial equations has endured over time, and as computer science developed, mathematicians were naturally led to questions about the computability and complexity of solving equations.
Already at the turn of the 20\ts{th} century, predating the computer age, \textcite{Hilbert02} posed a list of twenty-three problems that might guide mathematical progress.
Solving systems of equations was the subject of his tenth problem.
In modern terms, Hilbert asked if there is an algorithm that takes as input polynomials $f_1, \ldots, f_m \in \bZ\bs{x_1, \ldots, x_n}$ and decides if the system of equations $f_1 = \cdots = f_m = 0$ has an integer solution.
\textcite{Mat70}, building on work of \textcite{DPR61}, showed that this problem is undecidable.
However, if we seek solutions in other domains, such as the real or complex numbers, this problem becomes decidable and has a different story to tell.

Over the complex numbers---and more generally, an algebraically closed field---the starting point for the solution of polynomial equations is Hilbert's Nullstellensatz.
Given polynomials $f_1, \ldots, f_m \in \bC\bs{x_1, \ldots, x_n}$, the Nullstellensatz says that the system of equations $f_1 = \cdots = f_m = 0$ has no solution in $\bC^n$ if and only if there are polynomials $g_1, \ldots, g_m \in \bC\bs{x_1, \ldots, x_n}$ such that
\[
  f_{1}\br{\vx} g_{1}\br{\vx} + f_{2}\br{\vx} g_{2}\br{\vx} + \cdots + f_{m}\br{\vx} g_{m}\br{\vx} = 1.
\]
Thus, we can decide if the system $f_1 = \cdots = f_m = 0$ has a solution by instead searching for polynomials $g_1, \ldots, g_m$ that satisfy $\sum_{i=1}^m f_i g_i = 1$.
Because of the key role the Nullstellensatz plays in deciding the solvability of systems of equations, the computational problem of deciding whether a system of polynomial equations has a solution is likewise referred to as Hilbert's Nullstellensatz, which we abbreviate as $\HN$.

The Nullstellensatz reduces our original nonlinear problem to a linear one: inspecting the coefficient of the monomial $\vec{x}^{\vec{a}}$ on both sides of the equation $\sum_{i=1}^m f_i g_i = 1$ yields a linear equation in the unknown coefficients of the polynomials $g_1, \ldots, g_m$.
To turn this into an algorithm, we need to bound the degrees of the polynomials $g_1, \ldots, g_m$, should they exist.
\textcite{Hermann1926} was the first to do so.
She showed that if the $f_i$ have degree bounded by $d$ and do not have a common solution, then one can find $g_i$ of degree at most $(2 d)^{n 2^n}$.\footnote{The proofs in Hermann's original paper are incomplete, and the gaps were later filled by \textcite{seidenberg74}.}
This degree bound reduces the task of deciding the solvability of $f_1 = \cdots = f_m = 0$ to solving a system of linear equations of double-exponential size, a task that can be carried out in double-exponential time.
Since then, numerous works have proved stronger bounds on the degrees (and heights) of the $g_i$, leading to bounds of the form $\deg(g_i) \le d^n$ \cite{Brownawell87,CGH88,Kollar88,FG90,Sombra99,KPS01,Jelonek05}.
Naturally, these single-exponential degree bounds imply that $\HN$ can be decided in exponential time, rather than double-exponential time, in the size of the input.

Although degree bounds of the form $\deg(g_i) \le d^n$ are essentially tight for the Nullstellensatz, there is more to say about its complexity, as linear equations can be solved in a surprisingly efficient manner.
In 1976, \textcite{Csanky76} designed an ingenious parallel algorithm to compute the determinant of an $n \times n$ matrix in $O(\log^2 n)$ time.
The ensuing decade saw the design of parallel algorithms for many problems of linear algebra, including the algorithm of \textcite{BGH82} that solves linear systems in $O(\log^2 n)$ parallel time.
Combined with an observation due to \textcite{Borodin1977} that parallel algorithms can be simulated in small space, this leads to an algorithm that solves linear systems using only $O(\log^2 n)$ space.\footnote{The algorithms of \textcite{Csanky76,BGH82} were stated as arithmetic algorithms, without regard to bit complexity, but they can likewise be implemented in $O(\log^2 n)$ parallel time when accounting for bit complexity.}
Scaled up to the exponentially-large systems of equations appearing in the Nullstellensatz (where the matrices have size $d^{\poly\br{n}}$), this results in a $\PSPACE$ algorithm to decide $\HN$ over an algebraically closed field, which is the state of the art for this problem.

Much less is known about the computational hardness of solving systems of equations.
By arithmetizing $\SAT$, it is easy to see that $\HN$ is $\NP$-hard and that counting the number of solutions to a system of equations (if this number is guaranteed to be finite) is $\#\P$-hard.
These are the strongest lower bounds known for the problem, and for good reason.
\textcite{koiran96} showed that, assuming the Generalized Riemann Hypothesis, $\HN$ over the complex numbers is in the complexity class $\AM$.
Under plausible hardness assumptions, \textcite{MV05} showed that $\AM = \NP$, so Koiran's work establishes that $\HN$ is conditionally $\NP$-complete over $\bC$.
In light of this, any lower bound stronger than $\NP$-hardness for $\HN$ would imply either a surprising collapse of complexity classes, or that one of the two conjectures underlying Koiran's $\NP$ algorithm is false.

Over other fields, the complexity of $\HN$ is less well understood.
The aforementioned degree bounds are valid over any field, so the Nullstellensatz can be decided in $\PSPACE$ over any algebraically closed field, as long as the inputs lie in a subfield where arithmetic can be performed efficiently.
Koiran's algorithm was recently extended by \textcite{MBNSW25} to decide the Nullstellensatz over $\overline{\bC(y_1, \ldots, y_k)}$ in $\AM$, again assuming the Generalized Riemann Hypothesis.
The situation is murkier over fields of positive characteristic.
For example, over finite fields, we only know that $\HN$ is $\NP$-hard and can be decided in $\PSPACE$.
No stronger lower bounds, nor better algorithms, are known, even conditionally.

Hilbert's Nullstellensatz, in addition to being a problem of intrinsic interest, also plays a central role in algebraic complexity, similar to the one occupied by boolean satisfiability in classical complexity theory.
Over a commutative ring $R$, the problem $\HN_R$---that is, deciding whether a system of polynomial equations over $R$ has a solution in $R$---is complete for $\NP_R$, the analogue of $\NP$ in the Blum--Shub--Smale model of computation \cite{BSS89,BCSS98}.
The counting variant of $\HN_R$, denoted by $\#\HN_R$, where one must count the number of solutions to a system of equations over a ring $R$, is similarly important in the theory of counting problems.
\textcite{BC06} defined the complexity classes $\#\P_\bR$ and $\#\P_\bC$, the counting versions of $\NP_\bR$ and $\NP_\bC$, respectively, and showed that $\#\HN_\bR$ and $\#\HN_\bC$ are complete for their respective classes.
The Nullstellensatz is also a useful algorithmic primitive in computational algebraic geometry: many problems of interest, such as computing the dimension of an algebraic variety over an algebraically closed field \cite{koiran1997randomized}, use the Nullstellensatz as an essential subroutine.

\subsection{Our results}

We show that a variety of problems related to solving systems of polynomial equations, chief among them the Nullstellensatz, can be decided in $\CH$, the counting hierarchy.
We work in the Turing machine model, representing polynomials by the binary encoding of their coefficients.
Throughout, we consider systems of equations that have coefficients in the rational numbers, a number field, a finite field, or a polynomial ring over such a field, all of which can be efficiently encoded in binary, and are interested in their solvability over the algebraic closure of the coefficient field.

We use the dense representation of polynomials, where all monomials and their coefficients are listed, including monomials with a coefficient of zero.
One could also consider more expressive encodings of a multivariate polynomial, either writing only the nonzero monomials of the polynomial (the \emph{sparse representation}) or describing an arithmetic circuit that computes it (the \emph{straight-line program representation}).
By adding polynomially-many extension variables, both the sparse and straight-line program representations can be converted to the dense representation.
This conversion preserves both the satisfiability and number of solutions (if this number is finite) of the original system of equations.
\roberttodo[inline]{I promoted the footnote on representations of polynomials to a paragraph on its own.}

Our main result is that Hilbert's Nullstellensatz---that is, whether a system of equations over a field $\bF$ has a solution in the algebraic closure $\overline{\bF}$---can be decided in $\CH$.

\begin{theorem}[see \cref{theorem: nullstellensatz in ch}] \label{theorem: nullstellensatz intro}
  Let $\bF$ be one of the fields $\bQ$, $\bQ(y_1,\ldots,y_k)$, a number field $\bK$, $\bK(y_1, \ldots, y_k)$, the finite field $\bF_q$, or $\bF_q(y_1,\ldots,y_k)$.
  Then Hilbert's Nullstellensatz over $\bF$ can be decided in $\CH$.
\end{theorem}

Before proceeding to our other results, we pause briefly to recall the counting hierarchy.
The counting hierarchy $\CH$ was first defined by \textcite{Wagner86} to characterize the complexity of combinatorial counting problems.
The $k$\ts{th} level of $\CH$, denoted $\C_k\P$, is defined inductively as
\begin{align*}
  \C_0\P &\coloneqq \P \\
  \C_1\P &\coloneqq \PP \\
  \C_k\P &\coloneqq \PP^{\C_{k-1}}.
\end{align*}
The counting hierarchy is the union of these classes, i.e., $\CH \coloneqq \bigcup_{k \ge 0} \C_k\P$.
Alternatively, one can define $\CH$ as the class of languages that are decidable by polynomial-time Turing machines with a constant number of polynomially-bounded majority quantifiers, analogous to how $\PH$ captures languages that are decidable in polynomial time by Turing machines that may use a constant number of polynomially-bounded existential and universal quantifiers.

Where does the counting hierarchy sit in the landscape of complexity classes?
Toda's result that $\PH \subseteq \P^{\PP}$ \cite{Toda91} shows that the polynomial-time hierarchy is contained in the second level of $\CH$, so we have the inclusion $\PH \subseteq \CH$.
It is also easy to see that $\CH \subseteq \PSPACE$ as a consequence of $\PP \subseteq \PSPACE$.
As the name suggests, the counting hierarchy captures the complexity of counting: as a consequence of the definition, if $\FP = \#\P$ (i.e., if the permanent of a matrix can be computed in polynomial time), then we have the collapse $\P = \CH$.
For more details on the counting hierarchy, we refer the reader to \textcite{AW90,Toran91}.

Returning to our results, Hilbert's Nullstellensatz shows that a system of polynomial equations $f_1, \ldots, f_m \in \bF[x_1, \ldots, x_n]$ is unsatisfiable if and only if the constant polynomial $1$ is in $\radideal{f_1, \ldots, f_m}$, the radical of the ideal generated by $f_1, \ldots, f_m$.
More generally, the Nullstellensatz says that a polynomial $h \in \bF[x_1, \ldots, x_n]$ vanishes on the common zeroes of $f_1, \ldots, f_m$, called the \emph{variety} of $f_1, \ldots, f_m$ and denoted by $\var{f_1, \ldots, f_m}$, if and only if $h \in \radideal{f_1, \ldots, f_m}$.
The problem of deciding if $h \in \radideal{f_1, \ldots, f_m}$, known as the \emph{radical ideal membership problem}, is another basic problem in computational algebraic geometry.
The previously-mentioned degree bounds for the Nullstellensatz also apply to the radical ideal membership problem, with an additional factor to account for the degree of the polynomial $h$.
These degree bounds imply that radical ideal membership, like $\HN$, can be solved in $\PSPACE$.
By combining our main theorem with the Rabinowitsch trick, we improve this bound, showing that radical ideal membership can be decided in $\CH$.

\begin{theorem}[see \cref{lemma: application radical membership}]
  Let $\bF$ be one of the fields $\bQ$, $\bQ(y_1,\ldots,y_k)$, a number field $\bK$, $\bK(y_1,\ldots,y_k)$, the finite field $\bF_q$, or $\bF_q(y_1,\ldots,y_k)$.
  Then the radical ideal membership problem over $\bF$ can be decided in $\CH$.
\end{theorem}

In contrast, \textcite{MM82} proved that the general ideal membership problem is $\EXPSPACE$-complete.
Because of this, a similar improvement in the complexity of the ideal membership problem would apply to the class $\EXPSPACE$ itself.

Once we know that a system of equations $f_1, \ldots, f_m \in \bF\bs{x_1, \ldots, x_n}$ has a solution, a natural next step is to understand how many solutions this system has.
One way to make this question precise is to ask for the dimension of the variety $\var{f_1, \ldots, f_m}$, which is a coarse measure of the size of the solution set.
A straightforward application of our main theorem lets us compute this dimension in $\FP^\CH$.

\begin{theorem}[see \cref{lemma: application dimension}]
  Let $\bF$ be one of the fields $\bQ$, $\bQ(y_1,\ldots,y_k)$, a number field $\bK$, $\bK(y_1, \ldots, y_k)$, the finite field $\bF_q$, or $\bF_q(y_1,\ldots,y_k)$.
  Given a set of polynomials $f_{1}, \dots, f_{m}$ in $\bF\bs{x_{1}, \dots,x_{n}}$, one can compute the dimension of the variety $\var{f_{1}, \dots, f_{m}}$ in $\FP^{\CH}$.
\end{theorem}

When the variety $\var{f_1, \ldots, f_m}$ is zero-dimensional, it consists of a finite set of points, so it makes sense to ask for the number of solutions to the system $f_1 = \cdots = f_m = 0$.
Unlike the last two applications of our main theorem, which follow by simple reductions to $\HN$, we are not aware of a simple reduction from counting points on a variety to $\HN$.
Despite this, our techniques easily extend to this problem, allowing us to compute the number of points on a variety in $\FP^\CH$.

\begin{theorem}[see \cref{theorem: counting nullstellensatz in ch}] \label{theorem: counting intro}
  Let $\bF$ be one of the fields $\bQ$, $\bQ(y_1,\ldots,y_k)$, a number field $\bK$, $\bK(y_1,\ldots,y_k)$, the finite field $\bF_q$, or $\bF_q(y_1,\ldots,y_k)$.
  Given a set of polynomials $f_{1}, \dots, f_{m}$ in $\bF\bs{x_{1}, \dots,x_{n}}$, one can compute the cardinality of the variety $\var{f_1, \ldots, f_m}$ in $\FP^{\CH}$.
\end{theorem}

Counting points on a variety is a basic problem of counting complexity whose precise classification is not well understood.
\textcite{BC06} introduced the complexity class $\GCC$, for \emph{geometric counting complex problems}, and showed that counting points on a complex variety is complete for $\GCC$.\footnote{\textcite{BC06} also introduced the counting classes $\#\P_\bC$ and $\#\P_\bR$ as analogues of $\#\P$ in the Blum--Shub--Smale model of computation. Among other results, they proved that counting points on complex or real varieties are complete problems for $\#\P_\bC$ and $\#\P_\bR$, respectively.}
The class $\GCC$ and its functional analogue $\FP^{\GCC}$ contain natural counting problems in algebraic geometry, including computing the geometric degree of a variety \cite{BC06}, the Euler characteristic of a variety \cite{BCL05}, and the Hilbert polynomial of a smooth equidimensional variety \cite{BL07}.
By arithmetizing $\#\SAT$, it is easy to see that counting points on a variety is at least as hard as counting the number of satisfying assignments to a boolean formula, so $\#\P \subseteq \GCC$.
Prior to our work, the best known upper bound on $\GCC$ was $\GCC \subseteq \FPSPACE$ \cite{BC06}.
\cref{theorem: counting intro} improves this to $\GCC \subseteq \FP^{\CH}$.

\cref{theorem: counting intro} explains why our techniques only prove an upper bound of $\CH$ on the complexity of $\HN$, as opposed to placing $\HN$ in the polynomial-time hierarchy.
The same techniques underlie the proofs of \cref{theorem: nullstellensatz intro,theorem: counting intro}.
Because \cref{theorem: counting intro} addresses a $\#\P$-hard problem, the best upper bound we can hope to prove with these techniques---excepting a surprising collapse of complexity classes---is $\#\P$.

Finally, we mention a straightforward application of \cref{theorem: nullstellensatz intro} to computing tensor rank, an important problem in algebraic complexity related to the determination of the exponent of matrix multiplication.
By a direct reduction to $\HN$, we can compute the rank of a given tensor in $\FP^\CH$.

\begin{theorem}[see \cref{lemma: application tensor rank}]
  Let $\bF$ be one of the fields $\bQ$, $\bQ(y_1,\ldots,y_k)$, a number field $\bK$, $\bK(y_1,\ldots,y_k)$, the finite field $\bF_q$, or $\bF_q(y_1,\ldots,y_k)$.
  Given a tensor $T \in \bF^{d_1} \otimes \cdots \otimes \bF^{d_k}$, one can compute the tensor rank of $T$ over $\overline{\bF}$ in $\FP^\CH$.
\end{theorem}

\subsection{Proof overview}

As we saw, we can decide if a system of polynomials $f_1, \ldots, f_m \in \bC\bs{x_1,\ldots,x_n}$ of degrees at most $d$ has a solution by deciding if the linear system
\[
  f_{1}\br{\vx} g_{1}\br{\vx} + f_{2}\br{\vx} g_{2}\br{\vx} + \cdots + f_{m}\br{\vx} g_{m}\br{\vx} = 1.
\]
is solvable, where the unknowns are the coefficients of the polynomials $g_1, \ldots, g_m \in \bC\bs{x_1, \ldots, x_n}$ of degrees at most $d^{n}$, which results in a system of equations of size $d^{\poly{\br{n}}}$.
To improve the complexity of $\HN$, we will take advantage of the fact that these linear systems are not arbitrary, but instead highly structured.

\subsubsection{The resultant}

To see what makes linear systems like $\sum_{i=1}^m f_i g_i = 1$ special, let us consider the problem of deciding if a system of two equations $f_1, f_2 \in \bC[x]$ in one variable has a solution.
The determinant of the linear system $f_1(x) g_1(x) + f_2(x) g_2(x) = 1$ is a well-known function of the coefficients of $f_1$ and $f_2$ called their \emph{resultant}, denoted by $\res(f_1, f_2)$.
The resultant has the remarkable property that $\res(f_1, f_2) = 0$ if and only the system $f_1(x) = f_2(x) = 0$ has a solution.
By definition, the resultant is the determinant of a polynomially-large matrix, so it can be computed in polylogarithmic space using Csanky's algorithm \cite{Csanky76} for the determinant.
This corresponds to a scaled-down version of the $\PSPACE$ algorithm for $\HN$.

Although the resultant is defined as a determinant, the resultant can be computed more efficiently than the determinant itself.
\textcite{AW24} showed that the resultant of two polynomials can be computed by arithmetic circuits of constant depth and polynomial size, something which is provably impossible for the determinant \cite{LST25}.
Because iterated addition and multiplication of rational numbers can be computed by threshold circuits of constant depth and polynomial size, i.e., in the class $\TC^0$ \cite{HAB02}, this implies that the resultant can also be computed in $\TC^0$ when the input is represented in binary. 
As we will soon see, the class $\TC^0$ corresponds to a scaled-down version of the counting hierarchy $\CH$, so this algorithm is evidence that it may be possible to improve the complexity of $\HN$.
The improved algorithm for the resultant relies on the \emph{Poisson formula}, which expresses $\res(f_1, f_2)$ in terms of the complex roots of $f_1$ and $f_2$.
Suppose $f_1$ and $f_2$ are monic, and that $f_1$ factors as $f_1(x) = \prod_{i=1}^d (x - \alpha_i)$ over $\bC$.
In this case, the Poisson formula expresses the resultant of $f_1$ and $f_2$ as
\[
  \res(f_1, f_2) = \prod_{i=1}^d f_2(\alpha_i).
\]
Although not immediately obvious, this identity can be used to compute the resultant in constant depth, since the coefficients of $f_1$ provide useful information about its roots $\alpha_1, \ldots, \alpha_d$.

The resultant can be generalized to many polynomials in several variables, and designing algorithms to compute it will be essential for our work.
Suppose $F_0, \ldots, F_N \in \bC\bs{x_0, \ldots, x_n}$ are homogeneous polynomials of degrees $d_0, \ldots, d_n$, respectively.
There is a polynomial function of their coefficients, likewise called the resultant and denoted $\res(F_0,\ldots,F_n)$, such that $\res(F_0, \ldots, F_n) = 0$ exactly when the system $F_0(\vx) = \cdots = F_n(\vx) = 0$ has a nonzero solution.\footnote{Because the polynomials $F_0, \ldots, F_n$ are homogeneous, the all-zeroes point is always a solution of the system $F_0 = \cdots = F_m = 0$. The resultant detects when this system has a nontrivial solution, or equivalently, when this system has a solution in projective space $\bP^n$.}
Thus, if we can compute the resultant efficiently, we can also decide $\HN$, at least in the case of homogeneous systems where the number of polynomials and variables match.
The multivariate resultant can be expressed as the quotient of two minors of a structured matrix of size $\binom{d_0 + \cdots + d_n}{n}$, called the \emph{Macaulay matrix} of $F_0, \ldots, F_n$.
At worst, this matrix is exponentially large compared to $\sum_{i=0}^n \binom{n + d_i}{n}$, the number of coefficients of the polynomials $F_0, \ldots, F_n$, so any one of its minors can be computed in $\PSPACE$.
Designing an improved algorithm to compute the multivariate resultant in $\FP^\CH$ will be the main step in our proof that $\HN \in \CH$.

\subsubsection{The counting hierarchy}

Before describing how to compute the resultant in $\FP^\CH$, let us see what sort of algorithmic power $\CH$ provides.
Just as the polynomial-time hierarchy $\PH$ is connected to bounded-depth boolean circuits built from AND, OR, and NOT gates \cite{FSS84}, the counting hierarchy is related to bounded-depth boolean circuits built from threshold and negation gates.
For us, it will be useful to view $\CH$ through its characterization in terms of polynomially-bounded majority quantifiers: a language $L \subseteq \bc{0,1}^*$ is in $\CH$ if there is a polynomially-bounded function $p$ and a polynomial-time Turing machine $M$ such that
\[
  x \in L \iff \maj y_1 \in \bc{0,1}^{p(|x|)} \cdots \maj y_k \in \bc{0,1}^{p(|x|)} \text{$M(x, y_1, \ldots, y_k)$ accepts},
\]
where the majority quantifier $\maj y$ stipulates that the subsequent formula is true for the majority of choices of the variable $y$.
With a single majority quantifier, we can compute the majority function (and more generally, any threshold function) over an exponentially-large set of bits, as long as any single bit in this set can be computed in polynomial time.
By repeating this observation, we can evaluate an exponentially-large constant-depth threshold circuit in $\CH$, as long as the connectivity properties of this circuit can be decided by a polynomial-time Turing machine.
More precisely, if a boolean function $f : \bc{0,1}^* \to \bc{0,1}$ can be computed by a polylogtime-uniform family of threshold circuits of constant depth and exponential size, then $f \in \CH$.

Bounded-depth threshold circuits are surprisingly powerful.
Basic operations on rational numbers, such as iterated addition, iterated multiplication, and division with remainder, can all be computed by logtime-uniform families of $\TC^0$ circuits \cite{HAB02}.
This implies that threshold circuits can simulate arithmetic circuits while only losing a constant factor in depth and a polynomial factor in size.
In particular, in $\CH$, we can simulate a polylogtime-uniform family of arithmetic circuits of bounded depth and exponential size.
This is the main property of $\CH$ that we use in our algorithms.

\subsubsection{Computing the resultant in constant depth}

Now we return to the task of computing the multivariate resultant in $\FP^\CH$.
From our discussion on the power of $\CH$, it suffices to construct a polylogtime-uniform family of arithmetic circuits of constant depth and exponential size that computes the resultant, and this is the route we will take.

Just as in the case of two polynomials, the multivariate resultant can be expressed in terms of the evaluations of one polynomial at the common roots of the others.
To set notation, let
\begin{align*}
  \overline{F}_i(\vx) &\coloneqq F_i(x_0,\ldots,x_{n-1}, 0) \\
  f_i(\vx) &\coloneqq f_i(x_0,\ldots,x_{n-1},1).
\end{align*}
Then, assuming the smaller resultant $\res(\overline{F}_0,\ldots,\overline{F}_{n-1})$ is not zero, the \emph{Poisson formula} shows the resultant factors as
\[
  \res(F_0,\ldots,F_n) = \res(\overline{F}_0,\ldots,\overline{F}_{n-1})^{d_{n-1}} \prod_{\valpha \in \var{f_0,\ldots,f_{n-1}}} f_n(\valpha)^{m(\valpha)},
\]
where $\var{f_0,\ldots,f_{n-1}} \subseteq \bC^{n}$ is the set of common roots of $f_0,\ldots,f_{n-1}$ (which is necessarily finite by the assumption on the smaller resultant), and $m(\valpha)$ is the multiplicity of $\valpha$ as a solution to the system $f_0 = \cdots = f_{n-1} = 0$.
By recursively applying the Poisson formula to the smaller resultant $\res(\overline{F}_0, \ldots, \overline{F}_{n-1})$, we see that the original resultant $\res(F_0, \ldots, F_n)$ we wanted to compute can be expressed as a product of terms of the form $\prod_{\valpha \in \var{f_0,\ldots,f_{n-1}}} f_n(\valpha)^{m(\valpha)}$.
If we can compute one such a term using a bounded-depth arithmetic circuit, then by computing all terms in parallel and multiplying them, we obtain a bounded-depth arithmetic circuit for the resultant itself.

To compute a product of the form $\prod_{\valpha \in \var{f_0,\ldots,f_{n-1}}} f_n(\valpha)^{m(\valpha)}$, it appears that we need to solve the system of equations $f_0 = \cdots = f_{n-1} = 0$.
In the case $n = 1$, this can avoided by using the Girard--Newton identities \cite{AW24}.
While there are similar identities in the case $n \geq 2$ (see \cite{AK81} and references therein), they are more complicated, and it is not clear if they can be implemented in constant depth.
Instead, we will use a computationally explicit version of the implicit function theorem to express the solutions of $f_0 = \cdots = f_{n-1} = 0$ as power series in the coefficients of $f_0, \ldots, f_{n-1}$, where initial segments of these power series can be computed by constant-depth circuits.
A similar idea appears in recent work of \textcite{BKRRSS25a}, who used Lagrange inversion to show that low-depth arithmetic circuits are closed under factorization.
Their work, particularly its use of an explicit version of the implicit function theorem, was a key inspiration for the results of this paper. \roberttodo{I added an extra sentence extolling the influence of \cite{BKRRSS25a}}

The implicit function theorem requires the Jacobian of $(f_0, \ldots, f_{n-1})$ to be invertible at $\valpha$, which may not be true of the system we are given.
To remedy this, we use the method of homotopy continuation to obtain the roots of $f_0 = \cdots = f_{n-1} = 0$ by solving a different---but related---system of equations.
We will choose another system of homogeneous equations $G_0, \ldots, G_n \in \bC\bs{x_0, \ldots, x_n}$, unrelated to the $F_i$, where the solutions of $G_0 = \cdots = G_n = 0$ are explicitly known in advance and at which the Jacobian is invertible.
Letting $t$ be a fresh variable, we then consider the system of equations $H_0 = \cdots = H_n = 0$, where $H_i$ is given by
\[
  H_i(t, \vx) \coloneqq (1 - t) \cdot G_i(\vx) + t \cdot F_i(\vx).
\]
This system is easy to solve at $t = 0$, since it simplifies to $G_0 = \cdots = G_n = 0$, whose solutions are known to us.
At $t = 1$, we recover the original system we wanted to solve.
In particular, if we can instead compute the resultant $\res(H_0, \ldots, H_n)$, then evaluating at $t = 1$ will recover $\res(F_0, \ldots, F_n)$.

Let 
\begin{align*}
  g_i(\vx) &\coloneqq G_i(x_0, \ldots, x_{n-1}, 1) \\
  h_i(t,\vx) &\coloneqq H(t,x_0,\ldots,x_{n-1},1).
\end{align*}
For each solution $\valpha \in \bC^n$ of the system $g_0 = \cdots = g_{n-1} = 0$, there is a corresponding power series solution $\vvarphi_{\valpha}(t) \in \bC\bsd{t}^n$ of $h_0 = \cdots = h_{n-1} = 0$, and all solutions of $h_0 = \cdots = h_{n-1} = 0$ arise in this way.
The coefficients of the power series solution $\vvarphi_{\valpha}(t)$ are polynomial functions of the coefficients of the original system $f_0, \ldots, f_{n-1}$.
Importantly, the first $(n+1) d^{n}$ coefficients of these power series can be computed from the coefficients of $f_0, \ldots, f_{n-1}$ using an arithmetic circuit of constant depth and size bounded by $d^{\poly(n)}$, where $d = \max(d_0, \ldots, d_{n-1})$.
This allows us to compute an approximate solution $\tilde{\vvarphi}_{\valpha}(t) \in \bC[t]^n$ that agrees with $\vvarphi_{\valpha}(t)$ up to degree $(n+1) d^{n}$.
By evaluating $h_n$ on each of the approximate solutions $\tilde{\vvarphi}_{\valpha}(t)$ and using the Poisson formula, we compute a polynomial $\tilde{\res}(H_0, \ldots, H_n)$ that agrees with the resultant $\res(H_0, \ldots, H_n)$ modulo $t^{(n+1) d^n}$.
Because the resultant is a polynomial of degree at most $(n+1) d^n$, we can use interpolation to recover the terms of $\tilde{\res}(H_0,\ldots,H_n)$ of degree at most $(n+1)d^n$ in $t$, which necessarily equal the resultant $\res(H_0,\ldots,H_n)$.

This approach produces a family of arithmetic circuits of constant depth and exponential size to compute the resultant.
To conclude that the resultant can be computed in $\FP^\CH$, we need to ensure that this family of circuits is sufficiently uniform.
Most parts of this algorithm are easily shown to be uniform.
The only part that is not immediately uniform is our use of evaluation-interpolation to compute the coefficients of a polynomial from its evaluations.
Doing this requires constructing a uniform family of constant-depth arithmetic circuits to compute the inverse of a Vandermonde matrix, which we do using an explicit description of the inverse of the Vandermonde matrix at the points $1, 2, \ldots, N$ in terms of Stirling numbers.

So far, we have seen how to compute the resultant of polynomials over $\bQ$ in $\FP^\CH$.
The same algorithm would likewise produce constant-depth arithmetic circuits to compute the resultant over all fields, but the notion of uniformity for arithmetic circuits becomes more difficult to work with over finite fields.
In particular, our circuits would need access to constants from an extension field, and one must bound the uniformity of these constants in some way.
In contrast, over $\bQ$, constant-free circuits are sufficient for our purposes, and one only needs to bound the uniformity of the circuit's structure.
To compute the resultant over other fields, such as the finite field $\bF_p$, we use the fact that the resultant is essentially the same polynomial over all fields.
In particular, for polynomials $F_0, \ldots, F_n \in \bF_p[x_0,\ldots,x_n]$, if we lift $F_i$ to the polynomial $\hat{F}_i \in \bZ[x_0,\ldots,x_n]$, then the resultant satisfies
\[
  \res(F_0, \ldots, F_n) = \res(\hat{F}_0, \ldots, \hat{F}_n) \bmod{p}.
\]
Thus, to compute the resultant over $\bF_p$ in $\FP^\CH$, we lift the input to the integers, compute the resultant over the integers, and then reduce this value modulo $p$.

\subsubsection{From homogeneous to affine systems}

Computing the resultant in $\FP^\CH$ allows us to decide if a system of $n$ homogeneous equations in $n$ variables has a nonzero solution, or in other words, if it has a solution in projective space $\bP^{n-1}$.
It is not difficult to extend this to an algorithm that handles non-square projective systems: if we are given $m$ equations $F_1 = \cdots = F_m = 0$ in $n$ variables where $m > n$, one can show that the system $G_1 = \cdots = G_n = 0$ obtained by taking each $G_i$ to be a random linear combination of the $F_j$ has the same set of solutions with high probability.
Deciding the solvability of inhomogeneous equations is more difficult.
The natural approach is to homogenize a given inhomogeneous system, but this does not necessarily preserve the solvability of the system of equations.
For example, the system $x + y - 1 = x + y - 2 = 0$ has no solution, but its homogenization $x + y - z = x + y - 2z = 0$ has the nonzero solution $(1, -1, 0)$.

Fortunately, this problem has already been solved.
The \emph{generalized characteristic polynomial} of \textcite{canny90}, later extended by \textcite{Ierardi89}, allows us to decide if an inhomogeneous system of equations has a solution by performing a resultant computation on a slight perturbation of the given system.
A small variation on the construction of the generalized characteristic polynomial produces a univariate polynomial whose roots are in one-to-one correspondence with solutions of the given system of equations, assuming this number is finite.
This reduces the task of counting solutions to a multivariate system of equations to counting the number of distinct roots of a univariate polynomial of exponentially-large degree, but whose coefficients can be computed in $\FP^\CH$.
To do this in $\FP^\CH$, we adapt the squarefree factorization algorithm of \textcite{AW24} that computed the squarefree decomposition of a univariate polynomial using arithmetic circuits of constant depth and polynomial size.
Scaled up to polynomials of exponential degree, this results in an $\FP^\CH$ algorithm that counts the number of distinct roots of a given polynomial.

\subsection{Organization}

The rest of this paper is organized as follows.
\cref{section: preliminaries} introduces notation and also introduces the notion of uniformity we use for boolean and arithmetic circuit families.
It also covers preliminary material in computational algebra.
In \cref{section: arithmetic circuits}, we show that basic operations, including polynomial interpolation, can be performed by uniform families of constant-depth arithmetic circuits.
\cref{section: resultant basics} introduces the multivariate resultant, collecting well-known properties of the resultant and describing how the resultant can be used to decide the Nullstellensatz and count solutions to zero-dimensional systems of equations.
In \cref{section: resultant constant depth}, we construct a uniform family of constant-depth arithmetic circuits that compute the multivariate resultant.
\cref{section: counting hierarchy} transfers this circuit family to the boolean setting, obtaining $\CH$ algorithms for the multivariate resultant and, as a consequence, the Nullstellensatz.
Finally, we conclude in \cref{section: applications} with a few straightforward applications of our algorithm for the Nullstellensatz.

%% file: sections/preliminaries.tex
\subsection{Notation and conventions} \label{subsection: notation}

Throughout this work, we use $\bF$ to denote a field.
Various results that we use and present have differing requirements on the field $\bF$
(for  example, some results require the field to be large enough).
Each statement will clearly state the requirement on the field $\bF$.
If we make no such qualification in the statement of a particular result, then that result holds for any field.
We denote by $\overline{\bF}$ the algebraic closure of $\bF$.

The \emph{coefficient subfield} of a field $\bF$ is the largest algebraic extension of the prime field of $\bF$ within $\bF$.
For example, if $\bF = \bQ\br{y}$, then the coefficient subfield of $\bF$ is $\bQ$, and if $\bF = \bF_{p^{a}}$, then $\bF$ is its own coefficient subfield.

We write $\bF\bs{x_1, \ldots, x_n}$ and $\bF\br{x_1, \ldots, x_n}$, respectively, for the ring of polynomials and field of rational functions in the variables $x_1, \ldots, x_n$ and with coefficients in $\bF$.
We abbreviate vectors as $\vec{x} = (x_1, \ldots, x_n)$.
For $\vec{a} \in \bN^n$, we write $\vec{x}^{\vec{a}}$ for the monomial $x_1^{a_1} \cdots x_n^{a_n}$.
We write $\abs{\va}$ for the sum $\sum_{i=1}^{n} a_{i}$.

By a \emph{form} of degree $d$, we mean a homogeneous element of $\bF\bs{x_{1}, \dots, x_{n}}$ of degree $d$.
We denote forms by capital letters and inhomogeneous polynomials by lowercase letters.
Often, we will work with a pair of polynomials where one is the (de-)homogenization of the other.
In these cases, we use upper- and lowercase variants of the same letter as a mnemonic device to indicate this relationship.
For example, if $F \in \bF\bs{x_{0}, x_{1}, \ldots, x_{n}}$ is a form, we will denote its dehomogenization $F(1, x_{1}, \dots, x_{n})$ by $f(x_1, \ldots, x_{n})$.
Likewise, if $f \in \bF\bs{x_1, \ldots, x_n}$ is an inhomogeneous polynomial of degree $d$, we write $F$ for its homogenization $x_0^d f(x_1 / x_0, \ldots, x_n / x_0)$.
In some settings we might choose to homogenize and dehomogenize using a different variable, which will be made clear in context.

For polynomials $f_1, \ldots, f_m \in \bF\bs{x_1, \ldots, x_n}$, we write $\Var(f_1, \ldots, f_m)$ for the set of common zeroes of $f_1, \ldots, f_m$ in $\overline{\bF}^n$.
When dealing with forms $F_1, \ldots, F_m \in \bF\bs{x_0, \ldots, x_n}$, we instead consider their zero set $\Var(F_1, \ldots, F_m)$ in the projective space $\bP_{\overline{\bF}}^{n}$.
We usually suppress the field dependence from this notation and just write $\bP^{n}$ for $n$-dimensional projective space.

We use $\vand{x_{1}, \dots, x_{n}}$ to denote the Vandermonde matrix defined using $x_{1}, \dots, x_{n}$.
Specifically, the entry in position $i, j$ in the above matrix is $x_{i}^{j-1}$.

We will frequently have to refer to the trailing term and trailing part of a polynomial with respect to some variable.
For convenience, we define the following piece of notation.

\begin{definition} \label{definition: trailing term}
  Suppose $f \in \bF\bs{t, y_{1}, \dots, y_{m}}$ is a non-zero polynomial.
  The \emph{trailing term of $f$ with respect to $t$}, denoted by $\Tt_t f(y_1, \ldots, y_m)$, is the coefficient of the lowest-degree term in $f$ when viewed as a polynomial in $t$.
  Concretely, if $f = t^{i} f_{i}(y_{1}, \dots, y_{m}) + t^{i+1} f_{i+1}(y_{1}, \dots, y_{m}) + \cdots + t^{d} f_{d}(y_{1}, \dots, y_{m})$ where $f_i(\vy) \neq 0$, then $\Tt_{t} f = f_{i}$.

  The \emph{trailing part of $f$ with respect to $t$}, written $\TP_t f$, is the polynomial obtained by factoring out the highest power of $t$ from $f$.
  In the above setting, we have $\TP_t f = f_i + t f_{i+1} + \cdots + t^{d - i} f_d$.

  By convention, the trailing term and the trailing part of the zero polynomial are zero.
\end{definition}

We follow the convention that a variety is any set of solutions of a system of polynomial equations over an algebraically closed field.
In particular, varieties need not be irreducible.
We also follow the convention that the degree of a reducible projective or affine variety (denoted $\deg V$) is the sum of the degrees of all irreducible components.
The greatest common divisor of any two univariate polynomials with coefficients in a field $\bF$ will always be monic.

\subsection{Uniformity of boolean and arithmetic circuit families} \label{subsection: arithmetic}

In this subsection, we discuss the notion of uniformity for boolean and arithmetic circuit families, starting with the boolean case.
We restrict our discussion to families of boolean threshold circuits.

There are a number of different ways of measuring the uniformity of boolean circuits \cite{vollmer1999}.
In our work, we measure the uniformity of a boolean circuit by the complexity of its \emph{direct connection language}, defined below.
This notion is particularly well-suited for very weak circuit classes, such as constant-depth circuits.
We follow the presentation of \textcite{vollmer1999}, but with one important modification: while most work on uniformity addresses circuit families indexed by a single parameter $n \in \bN$, we will work with circuit families indexed by an arbitrary number of parameters $n_1, \ldots, n_k \in \bN$, and so we adapt the notion of the direct connection language to this setting.
Before we define the direct connection language, we must define the notion of \emph{admissible encodings} of circuits.

\begin{definition}[{\cite[Definition~2.14]{vollmer1999}}] \label{defintion: boolean admissible encoding}
  Let $\mathcal{C} = (C_{k, n_1, \ldots, n_k})_{k, n_1, \ldots, n_k \in \bN}$ be a family of boolean circuits of size $s \coloneqq s(k, n_1, \ldots, n_k)$.
  An \emph{admissible encoding of $\cC$} is a numbering of the gates of each circuit $C = C_{k, n_{1}, \dots, n_{k}}$ in the family with the following properties.
  \begin{enumerate}[noitemsep]
    \item
      If $C$ has $m$ input gates, then they are numbered $0, \dots, m-1$.
    \item
      If $C$ has $m'$ output gates, then they are numbered $m, \dots, m + m' - 1$.
    \item
      There is a constant $c$, depending only on the circuit family, such that the binary representation of the number of any gate in $C$ has length at most $c \cdot \log\br{s}$.
      \qedhere
  \end{enumerate}
\end{definition}

For circuit families of threshold circuits with admissible encodings, we define the direct connection language as follows.
\begin{definition}[Direct connection language] \label{defintion: boolean direct connection language}
  Let $\mathcal{C} = (C_{k, n_1, \ldots, n_k})_{k, n_1, \ldots, n_k \in \bN}$ be a family of threshold circuits with an admissible encoding.
  Suppose a numbering of the gate types (AND, OR, NOT, MAJ, 0, 1) is fixed.
  The \emph{direct connection language of $\mathcal{C}$}, denoted by $L_{DC}(\mathcal{C})$, is the language consisting of binary encodings of tuples $(k, n_1, \ldots, n_k, a, p, b)$, where
  \begin{enumerate}[noitemsep]
    \item
      $a$ is the number of a gate in $C_{k, n_{1}, \dots, n_{k}}$.
    \item
      If $p \neq \epsilon$, then $b$ is the number of a gate in $C_{k, n_{1}, \dots, n_{k}}$ that is a direct predecessor of $a$.
    \item
      If $p = \epsilon$, then $b$ encodes the gate type of $a$ (using the fixed numbering of gate types).
      \qedhere
  \end{enumerate}
\end{definition}

By modifying the allowed gate types, the above definition can also be extended to other types of families of boolean circuits, however we will only work with threshold circuits in this article.
The uniformity of a circuit family can be measured using the complexity of the direct connection language as follows.

\begin{definition} \label{definition: uniform threshold circuit}
  Let $\mathcal{C} = (C_{k, n_1, \ldots, n_k})_{k, n_1, \ldots, n_k \in \bN}$ be a family of threshold circuits of size $s \coloneqq s(k, n_1, \ldots, n_k)$ with an admissible encoding.
  We say that the circuit family $\mathcal{C}$ is \emph{polylogtime-uniform} if 
  \begin{enumerate}[noitemsep]
    \item
      its direct connection language $L_{DC}(\mathcal{C})$ can be decided in time polynomial in $\log(s) + \log(k) + \sum_{i=1}^k \log(n_i)$, and
    \item
      the number of inputs and outputs of $C_{k,n_1, \ldots, n_k}$, and an upper bound on the size $s$ can be computed in time polynomial in $\log(s) + \log(k) + \sum_{i=1}^k \log(n_i)$, given $(k, n_1, \ldots, n_k)$ as input.
  \end{enumerate}
  We say that the circuit family $\mathcal{C}$ is \emph{logtime-uniform} if both of these tasks can be done in time linear in $\log\br{s} + \log\br{k} + \sum_{i=1}^{k} \log\br{n_{i}}$.
\end{definition}

For circuit families indexed by a single parameter, the above definition requires that a Turing machine can decide the direct connection language in time polylogarithmic in the size of the circuit and the index within the circuit family.
We remark that the second item above is not usually part of the standard definition.
We include it to make it easier to compose circuits uniformly.

In the preceding definitions, the first index is the count of the number of remaining indices.
In the circuits we design, we will extend this slightly, and allow the number of indices to be a function of the first index.
This will always be a very simple function, in all applications the number of indices will be linear in the first index, for example twice the first index plus two.
In the cases where the number of indices is just constant, for example families indexed by a single parameter, we do not use the first index as the count of the number of indices.
For example, a single indexed family of circuits is just written $\br{C_{n}}_{n}$, as opposed to $\br{C_{1, n}}_{n}$ as the above notation might suggest.
We also use vector notation as short hand for multi-indexed families.
For example, we use $\br{C_{\vn}}_{\vn}$ to denote a family indexed by $k, n_{1}, \dots, n_{k}$, and use $C_{\vn}$ to denote a specific circuit in the family.
In all cases, the exact number of indices will always be clear from context.

\begin{remark} \label{remark: gate naming convention}
  The definition of admissible encodings gives us a lot of freedom when picking how the gates have to be numbered.
  Other than the input and output gates, we are allowed to number the gates in any way, as long as the restriction on the length is respected.
  This freedom allows us to number our gates using more than just the integers $1$ through $s$: it makes it possible to interpret tuples of integers as gate names, as we explain now.
  In what follows, to avoid confusion, when we number gates by objects other than natural numbers, we will refer to the \emph{name} of a gate instead of its number.

  Here is for instance how we can design circuits where gate names are pairs of natural numbers.
  There exist efficient schemes to encode pairs of natural numbers as a single natural number, and these encoding schemes are efficiently invertible when an inverse exists.
  If we fix such a scheme, we can use pairs such as $(i, j)$ as names, with the understanding that the actual number of the gate is encoding of the pair $(i, j)$ as a single natural number.
  
  Care must be taken to respect the condition on the input and output numbering: we have to ensure that the encoding scheme is such that these gates are mapped to $0, \dots, m-1$ and $m, \dots, m+m'-1$.
  In all instances where we rely on such naming conventions, given the name of an input, resp.\ output, gate (by means of a pair, or tuple, of integers) we will always be able to decide that this is an input, resp.\ output, and determine its index in the allotted time; the same will hold for the converse direction.
  For gates other than inputs and outputs, we will use a scheme as described above to map a pair (or tuple) of integers to a single integer greater that $m+m'-1$ and conversely (recall that the number of inputs and outputs can be computed in polylogarithmic time).
\end{remark}
\begin{remark} \label{remark: find type of gate}
  Suppose we are working with a polylogtime-uniform circuit family.
  The gate name and the index of the circuit can be used to efficiently determine the type of a gate as follows.
  Given the index $(k, n_1, \ldots, n_k)$ of a circuit and the name $a$ of a gate in the circuit, we iterate over all gate types $b$ and check if the tuple $(k, n_1, \ldots, n_k, a, \epsilon, b)$ is in the direct connection language.
  Since there are a constant number of gate types, this allows us to compute the type of a gate, given its name.

  On input $(k, n_1, \ldots, n_k, a)$, we can also decide if a candidate gate name $a$ is valid in time polynomial in $\log(s) + \log(k) + \sum_{i=1}^k \log(n_i)$.
  We first compute an upper bound $s'$ on $s$; this can be done in the allotted time, by definition; then, we test if the length of $a$ is at most $c \log(s')$, for the constant $c$ from \cref{defintion: boolean admissible encoding}.
  Since $\log(s')$ itself is  polynomial in $\log(s) + \log(k) + \sum_{i=1}^k \log(n_i)$, this also fits in our time bound.
  If not, we reject.
  Else, we proceed as above, forming the strings $(k, n_1, \ldots, n_k, a, \epsilon, b)$ for all gate types $b$ and checking if they belong to the direct connection language.
  
  Similarly, we can always determine if a given gate is an input or output gate, since we can determine the number of inputs and outputs to a particular circuit.
\end{remark} 
\begin{remark} \label{remark: ordering input outputs}
  In our circuits, we will also carefully order the inputs and outputs to make other computations easier.
  For example, if the output of a circuit represents the entries of a matrix, we might choose to arrange the outputs in row major order.
  If the outputs of the circuits represent coefficients of a polynomial, we will fix a monomial ordering and order the outputs accordingly.
\end{remark}

We now define uniform families of arithmetic circuits.
We start with the definition of an arithmetic circuit.

\begin{definition} \label{definition: arithmetic circuit}
  Let $\bF$ be a field and let $\bF(\vec{x})$ be the field of rational functions in the variables $x_1, \ldots, x_n$ over $\bF$.
  An \emph{arithmetic circuit over $\bF$} is a directed acyclic graph.
  The vertices of this graph of in-degree zero are either called \emph{input gates} and are labeled by a variable $x_i$, or  called \emph{constant gates} and are labeled by a field element $\alpha \in \bF$.
  Vertices of positive in-degree are called \emph{internal gates} and are labeled by an operation from $\bc{+, \times, \div}$.
  Vertices of out-degree zero are called \emph{output gates}.
  Each gate of the circuit naturally computes an element of $\bF(\vec{x})$, assuming no division by zero takes place in the circuit, which we require.
  If $\bc{f_1, \ldots, f_m}$ are the functions computed by the output gates of the circuit, we say that the circuit \emph{computes} the functions $f_1, \ldots, f_m$.
  The \emph{size} of the circuit is the number of wires in the circuit.
  The \emph{depth} of the circuit is the length of the longest path from an input to an output of the circuit.

  If the only constants labeling the input gates of the circuit are $0, +1$, and $-1$, we say that the circuit is \emph{constant-free}.
  For a fixed variable $x_{i}$, if the subtree rooted at the denominator of each division gate in the circuit does not contain the input gate corresponding to $x_{i}$, then we say that the circuit is \emph{division-free with respect to $x_{i}$.}
  If $X'$ is a subset of variables and the circuit is division-free with respect to every variable in $X'$, we say the circuit is \emph{division-free with respect to $X'$.}
  If the circuit is division-free with respect to all variables, we say the circuit is \emph{weakly division-free}.
\end{definition}

When working with families of arithmetic circuits, uniformity also has to take into account the complexity of constructing the field elements used in the circuit (the elements $\alpha$ in \cref{definition: arithmetic circuit}).
Various definitions of uniformity that take the field elements into account are discussed in \cite{vonzurGathen86survey}.
In this work, we will restrict ourselves to circuits that are constant-free.
Any field element that we want to use within the circuit has to be constructed using $+1$ and $-1$.
This sidesteps the issue of having to define a notion of uniformity for the field elements used in the circuit family.
We can thus work with a notion of uniformity for arithmetic circuit families that closely mirrors the notion for boolean circuit families.

We will later want to simulate our arithmetic circuits using boolean threshold circuits.
We will define our notion of uniformity in a way that affords this simulation.
In the boolean case, the direct connection language only answered queries of the form \emph{is a gate of type \emph{b}?} and \emph{is gate \emph{b} a predecessor of gate \emph{a}?}.
For the sake of the simulation, we place a slightly more stringent requirement on the uniformity of arithmetic circuit families.
Given a gate $a$ with $k$ predecessors, we want to be able to answer queries of the form \emph{is gate \emph{b} the $i\ts{th}$ predecessor of gate \emph{a}?}.
For this, we of course require an ordering of all predecessors of each gate in the circuit.
This will be part of the definition of an admissible encoding of an arithmetic circuit.

\begin{definition}[Admissible encoding] \label{defintion: arithmetic admissible encoding}
  Let $\bF$ be a field and let $\mathcal{C} = (C_{k, n_1, \ldots, n_k})_{k, n_1, \ldots, n_k \in \bN}$ be a family of constant-free arithmetic circuits over $\bF$ of size $s \coloneqq s(k, n_1, \ldots, n_k)$.
  An \emph{admissible encoding of $\cC$} is a numbering of the gates of each circuit $C = C_{k, n_1, \ldots, n_k}$ in the family with the following properties.
  \begin{enumerate}[noitemsep]
    \item
      If $C$ has $m$ input gates, then they are numbered $0, \dots, m-1$.
    \item
      If $C$ has $m'$ output gates, then they are numbered $m, \dots, m + m' - 1$.
    \item
      There is a constant $c$, depending only on the circuit family, such that the binary representation of the number of any gate in $C$ has length at most $c \cdot \log\br{s}$.
  \end{enumerate}
  Further, for each gate $a \in C$, we fix a numbering of the set of immediate predecessors of $a$ in $C$.
  We require that this numbering of predecessors is contiguous, that is, it starts at $1$ and ends at the arity of $a$.
  We do not require the numbering of predecessors to respect the numbering of the gates themselves: if $b_{1}$ and $b_{2}$ are two predecessors of $a$, with $b_{1}$ preceding $b_2$ in the numbering of the gates of $C$, we do not require $b_{1}$ to precede $b_2$ in the numbering of the predecessors of $a$.
  For division gates, we require that the numerator occurs before the denominator in the numbering of the predecessors.
\end{definition}

\begin{definition} \label{defintion: direct connection language}
  Let $\bF$ be a field and let $\mathcal{C} = (C_{k, n_1, \ldots, n_k})_{k, n_1, \ldots, n_k \in \bN}$ be a family of constant-free arithmetic circuits over $\bF$ with an admissible encoding.
  Suppose a numbering of the gate types ($+, -, \times, \div, +1, 0, -1$) is fixed.
  The \emph{direct connection language of $\mathcal{C}$}, denoted by $L_{DC}(\mathcal{C})$, is the language consisting of binary encodings of tuples $(k, n_1, \ldots, n_k, a, p, b)$, where
  \begin{enumerate}[noitemsep]
    \item
      $a$ is the number of a gate in $C_{k, n_{1}, \dots, n_{k}}$.
    \item
      If $p \neq \epsilon$, then $b$ is the number of the $p\ts{th}$ predecessor gate in $C_{k, n_{1}, \dots, n_{k}}$ of $a$.
    \item
      If $p = \epsilon$, then $b$ encodes the gate type of $a$ (using the fixed numbering of gate types).
      \qedhere
  \end{enumerate}
\end{definition}

As in the boolean case, we measure the uniformity of a circuit family $\mathcal{C} = (C_{k, n_1, \ldots, n_k})_{k, n_1, \ldots, n_k \in \bN}$ by the complexity of its direct connection language $L_{DC}(\mathcal{C})$.

\begin{definition} \label{definition: uniform arithmetic circuit}
  Let $\bF$ be a field and let $\mathcal{C} = (C_{k, n_1, \ldots, n_k})_{k, n_1, \ldots, n_k \in \bN}$ be a family of constant-free arithmetic circuits of size $s \coloneqq s(k, n_1, \ldots, n_k)$ over $\bF$ with an admissible encoding.
  We say that the circuit family $\mathcal{C}$ is \emph{polylogtime-uniform} if 
  \begin{enumerate}[noitemsep]
    \item
      its direct connection language $L_{DC}(\mathcal{C})$ can be decided in time polynomial in $\log(s) + \log(k) + \sum_{i=1}^k \log(n_i)$, 
    \item
      the number of inputs and outputs of $C_{k,n_1, \ldots, n_k}$, and an upper bound on the size $s$ can be computed in time polynomial in $\log(s) + \log(k) + \sum_{i=1}^k \log(n_i)$, given $(k, n_1, \ldots, n_k)$ as input.
    \item
      the arity of a gate can be computed in time polynomial in $\log(s) + \log(k) + \sum_{i=1}^k \log(n_i)$, given $(k, n_1, \ldots, n_k)$ and the gate name as input.
  \end{enumerate}
  We say that the circuit family $\cC$ is \emph{logtime-uniform} if these tasks can be carried out in time linear in $\log(s) + \log(k) + \sum_{i=1}^k \log(n_i)$.
\end{definition}

\cref{remark: gate naming convention,remark: find type of gate,remark: ordering input outputs} apply verbatim in the arithmetic setting.

\subsection{Integer, rational, and finite field arithmetic} \label{subsection: finite field arithmetic}

In this section, we discuss how elements of various domains are represented by Turing machines and boolean circuits.
We also outline some basic integer and rational arithmetic operations that are in $\TC^{0}$.

Elements of $\bZ$ will be represented in base two with an additional sign bit.
For an integer $a \in \bZ$, the \emph{height} of $a$ is defined to be $\ceil{\log_{2}(\abs{a})}$, therefore, an integer of height $h$ is represented using $h+1$ bits (taking into account the sign bit).
The height of a polynomial with coefficients in $\bZ$ is defined to be the maximum of the heights of its coefficients.
In this representation, the following holds.

\begin{theorem} \label{theorem: integer ops in tc0}
  The following functions are in logtime-uniform $\TC^0$.
  \begin{enumerate}
    \item 
      \textsc{Iterated Integer Addition}: Given $n$ integers $a_1, \ldots, a_n$ as input, each consisting of at most $n$ bits, compute the sum $a_1 + \cdots + a_n$.
    \item 
      \textsc{Iterated Integer Multiplication}: Given $n$ integers $a_1$, \ldots, $a_n$ as input, each consisting of at most $n$ bits, compute the product $a_1 \cdots a_n$.
    \item
      \textsc{Integer Division}: Given integers $a$ and $b$ as input, compute $\lfloor a/b \rfloor$.
    \item 
      \textsc{Polynomial Division with Remainder}: given integers $a_0, \ldots, a_n$ and $b_0, \ldots, b_m$ as input, where $b_m \neq 0$ and $n \geq m$, compute $d \coloneqq b_m^{n-m+1}$ and integers $q_0, \ldots, q_{n-m}, r_0,\ldots,r_{m-1}$ such that for
      \begin{align*}
        f(x) &\coloneqq a_n x^n + \cdots + a_1 x + a_0 &
        q(x) &\coloneqq q_{n-m} x^{n-m} + \cdots + q_1 x + q_0 \\
        g(x) &\coloneqq b_m x^m + \cdots + b_1 x + b_0 &
        r(x) &\coloneqq r_{m-1} x^{m-1} + \cdots + r_1 x + r_0,
      \end{align*}
      we have $d \cdot f(x) = q(x) g(x) + r(x)$.
    \item
      \textsc{Iterated Polynomial Multiplication}: given integers $a_{1,0}, \ldots, a_{m,n}$ as input, compute integers $b_0, \ldots, b_{mn}$ such that
      \[
        b_{mn} x^{mn} + b_{mn-1} x^{mn - 1} + \cdots + b_1 x + b_0 = \prod_{i=1}^m (a_{i, n} x^n + a_{i, n-1} x^{n-1} + \cdots + a_{i, 1} x + a_{i, 0}).
      \]
  \end{enumerate}
\end{theorem}
\begin{proof}
  Item 1 is classical, see \cite{reif92}.
  Items 2 and 3 are the main result of \cite{HAB02}.
  Items 4 and 5 are from \cite[Corollary~6.5]{HAB02}.
  For all items above, the number of input and output bits are easily computable from the index of a circuit in the corresponding family.
\end{proof}

The operation performed in item 4 is \emph{pseudodivision}.
The polynomials $q$ and $r$ in the notation above are called the \emph{pseudoquotient} and \emph{pseudoremainder} of $f$ and $g$ respectively.

Multivariate polynomials will be represented in the dense representation.
In other words, the coefficient of every monomial up to a specified degree, including the ones that are zero, will be listed.
To multiply polynomials in $\bZ\bs{y_{1}, \dots, y_{k}}$, we use the Kronecker substitution to reduce to the univariate case.
If we want to multiply $m$ such polynomials of degree $d$, then we substitute $z^{\br{md+1}^{i}}$ for each $y_{i}$.
This gives us $m$ polynomials of degree $\br{md}^{O(k)}$, which we can multiply using the univariate multiplication circuit provided by \cref{theorem: integer ops in tc0}.
The coefficients of the $\vy$ monomials in the product can be read off from the coefficients of $z$ in the product after the above substitution.

An element of $q \in \bQ$ will be represented by a pair of integers $a, b$ with $b \neq 0$ such that $q = a / b$, where $a, b$ may have common factors.
The reason we allow $a, b$ to have common factors is that computing integer GCD is not known to be in $\TC^{0}$, therefore the circuits we are working with cannot convert a fraction to lowest terms.
The height of a pair of integers $a, b$ is simply the maximum of the heights of $a$ and $b$.

Iterated multiplication and addition of rational numbers in the above representation can be carried out by polylogtime-uniform $\TC^{0}$ circuits as follows.
For iterated multiplication, we separately multiply the numerators and denominators using the circuits for iterated integer multiplication (\cref{theorem: integer ops in tc0}).
This directly gives us a representation of the product.
For iterated addition, we first bring all elements to a common denominator.
This common denominator is the product of the denominators of the input, and so a representation of each input over this common denominator can be computed using the circuit for iterated integer multiplication from \cref{theorem: integer ops in tc0}.
Once we have this representation, we can add the numerators using the circuit for iterated integer addition, again from \cref{theorem: integer ops in tc0}.

By the primitive element theorem, every number field $\bK$ is of the form $\bQ\bs{\alpha} \cong \bQ\bs{z} / \ideal{g(z)}$, where $\alpha$ is an algebraic number with minimal polynomial $g(z)$.
Elements of a number field $\bK$ will therefore be represented by polynomials in $\bQ\bs{z}$ of degree less than $\deg g$.
Whenever we work with such a number field, we will assume that $g(z)$ is given as part of the input.

Elements of $\bF_{p}$, where $p$ is a prime, will be represented by an integer between $0$ and $p-1$ inclusive.
The \emph{height} of an element of $\bF_{p}$ is defined to be $\log\br{p}$.

The finite field $\bF_{p^{a}}$ is isomorphic to $\bF_{p}\bs{z} / \ideal{g(z)}$ for some irreducible polynomial $g(z) \in \bF_{p}\bs{z}$ of degree $a$.
Elements of the field $\bF_{p^{a}}$ will therefore be represented by polynomials in $\bF_{p}\bs{z}$ of degree less than $a$.
Whenever we work with such a finite field, we will assume that $g(z)$ is given as part of the input and that $g(z)$ is monic.
The \emph{height} of an element of $\bF_{p^{a}}$ is defined to be $a \log\br{p}$.
For each of the rings $R$ discussed above, the height of a polynomial with coefficients in $R$ is the maximum of the heights of the coefficients.

Our algorithms will occasionally require the fields we work with to have coefficient subfields that are larger than some bound $B$.
If our inputs lie in a field with coefficient subfield $\bF_{p^{a}}$ where $p^{a}$ is less than $B$, then we will instead pass to a field extension that is large enough and work over that field instead.
The following lemma shows that this can be done efficiently by a Turing machine.

\begin{lemma}
  \label{lemma: passing to bigger finite field}
  Given a finite field $\bF_{p^{a}}$ via an irreducible polynomial $g(y) \in \bF_p[y]$ of degree $a$, and a bound $B$, there is a Las Vegas algorithm that explicitly constructs a finite field $\bF_{p^{b}}$ such that $p^{b} \geq B$ and $\bF_{p^{a}}$ is a subfield of $\bF_{p^{b}}$.
  The algorithm also constructs an $\bF_{p}$-linear map $\phi$ (represented as a matrix in the powers of $y$ bases) that maps $\bF_{p^{a}}$ to an isomorphic subfield of $\bF_{p^{b}}$.
  The expected running time of the algorithm is polynomial in $a \log p$ and $\log |B|$.
\end{lemma}
\begin{proof}
  Recall that $\bF_{p^{n}}$ is a subfield of $\bF_{p^{m}}$ if and only if $m$ is a multiple of $n$.
  We therefore pick $b$ to be the first multiple of $a$ that is larger than $\log |B| / \log p$.
  To construct an irreducible polynomial of degree $b$, we simply sample a random polynomial of degree $b$ in $\bF_{p}\bs{y}$ and test for irreducibility.
  By \cite[Corollary~14.39]{vzGG13}, this can be performed using $\widetilde{O}\br{b^{3} + b^{2} \log p}$ operations in $\bF_{p}$ in expectation.
  Let this irreducible polynomial be $h(y)$.
  
  To construct the map between $\bF_{p^{a}}$ and $\bF_{p^{b}}$, we simply find a root of $g$ in $\bF_{p^{b}}$.
  That such a root exists follows from Fermat's little theorem and the structure of the factorization of polynomials in finite fields, see for example \cite[Theorem~14.2]{vzGG13}.
  By \cite[Corollary~14.16]{vzGG13}, finding this root requires $\widetilde{O}\br{ab\log{p}}$ operations in $\bF_{p^{b}}$ in expectation.
  If $\alpha$ is such a root, then the map $\phi: \bF_{p}\bs{y} / \ideal{g(y)} \to \bF_{p^{b}}$ sending $y$ to $\alpha$ is an isomorphism between $\bF_{p^{a}}$ and a subfield of $\bF_{p^{b}}$.
\end{proof}
We remark that the method above is far from the optimal way of constructing a compatible field extension.
However, it will suffice for our applications.

%% file: sections/uniformalgebraic.tex
In this section, we show that basic operations on arithmetic circuits can be carried out in a uniformity-preserving manner.
In particular, we show that polynomial interpolation---a standard tool used to construct low-depth arithmetic circuits---admits a uniform implementation.
As a consequence, we obtain uniform families of arithmetic circuits to compute the elementary symmetric polynomials and the inverse of a symbolic Vandermonde matrix.

The main technical tool we use these constructions is an explicit formula for the entries of the inverse of the Vandermonde matrix $\vand{1, \dots, n}$ in terms of Stirling numbers of the first kind.
Combined with an explicit formula for Stirling numbers, we obtain a circuit family that computes the entries of the inverse of $\vand{1, \dots, n}$.
We caution the reader that the following proof will be painstakingly detailed, only because it is the first proof that uses that the above notions of uniformity.

\begin{lemma} \label{lemma: special vand inverses are uniform}
  There exists a polylogtime-uniform family of constant-free, constant-depth circuits $\cC = (C_n)_{n \in \bN}$ over $\bQ$ where $C_{n}$ has no inputs, $n^2$ outputs, size $\poly(n)$ and computes the entries of the inverse of the Vandermonde matrix $\vand{1, \dots, n}$.
\end{lemma}
\begin{proof}
  An explicit formula for the entries of $\Vand^{-1}\br{1, \dots, n}$ was given by \textcite{macon1958inverses}, and the following statement is from \cite[Lemma~4]{eisinberg1998vandermonde}.
  Letting $V_{n}$ denote $\vand{1, \dots, n}$, we have
  \begin{equation} \label{equation: vandermonde inverse}
    \br{V_{n}^{-1}}_{j, i} = \br{-1}^{i+j} \sum_{k = \max\br{i, j}}^{n} \frac{1}{\br{k-1}!} \binom{k-1}{i-1} s(k, j),
  \end{equation}
  where $s(\cdot, \cdot)$ denotes the Stirling number of the first kind.
  These Stirling numbers can be computed using the formula
  \begin{equation} \label{equation: stirling number}
    s(a, b) = \sum_{j=a}^{2a - b} \binom{j-1}{b-1} \binom{2a - b}{j} \sum_{m=0}^{j-a} \frac{\br{-1}^{m+a-b} m^{j-b}}{m! \br{j - a - m}!},
  \end{equation}
  which follows from combining explicit formulas for Stirling numbers of the second kind with symmetric formulas relating the two kinds of Stirling numbers (see \cite[Equation~8.21]{charalmbides02}).

  We implement the circuit $C_n$ using these formulas.
  Our gate names will be tuples of numbers.
  \begin{enumerate}
    \item 
      We start with $2n$ copies of the constants $+1$ and $-1$.
      For each $i \in [2n]$, the gates with names $(0, i)$ and $(1, i)$ are constant gates labeled by the constants $+1$ and $-1$, respectively.
    
    \item
      Next, we compute $2n$ copies of each number from $1$ to $2n$.
      For each $i, j \in [2n]$, the gate with name $(2, i, j)$ is a $+$ gate whose inputs are the gates with names $(0, k)$ for $k \leq i$.
      The gate $(2,i,j)$ computes $i$.
      The $p\ts{th}$ predecessor of $(2, i, j)$ is $(0, p)$.
    
    \item
      We then compute factorials.
      For each $i \in [2n]$, the gate with name $(3, i)$ is a $\times$ gate whose inputs are the gates named $(1, j)$ for $j \leq i$.
      The gate $(3, i)$ computes $i!$.
      The $p\ts{th}$ predecessor of $(3, i)$ is $(1, p)$.
    
    \item
      We now compute binomial coefficients.
      For each $i \in [2n]$ and $j \leq i$, we have gates with names $(4, 1, i, j)$ and $(4, 2, i, j)$.
      The gate $(4, 1, i, j)$ is a $\times$ gate whose inputs are, in order, the gates $(3, j)$ and $(3, i-j)$.
      The gate $(4, 2, i, j)$ is a $\div$ gate whose numerator is $(3, i)$ and whose denominator is $(4, 1, i, j)$.
      Since $(3, i)$ computes $i!$ and $(4,1,i,j)$ computes $(i-j)! \cdot j!$, the gate $(4,2,i,j)$ computes $\binom{i}{j}$.
    
    \item
      Next, we compute numbers of the form $i^j$ as $i$ and $j$ vary over $[2n]$.
      In particular, for each $i, j \in [2n]$, the gate $(5, i, j)$ is a $\times$ gate with inputs $(2, i, k)$ for each $k \leq j$ (the predecessors will be numbered in this order also).
      In this way, the gate $(5, i, j)$ computes $i^{j}$.

    \item
      Now we compute the Stirling numbers using \cref{equation: stirling number}.
      For each $a, b \in [n]$ where $b \leq a$, we compute the $(j, m)$ term in the double sum, where $j \in \bc{a, \ldots, 2a-b}$ and $m \in \bc{0,\ldots,j-a}$, as follows.
      \begin{enumerate}
        \item
          The gate $(6, a, b, j, m, 1)$ is a $\times$ gate that has inputs $(3, m)$ and $(3, j-a-m)$.
          This gate computes the product of factorials $m! \cdot (j-a-m)!$.
        \item
          The gate $(6, a, b, j, m, 2)$ is a $\times$ gate whose first input is $(5, m, j-b)$ and whose second input is either $(0, 1)$ or $(1, 1)$, depending on the parity of $m+a-b$.
          This gate computes $\br{-1}^{m+a-b}m^{j-b}$.
        \item
          The gate $(6, a, b, j, m, 3)$ is a $\div$ gate with numerator $(6, a, b, j, m, 2)$ and denominator $(6, a, b, j, m, 1)$.
          This gate computes $\frac{(-1)^{m+a-b} m^{j-b}}{m! (j-a-m)!}$.
        \item
          Finally, the gate $(6, a, b, j, m, 4)$ is a $\times$ gate with inputs $(6, a, b, j, m, 3)$, $(4, 2, 2a-b, j)$, and $(4, 2, j-1, b-1)$.
          This gate computes $\binom{j-1}{b-1} \binom{2a-b}{j} \frac{(-1)^{m+a-b}}{m! (j-a-m)!}$, which is precisely the $(j, m)$ term appearing in the double sum used to compute the Stirling number $s(a,b)$.
      \end{enumerate}
      For each of the gates above, the predecessors are numbered in the order we have presented them.
      To compute the Stirling number $s(a,b)$, we add a $+$ gate $(7, a, b)$ whose inputs are $(6, a, b, j, m, 4)$ for all $j \in \bc{a, \ldots, 2a-b}$ and $m \in \bc{0,\ldots,j-a}$.
      The predecessors of $(7, a, b)$ are numbered in lexicographic order based on $(j, m)$.
      Given $a$, $b$, and $p$ in binary, computing the $p\ts{th}$ element in the lexicographic order of the pairs that satisfy the constraint can be done in polynomial time.

    \item
      Finally, we compute the entries of the inverse of the Vandermonde matrix using \cref{equation: vandermonde inverse}.
      For each $i, j \in [n]$, we compute the $k$\ts{th} term in the summation, where $k \in \bc{\max(i,j), \ldots, n}$, as follows.
      \begin{enumerate}
        \item
          The gate $(8, i, j, k, 1)$ is a $\div$ gate with numerator $(0, 1)$ and denominator $(3, k-1)$.
          This gate computes $\frac{1}{(k-1)!}$.
        \item
          The gate $(8, i, j, k, 2)$ is a $\times$ gate whose inputs are $(8, i, j, k, 1)$, $(4, 2, k-1, i-1)$, and $(7, k, j)$.
          This gate computes $\frac{1}{(k-1)!} \binom{k-1}{i-1} s(k,j)$.
      \end{enumerate}
      To compute the $(i,j)$ entry of $V_n^{-1}$, we add a $+$ gate $(9, i, j)$ whose inputs are $(8, i, j, k, 2)$ for all $k \in \bc{\max(i,j), \ldots, n}$.
      We add a final $\times$ gate $(10, i, j)$ with inputs $(9, i, j)$ and either $(0,1)$ or $(1,1)$, depending on the parity of $i+j$.
      By \cref{equation: vandermonde inverse}, the gate $(10, i, j)$ correctly computes the $(j,i)$ entry of $V_n^{-1}$.

      Finally, for each of the gates above, its predecessors are numbered in the order we presented them.
  \end{enumerate}

  From the preceding description of the circuit computing $V_n^{-1}$, it is clear that the size of the circuit is polynomial in $n$ and the depth is bounded by a universal constant.
  The above description also serves to bound the uniformity of the circuit family.
  There are a fixed number of rules that determine the type of a gate and its predecessors.
  These directly translate into a construction of a Turing machine that can, in polynomial time in $\log(n)$, determine the type of a gate or decide if one gate is a predecessor of another.
  The only gates with more than a constant number of predecessors are those computing various constants and the gates with names $(7, i, j)$ and $(9, i, j)$.
  For each of these gates, the circuit description shows how to compute the $p\ts{th}$ predecessor of $a$ in time polynomial in the binary length of the gate names and the index of the circuit.
  Therefore we can compute the $p^{th}$ predecessor of $a$ and check if this matches with $b$ to decide the direct connection language.
  The circuit construction explicitly states the arity of every gate, so this can also be computed efficiently.
  Finally, the size of the circuit is also apparent from the above construction, therefore the required upper bound on the size can also be computed efficiently.
\end{proof}

Uniformity of interpolation will follow from the above lemma.
We start by showing how the coefficients of a single variable can be interpolated, and then handle the general case.

\begin{lemma} \label{lemma: polynomial interpolation uniform}
  Let $\cC = (C_{\vn})_{\vn}$ be a polylogtime-uniform family of arithmetic circuits over $\bQ$ of size $s_{\vn}$ and depth bounded by a constant $\Delta$, and with one output.
  Let $f_{\vn}$ be the rational function computed by $C_{\vn}$, 
  let $y$ be a distinguished variable, and suppose that every $C_{\vn}$ is division-free with respect to $y$, so that $f_{\vn} \in \bQ(\vx)\bs{y}$.

  Let $d_{\vn}$ be an upper bound on the degree of $f_{\vn}$ in $y$, and let $f_{\vn,0}, \ldots, f_{\vn,d_{\vn}} \in \bQ(\vx)$ be such that
  \[
    f_{\vn}\br{\vx, y} = \sum_{i=0}^{d_{\vn}} f_{\vn, i}\br{\vx} y^{i}.
  \]
  Finally, suppose that a binary representation of $d_{\vn}$ can be computed from a binary representation of $\vn$ in time polynomial in the binary length of $\vn$ and $\log \br{d_{\vn}}$.
  Then there exists a polylogtime-uniform family $\cD = (D_{\vn})_{\vn}$ of arithmetic circuits over $\bQ$ of size $\poly(s_{\vn}, d_{\vn})$ and depth at most $\Delta + O(1)$ such that $D_{\vn}$ computes $f_{\vn, 0}, \dots, f_{\vn, d_{\vn}}$.
\end{lemma}
\begin{proof}
  Fix $C = C_{\vn}$, $f = f_{\vn}$, and $d = d_{\vn}$.
  We describe the circuit $D = D_{\vn}$.

  The circuit $D$ starts by computing the numbers $1, \dots, d+1$.
  As in the proof of \cref{lemma: special vand inverses are uniform}, we have a set of gates labeled by the constant $1$, and we use sum gates to compute $1, \ldots, d+1$ from these constants.
  The constant gates are named $(1, i)$ for $i \in [d+1]$, and the gates computing $1, \ldots, d+1$ are named $(2, i)$ for $i \in [d+1]$.
  For these sum gates, the $p\ts{th}$ predecessor of $(2,i)$ is $(1, p)$, as long as $p \leq i$.

  The circuit $D$ then contains $d+1$ copies of $C$.
  In the $i\ts{th}$ copy, the input gate $y$ is converted to sum gate with a single input, namely the gate $(2, i)$ that computes the number $i$.
  The other input gates of $C$ are also converted to sum gates whose only input is the input gate of $D$ labeled by the same variable.
  Therefore, these copies of $C$ compute the partial evaluations $f(\vx, 1), \dots, f(\vx, d+1)$.
  The condition that $y$ is not in any denominator of $C$ ensures that the circuit does not divide by zero.
  The gates in these copies of $C$ are named as follows: if $v$ is the name of a gate in $C$, then the corresponding gate in the $i\ts{th}$ copy of $C$ in $D$ is named $(3, i, v)$.

  The circuit $D$ also has a copy of the circuit constructed in \cref{lemma: special vand inverses are uniform} that computes the inverse of the Vandermonde matrix $\vand{1, \dots, d+1}$.
  We denote this circuit by $V$.
  For each gate named $v$ in $V$, its corresponding copy in $D$ is named $(4, v)$.

  The polynomial $f_{\vn, i}(\vx)$ is the inner product of a column of the inverse of the Vandermonde matrix, and the outputs of the copies of $C$.
  These are computed using $d+1$ product gates and one sum gate each.
  For each $i, j \in [d+1]$, the gate $(5, i, j)$ is a $\times$ gate connected to the output of the $j\ts{th}$ copy of $C$ and the $(j, i)$ entry of the inverse of the Vandermonde matrix, and its predecessors are numbered in this order.
  For each $i \in [d+1]$, the gate $(6, i)$ is a $+$ gate with inputs $(5, i, j)$ for all $j \in [d+1]$.
  The $p\ts{th}$ predecessor of $(6, i)$ is $(5, i, p)$.
  The gate $(6, i)$ computes the polynomial $f_{\vn, i}(\vx)$.

  The claims on the size and depth of the circuit family $\cD$ follow from the above description and the corresponding bounds on the subcircuits of $D$.
  We now argue that $\cD$ is a polylogtime-uniform family by describing how the direct connection language is decided.

  Let $T_{C}$ and $T_{V}$ denote the Turing machines that decide the direct connection languages of $C$ and $V$, respectively.
  We will design a Turing machine $T_D$ that decides the direct connection language of $D$.
  Let $(\vn, a, p, b)$ be an input to $T_D$.
  From the first entry in the gate name $a$, we can determine where in $D$ the gate lies: whether it is part of computing the constants, a gate within a copy of $C$, part of the circuit that computes the inverse of the Vandermonde matrix, or part of the inner product computation.
  If $a$ is part of the subcircuit that computes the constants, then deciding if the input is a YES instance is straightforward.

  Suppose $a$ is part of the $i\ts{th}$ copy of $C$.
  From $a$, we can extract the name $v_a$ of the gate within $C$ of which $a$ is a copy.
  The machine ensures $i \leq d+1$ (recall that $d=d_\vn$ can be computed in polynomial time in the binary length of $\vn$ and $d_\vn$), and rejects otherwise.
  If $v_{a}$ corresponds to the input variable $y$, then $T_{D}$ accepts if $p = \epsilon$ and $b$ denotes the type $+$, or if $p = 1$ and $b$ is the gate computing the constant $i$; all other inputs are rejected.
  If $v_{a}$ corresponds to any other input gate, then similarly we accept if either $p = \epsilon$ and $b$ denotes $+$, or if $p = 1$ and $b$ is the input gate in $D$ labeled by the same variable.
  If $v_{a}$ is not an input gate and if $p = \epsilon$, then the machine $T_D$ simulates $T_{C}$ on the input $(\vn, v_{a}, p, b)$ and accepts or rejects accordingly.
  If $p \neq \epsilon$, then $T_D$ checks if $b$ is also a gate in the $i\ts{th}$ copy of $C$.
  If not, the machine rejects.
  Otherwise, the machine $T_D$ simulates $T_{C}$ on $(\vn, v_{a}, p, v_{b})$, where $v_b$ is the name of the gate in $C$ corresponding to $b$.

  If $a$ is instead part of the circuit $V$, we instead will simulate $T_V$ on an appropriate input to decide if $(\vn, a, p, b)$ is in the direct connection language of $D$.
  This again requires extracting the name of the gate within the circuit $V$ from $a$, and potentially doing the same for $b$ if $p \neq \epsilon$.
  The requirement that the binary representation of $d_{\vn}$ can be written down in time polynomial in the binary representation of $\vn$ ensures that the input to $T_{V}$ can be computed by $T_{D}$ in the allotted time.

  Finally, if $a$ is part of the inner product subcircuit, then it is either connected to the outputs of the copies of $C$ and $V$ (which is the case for gates labeled $(5, i, j)$) or to other gates in the inner product subcircuit (which is the case for gates labeled $(6, i)$).
  In the first case, to decide if $b$ is a predecessor of $a$, we can check if $b$ is an output gate of the correct subcircuit.
  If $a = (5, i, j)$ then based on whether $p = 1$ or $2$, $b$ is either the output of the $j\ts{th}$ copy of $C$ or the $(j, i)$ entry of the inverse of the Vandermonde matrix.
  The latter case is handled similarly.
  The above description explicitly states the arity of every gate, so this can likewise be computed efficiently (with oracle calls to decide the arity of gates in $C$ if required).
  If $C_{\vn}$ has $m$ inputs then $D_{\vn}$ has $m-1$ inputs, and $D_{\vn}$ always has $d_{\vn} + 1$ outputs.
  Finally, a bound on the size of $D_{\vn}$ follows from a bound on the size of $C_{\vn}$, the parameter $d_{\vn}$, and the size of the circuit for the inverse of the Vandermonde matrix, therefore such a bound can easily be computed.
\end{proof}

Our proofs of uniformity will often be similar to the one above: we will construct our circuits using previously-constructed uniform subcircuits, along with some other simple machinery.
In all cases, the direct connection language will be decided by simulating the Turing machines that decide the direct connection languages of the subcircuits.
Therefore, the remaining proofs will not be as low-level as the two proofs of uniformity above.

Uniformity of multivariate interpolation follows from \cref{lemma: polynomial interpolation uniform} using Kronecker substitution.
\begin{lemma} \label{lemma: multivariate polynomial interpolation uniform}
  Let $\cC = (C_{\vn})_{\vn}$ be a polylogtime-uniform family of arithmetic circuits over $\bQ$ of size $s_{\vn}$ and depth bounded by a constant $\Delta$, and with one output.
  Let $f_{\vn}$ be the rational function computed by $C_{\vn}$, let $\vy$ be a distinguished set of $m_{\vn}$ variables, and suppose that every $C_{\vn}$ is division-free with respect to $\vy$, so that $f_{\vn} \in \bQ(\vx)\bs{\vy}$.

  Let $d_{\vn}$ be an upper bound on the degree of $f_{\vn}$ in $\vy$, and for each $\valpha \in \bN^{m_{\vn}}$ with $\abs{\valpha} \leq d_{\vn}$, let $f_{\vn, \valpha} \in \bQ(\vx)$ be the coefficient of $\vy^{\valpha}$ in $f_{\vn}$.
  Finally, suppose that binary representations of $d_{\vn}$ and $m_\vn$ can be computed from a binary representation of $\vn$ in time polynomial in the binary length of $\vn$ and $\log \br{d_{\vn}}$.

  Then there exists a polylogtime-uniform family $\cD = (D_{\vn})_{\vn}$ of arithmetic circuits over $\bQ$ of size $\poly\br{s_{\vn}, \br{d_{\vn} + 1}^{m_{\vn}}}$ and depth at most $\Delta + O(1)$ such that $D_{\vn}$ computes $f_{\vn, \valpha}$ for all $\valpha \in \bN^{m_{\vn}}$ with $\abs{\valpha} \leq d_{\vn}$.
\end{lemma}
\begin{proof}
  We use Kronecker substitution.
  Define $g_{\vn}$ as
  \[
    g_{\vn}(\vx, z) \coloneqq f_{\vn}\br{\vx, z, z^{d_{\vn}+1}, z^{\br{d_{\vn}+1}^{2}}, \dots, z^{\br{d_{\vn}+1}^{m_{\vn}-1}}}.
  \]
  The polynomial $g_{\vn}$ has degree at most $d_{\vn} \cdot \br{d_{\vn}+1}^{m_{\vn}}$ and can be computed by a constant-free circuit of size $\poly\br{s_\vn, \br{d_{\vn} + 1}^{m_{\vn}}}$ and depth $\Delta + O(1)$ as follows.
  We construct a circuit with inputs $\vx, z$.
  Using $+$ gates to copy $z$ and $\times$ gates, we can create gates that compute the powers $z^{\br{d_{\vn}+1}^{i}}$ for all $i \in [m_{\vn}-1]$.
  The numbering of the predecessor gates will be the natural one, similar to how we computed powers and factorials in previous proofs.

  We then have a copy of the circuit for $f_{\vn}$, with all input gates changed to $+$ gates and wired to the inputs $\vx$ and these powers of $z$.
  The gates are named as in previous proofs: the names indicate what power of $z$ the gate is involved in computing, or that the gate is part of the circuit for $f_{\vn}$, in which case the name will contain the name of the corresponding gate within the circuit for $f_{\vn}$.
  This circuit family that computes $\bc{g_{\vn}}_{\vn}$ is polylogtime-uniform.
  The direct connection language can be decided by simulating the Turing machine for the circuit family $\cC$ for gates in the subcircuit that computes $f_{\vn}$.

  We now apply \cref{lemma: polynomial interpolation uniform} to this circuit family (the family is division-free with respect to $z$, and we can  compute the binary representation of $d_{\vn} \cdot \br{d_{\vn}+1}^{m_{\vn}}$ in time polynomial in the binary length of $\vn$ and $m_\vn \log \br{d_{\vn}}$).
  This gives us a circuit family that computes the coefficients of $g_{\vn}$ as a polynomial in $z$, which are in bijection with the coefficients $f_{\vn, \valpha}$.
  The gate names in the circuit computing the coefficients of $g_{\vn}$ encode the power of $z$ whose coefficient is being computed.
  Given this, it is easy to compute the monomial $\valpha$ whose coefficient is being computed at each gate: this just involves changing a number to base $d_{\vn} + 1$ and listing the digits.
\end{proof}

\cref{lemma: multivariate polynomial interpolation uniform} allows us to extract the coefficient of every monomial in the distinguished variables.
For some applications, we will need to obtain the coefficients of a given polynomial when treated as a univariate polynomial in many different distinguished variables $y$.
The following lemma shows that this task can be performed uniformly.

\begin{lemma} \label{lemma: one at a time polynomial interpolation uniform}
  Let $\cC = (C_{\vn})_{\vn}$ be a polylogtime-uniform family of arithmetic circuits over $\bQ$ of size $s_{\vn}$ and depth bounded by a constant $\Delta$.
  Let $f^{(k)}_{\vn}$ be the rational functions computed by $C_{\vn}$, for $k \in [m'_\vn]$, let $\vy$ be a distinguished set of $m_{\vn}$ variables, and suppose that every $C_{\vn}$ is division-free with respect to $\vy$, so that each $f^{(k)}_{\vn}$ is in $\bQ(\vx)\bs{\vy}$.

  Let $d_{\vn}$ be an upper bound on the degree of all $f^{(k)}_{\vn}$ in $\vy$, and for each $0 \leq i \leq d_{\vn}$,  $j \in [m_{\vn}]$ and $k \in [m'_\vn]$, let $f^{(k)}_{\vn, j, i} \in \bQ(\vx)\bs{y_{1}, \dots, y_{j-1}, y_{j+1}, \dots, y_{m_{\vn}}}$ be such that
  \[
    f^{(k)}_{\vn} = \sum_{i=0}^{d_{\vn}} f^{(k)}_{\vn, j, i} y_{j}^{i}.
  \]
  Finally, suppose that binary representations of $d_{\vn}$ and $m_\vn$ can be computed from a binary representation of $\vn$ in time polynomial in the binary length of $\vn$ and $\log \br{d_{\vn}}$.

  Then there exists a polylogtime-uniform family $\cD = (D_{\vn})_{\vn}$ of arithmetic circuits over $\bQ$ of size $\poly\br{s_{\vn}, m_\vn, d_{\vn}}$ and depth $\Delta + O(1)$ such that $D_{\vn}$ computes $f^{(k)}_{\vn, j, i}$ for all $0 \leq i \leq d_{\vn}$, $j \in [m_{\vn}]$ and $k \in [m'_\vn]$.
\end{lemma}

\begin{proof}
  This proof is similar to that of \cref{lemma: polynomial interpolation uniform}.
  Fix $C = C_{\vn}$, $f = f_{\vn}$, and $d = d_{\vn}$.
  We describe the circuit $D = D_{\vn}$.

  The circuit $D$ has set of gates named $(1, i)$ that are constant gates labeled by the constant $1$, and a set of $+$ gates $(2, i)$ for $i \in [d]$ that compute the numbers $1, \dots, d+1$.
  The predecessors are numbered the obvious way.
  The circuit $D$ then has $m_{\vn} \br{d_{\vn} + 1}$ copies of $C$, one for every pair $(i, j)$ with $0 \leq i \leq d_{\vn}$ and $1 \leq j \leq m_{\vn}$.
  The gates in the $(i,j)$\ts{th} copy are named $(3, i, j, v)$, where $v$ is the name of the corresponding gate in $C$.
  In the $(i,j)$\ts{th} copy of $C$, the input gates labeled by $y_j$ are replaced by an arity-one summation gate whose child is the gate $(2,i)$ computing the integer $i$.

  The circuit $D$ then has a copy of the circuit $V$ from \cref{lemma: special vand inverses are uniform} that computes the inverse of the Vandermonde matrix $\vand{1, \dots, d_{\vn}+1}$.
  These gates are named $(4, v)$, where $v$ is the name of the corresponding gate in $V$.
  Following this, the circuit $D$ has gates $(5, k, *)$ and $(6, k, *)$ that compute all required products of the evaluations and the matrix inverse.
  Again the predecessors will be numbered just as in the proof of \cref{lemma: polynomial interpolation uniform}.

  The claimed bounds on the size and depth of $D$ follow from the description above.
  It remains to show how the direct connection language of $D$ is decided.
  Let $(\vn, a, p, b)$ be an input to the machine deciding the direct connection language.
  If $a$ is of the form $(t, *)$ with $t \neq 3$, then deciding if this is a YES instance can be done the exact same way as in the proof of \cref{lemma: polynomial interpolation uniform}.
  If $a$ is of the form $(3, i, j, v)$, and if $v$ is an input gate within $C$, we can use the indices $i, j$ to ensure that the input is wired to the correct constant or variable (we first check that $0\le i \le d_\vn$ and $1 \le j \le m_\vn$, which takes time polynomial in the binary length of $\vn$, by assumption).
  In particular, if the input is $x_{t}$ with $t \neq j$, then the gate must be wired to $x_{t}$ itself, otherwise if $t = j$ then the input is wired to the gate computing the constant $i$.
  The rest of the computation is also the exact same as in the proof of \cref{lemma: polynomial interpolation uniform}.
\end{proof}

A corollary of the above interpolation results is a uniform family of circuits that compute the elementary symmetric polynomials.
\begin{lemma} \label{lemma: esym uniform}
  There exists a polylogtime-uniform family of constant-free, weakly division-free, constant-depth arithmetic circuits $\cC = (C_n)_n$ over $\bQ$ such that $C_{n}$ has size $\poly(n)$ and computes the elementary symmetric polynomials $e_{1}, \dots, e_{n}$ on $n$ variables.
\end{lemma}

\begin{proof}
  We use the constant-depth circuits for the elementary symmetric polynomials designed by Ben-Or.
  There exists a polylogtime-uniform family $\cD = \br{D_{n}}_{n}$ of constant-free, polynomial-size, constant-depth circuits over $\bQ$ such that $D_{n}$ has $n+1$ inputs $x_{1}, \dots, x_{n}, t$, and computes $\prod_{i=1}^{n} \br{x_{i} + t}$.
  Indeed $D_{n}$ can just compute this polynomial using $n$ sum gates and a single product gate.
  The degree of $D_{n}$ is exactly $n$, and $\cD$ is division-free.
  Therefore, $\cD$ satisfies all the assumptions in \cref{lemma: polynomial interpolation uniform}, and we can invoke \cref{lemma: polynomial interpolation uniform} to obtain $\cC$.
\end{proof}

Using \cref{lemma: esym uniform}, we can also construct a uniform family of circuits that compute the inverse of a generic Vandermonde matrix.

\begin{lemma} \label{lemma: vandermonde inverse uniform}
  There exists a polylogtime-uniform family of constant-free, constant-depth arithmetic circuits $\cC = (C_n)_n$ over $\bQ$ such that $C_{n}$ has size $\poly(n)$ and computes the entries of the inverse of the matrix $\vand{x_{1}, \dots, x_{n}}$.
\end{lemma}
\begin{proof}
  The circuit uses the following two facts: the determinant of $V \coloneqq \vand{x_{1}, \dots, x_{n}}$ is $\prod_{i < j} \br{x_{i} - x_{j}}$, and the entries of the inverse of this matrix can be written as
  \[
    V^{-1}_{i, j} = \frac{\br{-1}^{n-i} e_{n-i}\br{x_{1}, \dots, x_{j-1}, x_{j+1}, \dots, x_{n}}}{\prod_{k \neq j} \br{x_{j} - x_{k}}}.
  \]
  We now describe the circuit $C_{n}$.
  This circuit contains $n$ copies of the circuit for the elementary symmetric polynomials in $n-1$ variables constructed in \cref{lemma: esym uniform}.
  To the $j\ts{th}$ copy, the inputs provided are $x_{1}, \dots, x_{j-1}, x_{j+1}, \dots, x_{n}$.
  The circuit also has $n$ subcircuits where the $j^{th}$ one computes $\prod_{k \neq j} \br{x_{j} - x_{k}}$.
  Finally, the circuit combines these to compute the entries of the inverse.
  The claims on the size and depth are straightforward given the above description.
  For uniformity, the argument is essentially the same as the one in the proof of \cref{lemma: polynomial interpolation uniform}: the Turing machine that has to decide the direct connection language of $\cC$ can simulate the Turing machines that decide the direct connection languages of the constituent subcircuits.
  Note that the output gate that computes the entry at position $(i, j)$ in the inverse will have $(i, j)$ encoded in its name.
\end{proof}

%% file: sections/resultantbasics.tex
This section discusses the theory of resultants.
We first introduce the multivariate resultant, which is a generalization of the usual resultant to multiple polynomials in several variables that tests when a system of homogeneous polynomial equations has a solution.
We will then discuss how the multivariate resultant can be used to test when systems of inhomogeneous equations have solutions.
Finally, we will discuss how the multivariate resultant can be used to count the number of solutions in zero-dimensional systems.

For a friendly introduction to the multivariate resultant, we refer the reader to \textcite[Chapter 3]{CLO05}.
A comprehensive treatment of resultants can be found in the book of \textcite{GKZ94}, in the survey of \textcite{IK93}, or in the work of \textcite{Jouanolou91}.

\subsection{Multivariate resultants} \label{subsection: multivariate resultant}

The resultant of two univariate polynomials $f, g$ is a polynomial in the coefficients of $f$ and $g$ that vanishes exactly when $f$ and $g$ have a common solution in the algebraic closure $\overline{\bF}$ of the base field $\bF$.
The multivariate resultant generalizes this, and detects when a homogeneous system of equations has projective solutions.

Throughout this section, we fix a natural number $n \in \bN$ and a choice of natural numbers $d_0, \ldots, d_n \in \bN$.
For $d \in \bN$, we write $M_{n, d}$ for the set of homogeneous degree-$d$ monomials in the $n+1$ variables $x_{0}, \dots, x_{n}$.
We take $\vec{u} = \setbuild{u_{i, \valpha}}{i \in \bc{0, \ldots, n}, \valpha \in M_{n, d}}$ to be a set of variables corresponding to the coefficients of $n+1$ homogeneous polynomials of degrees $d_0, \ldots, d_n$, respectively.

With this notation in hand, we now define the multivariate resultant.

\begin{definition}[{see, e.g., \cite[Proposition 2.3]{Jouanolou91}}] \label{definition: multivariate resultant all fields}
  Let $\bF$ be any field.
  The \emph{resultant} $\res_{d_0, \ldots, d_n} \in \bF[\vec{u}]$ is the unique polynomial satisfying the following conditions.
  \begin{enumerate}
    \item
      If $F_0, \ldots, F_n \in \bF\bs{x_0, \ldots, x_n}$ are homogeneous polynomials of degrees $d_0, \ldots, d_n$, respectively, then $\res_{d_{0}, \dots, d_{n}}(F_0, \ldots, F_n) = 0$ if and only if the system of equations $F_0(\vec{x}) = \cdots = F_n(\vec{x}) = 0$ has a solution in $\bP^{n}_{\overline{\bF}}$.
      By $\res_{d_0, \ldots, d_n}(F_0, \ldots, F_n)$ we mean the evaluation of $\res_{d_0, \ldots, d_n}$ obtained by setting $u_{i, \valpha}$ equal to the coefficient of the monomial $\vec{x}^{\valpha}$ in $F_i$.
    \item
      $\res_{d_0, \ldots, d_n}(x_0^{d_0}, \ldots, x_n^{d_n}) = 1$.
    \item
      The polynomial $\res_{d_0, \ldots, d_n}$ is irreducible in $\overline{\bF}[\vu]$.
  \end{enumerate}
  Moreover, if we write $\res_{d_0, \ldots, d_n}^{\bQ}$ and $\res_{d_0, \ldots, d_n}^{\bF}$ for the resultants over $\bQ$ and $\bF$, respectively, then the former has integer coefficients and $\res_{d_0, \ldots, d_n}^{\bF}$ is the image of $\res_{d_0, \ldots, d_n}^{\bQ}$ under the natural ring homomorphism $\bZ\bs{\vec{u}} \to \bF\bs{\vec{u}}$.
\end{definition}

When $n$ and the degrees $d_{0}, \dots, d_{n}$ are clear from context, we will drop the subscript and use $\res$ to denote the resultant $\res_{d_{0}, \dots, d_{n}}$.
The existence of a polynomial satisfying \cref{definition: multivariate resultant all fields} is not obvious.
We refer the reader interested in the existence of the resultant to \textcite{Jouanolou91}.
The final statement in \cref{definition: multivariate resultant all fields}, the fact that the resultant over any field is the image of the resultant over the integers, will be crucial in our proofs.

The multivariate resultant generalizes the determinant.
Consider a collection of homogeneous linear forms $L_0, \ldots, L_n \in \bF\bs{x_0, \ldots, x_n}$ given by $L_i(\vec{x}) = \sum_{j=0}^n a_{i,j} x_j$.
The system of equations $L_0(\vec{x}) = \cdots = L_n(\vec{x}) = 0$ has a nonzero solution exactly when
\[
  \det
  \begin{pmatrix}
    a_{0, 0} & a_{0, 1} & \cdots & a_{0, n} \\
    a_{1, 0} & a_{1, 1} & \cdots & a_{1, n} \\
    \vdots & \vdots & \ddots & \vdots \\
    a_{n, 0} & a_{n, 1} & \cdots & a_{n, n}
  \end{pmatrix}
  = 0.
\]
Likewise, from \cref{definition: multivariate resultant all fields}, we know that this system has a nonzero solution exactly when $\res_{1, \ldots, 1}(L_0, \ldots, L_n) = 0$.
This is no coincidence: the polynomials $\res_{1, \ldots, 1}(\vec{u})$ and $\det_{n+1}(\vec{u})$ are the same!

It is clear from the definition of the resultant that it is a polynomial in $\sum_{i=0}^n \binom{n + d_i}{d_i}$ variables.
Less obvious is its degree, which is provided by the following lemma.

\begin{lemma}[{see, e.g., \cite[Chapter 3, Theorem 3.1]{CLO05}}] \label{lemma: resultant degree}
  The resultant $\res_{d_0, \ldots, d_n}(\vec{u})$ is homogeneous of degree $d_0 \cdots d_{i-1} d_{i+1} \cdots d_n$ with respect to the variables $\setbuild{u_{i, \valpha}}{\valpha \in M_{n, d_i}}$ and is homogeneous of total degree $\sum_{i=0}^n d_0 \cdots d_{i-1} d_{i+1} \cdots d_n$.
\end{lemma}

Our work will focus on the design of uniform circuits of constant depth and size $d^{\poly\br{n}}$ that compute the resultant.
To design such circuits, we will use non-determinantal formulations of the resultant, since determinants provably require super-polynomial size to compute using circuits of bounded depth \cite{LST21,Forbes24} and are conjectured to not be computable by quasipolynomial-size circuits of bounded depth.
In particular, we will make use of the following identity, known as the \emph{Poisson formula}, that expresses the resultant $\res_{d_0, \ldots, d_n}$ as a product of two terms: one is the smaller resultant $\res_{d_0, \ldots, d_{n-1}}$, and the other is related to the values attained by $F_n$ on the common zeroes of $F_0, \ldots, F_{n-1}$.
The special case of the Poisson formula corresponding to $n = 1$ was used to design constant-depth circuits for the bivariate resultant \cite{AW24,BKRRSS25b}; its multivariate generalization was already used by \textcite{JS07}, and will likewise be key in designing low-depth circuits for the multivariate resultant.

\begin{theorem}[see, e.g., {\cite[Chapter 3, Section 3, Exercise 8]{CLO05}}] \label{theorem: multivariate poisson formula}
  Let $F_0, \ldots, F_n \in \bF[\vec{x}]$ be homogeneous polynomials of degrees $d_0, \ldots, d_n$, respectively.
  For $i \in \bc{0, 1, \ldots, n}$, let
  \begin{align*}
    \overline{F}_i(\vec{x}) & \coloneqq F_i(x_0, \ldots, x_{n-1}, 0) \\
    f_i(\vec{x}) & \coloneqq F_i(x_0, \ldots, x_{n-1}, 1).
  \end{align*}
  Suppose $\res_{d_0, \ldots, d_{n-1}}(\overline{F}_0, \ldots, \overline{F}_{n-1}) \neq 0$.
  Then the resultant $\res_{d_0, \ldots, d_n}(F_0, \ldots, F_n)$ satisfies the identity
  \[
    \res_{d_0, \ldots, d_n}(F_0, \ldots, F_n) = \res_{d_0, \ldots, d_{n-1}}(\overline{F}_0, \ldots, \overline{F}_{n-1})^{d_n} \cdot \prod_{\valpha \in \Var(f_0, \ldots, f_{n-1})} f_n(\valpha)^{m(\valpha)},
  \]
  where $\Var(f_0, \ldots, f_{n-1}) \subseteq \overline{\bF}^n$ is the finite set of common zeroes of $f_0, \ldots, f_{n-1}$ in the algebraic closure $\overline{\bF}^n$ and $m(\valpha) \in \bN$ is the multiplicity of $\valpha$ in $\Var(f_0, \ldots, f_{n-1})$.
\end{theorem}

\subsection{Satisfiability of affine systems} \label{subsection: hn to resultant}

The resultant naturally suggests an algorithm to decide the satisfiability (in the algebraic closure of the coefficient field) of homogeneous systems of $n+1$ equations in $n+1$ variables, which are known as \emph{square} systems.
Given equations $F_0(\vx) = F_1(\vx) = \cdots = F_n(\vx) = 0$, we simply compute the resultant $\res(F_0,\ldots,F_n)$ and report that this system is satisfiable if and only if $\res(F_0, \ldots, F_n) = 0$.
To solve a larger system $F_0 = \cdots = F_m = 0$ of homogeneous equations, one can reduce to the square case by forming $n+1$ random linear combinations of the $F_i$.
With good probability, the resulting square system will be equisatisfiable with the original system, so algorithms that compute the resultant allow us to decide the satisfiability of homogeneous systems of equations of any size.
What about systems of inhomogeneous equations?

A natural strategy to solve a system of inhomogeneous equations is to homogenize the system and attempt to use the resultant.
This strategy works for some, but not all, affine systems, depending on the behaviour at infinity of the homogenized equations.
A classical method for using the multivariate resultant to study affine systems is the method of generalized characteristic polynomials, introduced by \textcite{canny90}.
Here, instead of just homogenizing an affine system, an additional perturbation is used to eliminate degenerate behaviour at infinity.
Using this method, \textcite{Ierardi89} obtained a clean reduction from the task of testing satisfiability of affine systems to computing multivariate resultants of polynomials with coefficients in small polynomial rings.
In this subsection, we recall Ierardi's reduction.
We direct the reader to Ierardi's PhD thesis (\cite{ierardi1989complexity}) for an explanation of the rather beautiful geometric ideas that underlie the reduction and related constructions.

We now quote Ierardi's reduction from testing satisfiability of affine systems to computing multivariate resultants.
In the following result and in the rest of this section, we make the following assumptions:
\begin{enumerate}[font=\sffamily,label=(\Alph*)]
\item the field $\bF$ is one of $\bQ$, $\bF_{p^{a}}$, $\bQ\br{y_{1}, \dots, y_{k}}$ or $\bF_{p^{a}}\br{y_{1}, \dots, y_{k}}$ \label{item:assA}
\item in the latter two cases, the input to our algorithms have coefficients that are polynomial in $\vy$.\label{item:assB}
\end{enumerate}
In particular, the height of such input, as introduced in \cref{subsection: finite field arithmetic}, is always well-defined.
\begin{proposition}[\cite{Ierardi89}] \label{lemma: hn to resultant}
  Suppose that  \textsf{\labelcref{item:assA}} and \textsf{\labelcref{item:assB}} hold, and let $f_{1}, \dots, f_{m} \in \bF\bs{x_{1}, \dots, x_{n}}$
  be polynomials of degree at most $d$ and height at most $h$.
  Suppose the coefficient subfield of $\bF$ has at least $15 n d^{n}$ elements.
  There is a polynomial-time Monte Carlo algorithm with success probability $2/3$ that takes as inputs $f_{1}, \dots, f_{m}$, and produces a set of polynomials $G_{i, j} \in \bF\bs{t, w, u, x_{0}, \dots, x_{n}}$ with $0 \leq i \leq n$ and $1 \leq j \leq n$ with the following properties.
  \begin{itemize}
    \item Each $G_{i, j}$ is homogeneous in $x_{0}, \dots, x_{n}$ of degree at most $d$.
    \item Each $G_{i, j}$ has degree at most $n$ in $w$ and degree at most one in $t$ and $u$.
    \item Each $G_{i, j}$ has height at most $h \cdot \br{n \log d}^{c}$ for a universal constant $c$.
    \item The variety $\var{f_{1}, \dots, f_{m}}$ is nonempty if and only if there exists a $j$ such that 
    \[
      \br{\Tt_{w} \Tt_{t} \res\br{G_{0, j}, \dots, G_{n, j}}}\br{0} = 0,
    \]
    where the resultant $\res\br{G_{0,j}, \ldots, G_{n,j}}$ is computed by regarding the $G_{i,j}$ as polynomials in $\vx$ with coefficients in $\bF(t,w,u)$.
  \end{itemize}
\end{proposition}

The proof of the above reduction is essentially the content of \cite[\textsection\textsection 2--3]{Ierardi89}.
We provide an outline of the proof below, although we import a technical statement regarding the multivariate resultant without proof (\cref{lemma: uv resultant lemma}).
This technical lemma will also be useful when we show how the resultant can be used to count solutions in a zero dimensional system.

We start by defining a notion of limit set for projective varieties.
\begin{definition} \label{definition: limiting set}
  Let $G_{1}, \dots, G_{n} \in \bF\bs{x_{0}, \dots, x_{n}}$ be forms of degrees $d_{1}, \dots, d_{n}$ respectively.
  Let $V_{0} \coloneqq \var{G_{1}, \dots, G_{n}}$ be the variety defined by $G_{1}, \dots, G_{n}$.
  The \emph{limit set of $V_{0}$} is the subvariety $V_{0}^{*}$ of $V_{0}$ obtained by the following procedure.

  Let $t$ be a new variable.
  Define $\widehat{G_{i}} \coloneqq G_{i} + t \cdot x_{i}^{d_{i}}$ and let $V \coloneqq \var{\widehat{G_{1}}, \dots, \widehat{G_{n}}} \subseteq \bP^{n} \times \bA$.
  Let $V^{*}$ denote the union of the components of $V$ that are not contained within a subspace of the form $\var{t - \tau}$ for any $\tau \in \bF$.
  Equivalently, these are the components of $V$ whose projection onto $\bA$ is surjective.
  Define $V^{*}_{0} \coloneqq V^{*} \cap \var{t = 0}$, which we treat as a variety in $\bP^{n}$.
\end{definition}

\begin{lemma} \label{lemma: limiting set is finite}
  For any forms $G_{1}, \dots, G_{n} \in \bF\bs{x_0,\ldots,x_n}$ of degrees $d_1, \ldots, d_n$, the limit set $V_{0}^{*}$ is a finite subset of $\var{G_{1}, \dots, G_{n}}$ that contains all its isolated points.
  The set $V_{0}^{*}$ has size at most $\prod_{i=1}^{n} d_{i}$.
\end{lemma}
\begin{proof}
  The first statement is the content of \cite[Lemma~2.3]{Ierardi89}.
  The second statement follows from Bézout's inequality (\cite[Theorem~8.28]{BCS97}) applied to $\widehat{G_{1}}, \dots, \widehat{G_{n}}$: since $V_{0}^{*}$ is finite, its cardinality is bounded above by the Bézout number of $\widehat{G_{1}}, \dots, \widehat{G_{n}}$ seen as polynomials in $\bF(t)[x_0,\ldots,x_n]$.
\end{proof}

We now quote the technical lemma from \textcite{Ierardi89} that produces a polynomial $R^*_0$ whose factors are linear polynomials that correspond to the points of the limit set $V_0^*$.

\begin{lemma}[{\cite[Lemma~2.6]{Ierardi89}}] \label{lemma: uv resultant lemma}
  Let $G_{1}, \dots, G_{n} \in \bF\bs{x_{0}, \dots, x_{n}}$ be forms of degrees $d_{1}, \dots, d_{n}$.
  Let $F_{1}, \dots, F_{m} \in \bF\bs{x_{0}, \dots, x_{n}}$ be additional homogeneous forms of degree $d$.
  Define $\widehat{G_{i}} \coloneqq G_{i} + t x_{i}^{d_{i}}$.
  Define $L \in \bF\bs{x_{0}, \dots, x_{n}, u_{0}, \dots, u_{n}, v_{1}, \dots, v_{m}, t}$ as $L \coloneqq \sum_{i=0}^{n} u_{i} x_{i}^{d} + \sum_{j=1}^{m} v_{j} F_{j}$.
  Define
  \[
    R^{*}_{0} \coloneqq \Tt_{t} \res\br{\widehat{G_{1}}, \dots, \widehat{G_{n}}, L},
  \]
  where the resultant $\res\br{\widehat{G_{1}}, \dots, \widehat{G_{n}}, L}$ is computed by viewing the $G_i$ and $L$ as polynomials in $\vx$ with coefficients in $\bF(\vu,\vv,t)$.

  Then the polynomial $R^{*}_{0} \in \bF[\vu,\vv]$ factors into a product of linear forms over $\overline{\bF}$.
  For each point $\valpha \in V_0^*$ in the limit set of $\var{G_{1}, \dots, G_{n}}$, the linear polynomial $\sum_{i=0}^{n} u_{i} \alpha_{i}^{d} + \sum_{j=1}^{m} v_{j} F_{j}(\valpha)$ is a factor of $R^{*}_{0}$.
  Moreover, every factor of $R^{*}_{0}$ has this structure for some $\valpha \in V_{0}^{*}$.
\end{lemma}

The above lemma generalizes the method of $U$-resultants for zero-dimensional projective systems.
As a consequence of \cref{lemma: uv resultant lemma}, we can use resultants to detect isolated points of an affine system of equations.

\begin{lemma} \label{lemma: affine system detecting isolated roots}
  Let $B \subseteq \bF \setminus \bc{0}$ be a subset of the coefficient subfield of $\bF$.
  Let $g_{1}, \dots, g_{n}, f_{1}, \dots, f_{m} \in \bF\bs{x_{1}, \dots, x_{n}}$ be polynomials of degree at most $d$.
  Let $\gamma_{1}, \dots, \gamma_{m}$ be nonzero field elements picked independently and uniformly at random from $B$.
  Define $F_i$ and $G_{i}$ to be the homogenizations of the polynomials $f_i$ and $g_i$, respectively, with respect to a new variable $x_{0}$.
  Define $\widehat{G_{i}} \coloneqq G_{i} + t x_{i}^{\deg g_{i}}$ and $H \coloneqq u_{0} x_{0}^{d} + \sum_{i=1}^{n} w^{i} x_{i}^{d} + \sum_{j=1}^{m} \gamma_{j} x_{0}^{d - \deg F_{j}} F_{j}$.
  Define
  \[
    \beta \coloneqq \br{\Tt_{w} \Tt_{t} \res\br{\widehat{G_{1}}, \dots, \widehat{G_{n}}, H}}\br{0},
  \]
  where the resultant $\res\br{\widehat{G_{1}}, \dots, \widehat{G_{n}}, H}$ is computed by viewing the $\widehat{G_{i}}$ and $H$ as polynomials in $\vx$ with coefficients in $\bF(u_0,w,t)$.

  With probability at least $1 - d^{n} / |B|$ over the choice of $\gamma$, the following holds.
  If $\beta = 0$, then $\var{f_{1}, \dots, f_{m}, g_{1}, \dots, g_{n}} \neq \emptyset$, and in particular $\var{f_{1}, \dots, f_{m}} \neq \emptyset$.
  Conversely, if there exists an isolated point $\valpha' \in \var{g_{1}, \dots, g_{n}}$ such that $\valpha' \in \var{f_{1}, \dots, f_{m}}$, then $\beta = 0$.
\end{lemma}

\begin{proof}
  Define $L \coloneqq u_{0} x_{0}^{d} + \sum_{i=1}^{n} u_{i} x_{i}^{d} + \sum_{j=1}^{m} v_{j} x_{0}^{d - \deg F_{j}} F_j$, where the $u_{i}$ and $v_{j}$ are variables.
  Let
  \[
    R \coloneqq \res\br{\widehat{G_{1}}, \dots, \widehat{G_{n}}, L},
  \]
  where this resultant is computed by viewing the $\widehat{G_{i}}$ and $L$ as polynomials in $\vx$ with coefficients in $\bF(\vu,\vv,t)$.
  Let $R^{*}_{0} \coloneqq \Tt_{t}R$ and let $V^{*}_{0}$ be the limit set (as defined in \cref{definition: limiting set}) of $G_1, \ldots, G_n$.

  By \cref{lemma: uv resultant lemma}, we know that $R^{*}_{0}$ factors into linear forms that look like
  \[
    \sum_{i=0}^n u_i \alpha_i^d + \sum_{j=1}^m v_j \alpha_0^{d - \deg F_j} F_j(\valpha),
  \]
  where $\valpha \in V_0^*$ is a point in the limit set of $G_1, \ldots, G_n$.
  If there is a factor of $R^{*}_{0}$ that is independent of $v_{1}, \dots, v_{m}$, then the corresponding point $\valpha$ lies in $\var{x_{0}^{d - \deg F_{1}} F_{1}, \dots, x_{0}^{d - \deg F_{m}} F_{m}}$.
  If this factor further depends on $u_{0}$, then we can write $\valpha = \br{1, \valpha'}$, and the point $\valpha'$ lies in $\var{g_{1}, \dots, g_{n}, f_{1}, \dots, f_{m}}$.
  Conversely, if there is an isolated point $\valpha' \in \var{g_{1}, \dots, g_{n}}$ such that $\valpha' \in \var{f_{1}, \dots, f_{m}}$, then $\br{1, \valpha'} \in V_{0}^{*}$ and the corresponding factor of $R^{*}_{0}$ will depend on $u_{0}$ and be independent of $v_{1}, \dots, v_{m}$.

  Now let $r$ denote the resultant 
  \[
    \res\br{\widehat{G_{1}}, \dots, \widehat{G_{n}}, H},
  \]
  where the polynomials $\widehat{G_{i}}$ and $H$ are viewed as polynomials in $\vx$ with coefficients in $\bF(u_0, w, t)$.
  The form $H$ is obtained from $L$ by substituting $u_{i} \mapsto w^{i}$ for $i \in [n]$ and $v_{j} \mapsto \gamma_{j}$ for $j \in [m]$.
  Since the resultant is a polynomial function of the coefficients, we have $r = R(t, u_{0}, w, w^{2}, \dots, w^{n}, \gamma_{1}, \dots, \gamma_{m})$.
  Under the same substitution, a factor of $R_0^*$ corresponding to $\valpha \in V_0^*$ is mapped to
  \[
    u_0 \alpha_{0}^{d} + \sum_{i=1}^n w^i \alpha_i^d + \sum_{j=1}^m \gamma_j \alpha_0^{d - \deg F_j} F_j(\valpha).
  \]
  Each factor of $R_0^*$ depends on at least one of $u_{0}, \dots, u_{n}$, since $\valpha$ is a point in projective space.
  Combined with the fact that the polynomials $u_{0}, w, \dots, w^{n}$ are linearly independent, it follows that no factor of $R^{*}_{0}$ is mapped to zero under the above substitution.
  In particular, $R^{*}_{0}$ is nonzero under this substitution.
  This implies that $\Tt_{t} r = R^{*}_{0}(u_{0}, w, w^{2}, \dots, w^{n}, \gamma_{1}, \dots, \gamma_{n})$.
  Define $r_0^* \coloneqq \Tt_t r$.
  We have $\beta = (\Tt_w r_0^*)(0)$ by definition.

  Any factor of $R^{*}_{0}$ that depends on one of the variables in $\vv$ is mapped to a polynomial with nonzero constant term with probability at least $1 - 1/|B|$.
  Since $V_{0}^{*}$ has size at most $d^{n}$ (\cref{lemma: limiting set is finite}), every such factor of $R_0^*$ is mapped to a polynomial with a nonzero constant term with probability at least $1 - d^{n} / \abs{B}$.
  We show that the conclusion of the lemma holds whenever this event occurs.

  Assume therefore that the event occurs.
  Suppose $\beta = 0$.
  This implies that some factor of $r_0^*$ not divisible by $w$ has a constant term of zero.
  By assumption, every factor of $R_0^*$ that depends on a variable in $\vv$ results in a factor of $r_0^*$ that has a nonzero constant term, so it must be the case that some factor $\hat{R}$ of $R_0^*$ does not depend on any variable in $\vv$.
  Moreover, the factor $\hat{R}$ depends on $u_0$, since the corresponding factor of $r_0^*$ is not divisible by $w$.
  Together, these conditions imply that the point $\valpha \in V_0^* \subseteq V(G_1, \ldots, G_n)$ corresponding to $\hat{R}$ satisfies $\alpha_0 \neq 0$ and $F_j(\valpha) = 0$ for all $j \in [m]$.
  This means that $\valpha' \in \var{g_1, \ldots, g_n, f_1, \ldots, f_m}$, where $\valpha = (1, \valpha')$.
  This set is thus nonempty as claimed.

  Conversely, suppose there is an isolated point $\valpha' \in \var{g_1, \ldots, g_n}$ such that $\valpha' \in \var{f_1, \ldots, f_m}$.
  Because $\valpha'$ is an isolated point of $\var{g_1, \ldots, g_n}$, we have $(1, \valpha') \in V_0^*$, so there is a factor of $R_0^*$, and thus a factor of $r_0^*$, corresponding to $(1, \valpha')$.
  The corresponding factor of $r_0^*$ has the form
  \[
    u_0 + \sum_{i=1}^n w^i \alpha_i^d.
  \]
  This polynomial is not divisible by $w$ and has a constant term of zero, so it follows that $\beta = 0$, as desired.
\end{proof}

With \cref{lemma: affine system detecting isolated roots}, we can finish the proof of \cref{lemma: hn to resultant}.
The argument will require the following standard results regarding random hyperplane sections and varieties defined by random linear combinations of a given set of polynomials.

\begin{lemma}\label{lemma: intersect dimension drop}
    Let $B \subseteq \bF$ be a finite set.
    Let $V \subseteq \bP^{n}$ be a projective variety of dimension $r$ and degree $D$.
    Suppose $\ell$ is a linear form with each coefficient picked independently and uniformly at random from $B$.
    Then with probability at least $1 - D /\abs{B}$, the intersection $V \cap \var{\ell}$ has dimension $r-1$.
    If $\ell_{1}, \dots, \ell_{r+1}$ are linear forms chosen the same way, then $V \cap \var{\ell_{1}, \dots, \ell_{r+1}} = \emptyset$ with probability at least $1 - (r+1) D / \abs{B}$.

    Let $W \subseteq \bA^{n}$ be a projective variety of dimension $r$ and degree $D$.
    Suppose $\ell$ is a linear polynomial with each coefficient picked independently and uniformly at random from $B$.
    Then with probability at least $1 - 2D /\abs{B}$, the intersection $W \cap \var{\ell}$ has dimension $r-1$.
    If $\ell_{1}, \dots, \ell_{r+1}$ are linear polynomials chosen the same way, then $W \cap \var{\ell_{1}, \dots, \ell_{r+1}} = \emptyset$ with probability at least $1 - 2(r+1) D / \abs{B}$.
\end{lemma}

\begin{lemma}
  \label{lemma: input processing}
  Let $B \subseteq \bF$ be a finite set.
  Let $f_{1}, \dots, f_{m} \in \bF\bs{x_{1}, \dots, x_{n}}$ be polynomials of degree at most $d$.
  Suppose $g_{1}, \dots, g_{n+1}$ are linear combinations of $f_{1}, \dots, f_{m}$, with each coefficient picked independently and uniformly at random from $B$.
  Then with probability at least $1 - \br{n+1} d^{n} /\abs{B}$, the following is true.
  For each $s \in [n+1]$, every irreducible component of $\var{g_{1}, \dots, g_{s}}$ of codimension less than $s$ is a component of $\var{f_{1}, \dots, f_{m}}$.
  As a consequence, $\var{g_{1}, \dots, g_{n+1}} = \var{f_{1}, \dots, f_{m}}$.
  The analogous result in the projective setting also holds.
\end{lemma}
The proof of both \cref{lemma: intersect dimension drop} and \cref{lemma: input processing} are implicit in the methods of \textcite{Heintz83}.
For an explicit proof of the latter, see for example \cite[Lemma~36]{KP96}.
The former can be proved using essentially the same arguments.

\begin{proof}[Proof of \cref{lemma: hn to resultant}]
  If the coefficient subfield of $\bF$ is $\bQ$, let $B$ be the set of natural numbers $\bc{1, \dots, 15 n d^{n}}$.
  If the coefficient subfield of $\bF$ is finite, let $B$ be the set of nonzero elements of the coefficient subfield.
  For a fixed $j \in \bs{n}$, we construct $G_{0, j}, \dots, G_{n, j}$ as follows.
  Let $g_{1, j}, \dots, g_{j, j}$ be random linear combinations of $f_{1}, \dots, f_{m}$, where the coefficients are picked independently and uniformly from $B$.
  Let $g_{j+1, j}, \dots, g_{n, j}$ be random linear polynomials, with coefficients picked independently and uniformly from $B$.
  For $i \in \bs{n}$, let $G_{i,j}$ be defined as $G_{i, j} \coloneqq g_{i, j}^{h} + t x_{i}^{\deg g_{i,j }}$, where $g_{i, j}^{h}$ is the homogenization of $g_{i, j}$ with respect to a new variable $x_{0}$.
  To construct $G_{0, j}$, we pick elements $\gamma_{1, j}, \dots, \gamma_{m, j}$ independently and uniformly at random from $B$, and set $G_{0, j} \coloneqq u_{0} x_{0}^{d} + \sum_{i=1}^{n} w^{i} x_{i}^{d} + \sum_{k=1}^{m} \gamma_{k, j} x_{0}^{d - \deg f_{k}} F_{k}$, where $F_k$ is the homogenization of $f_k$ with respect to the variable $x_0$.
  These polynomials have the claimed degree and height bounds.
  For each $j$, define
  \[
    \beta_{j} \coloneqq \br{\Tt_{w} \Tt_{t} \res\br{G_{1, j}, \dots, G_{n, j}, G_{0, j}}}\br{0}.
  \]
  With probability at least $1 - n d^{n}/|B|$, the conclusion of \cref{lemma: affine system detecting isolated roots} holds for all the $\beta_{j}$.
  For the rest of the proof, we assume that this event occurs.
  If $\beta_{j} = 0$ for any $j$, then $\var{f_{1}, \dots, f_{m}} \neq \emptyset$.
  To complete the proof, it suffices to show that if $\var{f_{1}, \dots, f_{m}} \neq \emptyset$, then $\beta_{j} = 0$ for some $j$.
  We show that this happens with high probability.

  To this end, let $s \coloneqq n - \dim \var{f_{1}, \dots, f_{m}}$.
  By \cref{lemma: input processing}, with probability at least $1 - 2 n d^{n}/|B|$, the polynomials $g_{1, s}, \dots, g_{s, s}$ define an equidimensional variety of codimension $s$.
  Each component of $\var{f_{1}, \dots, f_{m}}$ is contained in some component of $\var{g_{1, s}, \dots, g_{s, s}}$.
  In particular, each component of $\var{f_{1}, \dots, f_{m}}$ of maximal dimension is a component of $\var{g_{1, s}, \dots, g_{s, s}}$.

  Since $g_{s+1, s}, \dots, g_{n, s}$ are linear polynomials with coefficients from $B$, by \cref{lemma: intersect dimension drop} it holds that
  \[
    \dim \br{\var{g_{1, s}, \dots, g_{s, s}} \cap \var{g_{s+1, s}, \dots, g_{n, s}}} = 0
  \]
  with probability at least $1 - 2 n d^{n}/|B|$.
  In particular, the variety $\var{g_{1,s}, \ldots, g_{n, s}}$ is nonempty and every point is isolated.
  Further, this variety contains points from each component of $\var{g_{1, s}, \dots, g_{s, s}}$.
  By the previous observation, it contains points from $\var{f_{1}, \dots, f_{m}}$.
  The converse direction of \cref{lemma: affine system detecting isolated roots} now shows that $\beta_{s} = 0$.
  Using a union bound we deduce that the above algorithm succeeds with probability at least $1 - 5nd^{n} / \abs{B}$, which can be lower bounded by $2/3$ using the size of $B$.
\end{proof}

\subsection{Counting solutions in zero-dimensional systems} \label{subsection: counting to resultant}

As we have seen so far, resultants are a useful tool for deciding the satisfiability of systems of polynomial equations, even in the non-square and affine cases.
If a system of equations is satisfiable, a natural next task is to determine how many solutions the system has.
Of course, the space of solutions may be positive-dimensional, in which case there are an infinite number of solutions.
When the solution set is zero-dimensional (and hence a finite set), we can meaningfully speak about the task of counting the number of solutions.
In this subsection, we show that the resultant is also useful for the task of counting the number of solutions to a zero-dimensional system of equations.

We start with the easier case of zero-dimensional projective systems.

\begin{lemma} \label{lemma: hn counting projective to resultant}
  Suppose that \textsf{\labelcref{item:assA}} and \textsf{\labelcref{item:assB}} hold, and let $F_{1}, \dots, F_{m} \in \bF\bs{x_{1}, \dots, x_{n}}$
  be homogeneous polynomials of degree at most $d$ and height at most $h$.
  Suppose the coefficient subfield of $\bF$ has at least $100 n d^{2n}$ elements.
  Suppose $\var{F_{1}, \dots, F_{m}}$ is a finite nonempty set.
  There is a polynomial-time Monte Carlo algorithm with success probability $2/3$ that takes as inputs $F_{1}, \dots, F_{m}$, and produces two sets of polynomials $G_{0, 1}, \dots, G_{n, 1}$ and $G_{0, 2}, \dots, G_{n, 2}$ with coefficients in $\bF\bs{u}$ with the following properties.
  \begin{itemize}
    \item Each $G_{i, j}$ is homogeneous in $x_{0}, \dots, x_{n}$ of degree at most $d$ and has degree at most one in $u$.
    \item Each $G_{i, j}$ has height at most $h \cdot \br{n \log d}^{c}$ for a universal constant $c$.
    \item The size of $\var{F_{1}, \dots, F_{m}}$ is exactly the number of distinct roots (in $\overline{\bF}$) of
      \[
        \gcd \br{ \res\br{G_{0, 1}, \dots, G_{n, 1}} ,\res\br{G_{0, 2}, \dots, G_{n, 2}}},
      \]
      where these resultants are computed by regarding the $G_{i,j}$ as polynomials in $\vx$ with coefficients in $\bF(u)$.
  \end{itemize}
\end{lemma}

\begin{proof}
  If the coefficient subfield of $\bF$ is $\bQ$, let $B$ be the set of natural numbers $\bc{1, \dots, 100 n d^{2n}}$.
  If the coefficient subfield of $\bF$ is finite, let $B$ be the set of nonzero elements of the coefficient subfield.
  We work with the polynomials
  \[
    \setbuild{x_{j}^{d - \deg F_{i}} \cdot F_{i}}{i \in [m], j \in \bc{0,1,\ldots,n}}
  \]
  instead of the original polynomials, as this does not change the zero set but ensures that all polynomials in our system have the same degree.
  In the rest of the proof we continue to use $F_{1}, \dots, F_{m}$ to refer to this new set of polynomials.

  Let $V \coloneqq \var{F_{1}, \dots, F_{m}}$ be the zero set of $F_{1}, \dots, F_{m}$.
  It consists of at most $d^{n}$ points by Bézout's inequality (\cite[Theorem~8.28]{BCS97}).
  We apply a random change of coordinates to ensure that all points in $V$ lie on the chart $x_{0} = 1$.
  We do this by replacing each $x_{i}$ by $\sum_{j=0}^n \gamma_{i, j} x_{j}$, where $\gamma_{i, j}$ are picked from $B$ independently and uniformly at random.
  With probability at least $1 - d^{n} / \abs{B}$, the points in $V$ lie on $x_0 = 1$ after this linear transformation.

  Let $G_{1, 1}, \dots, G_{n, 1}$ be a set of random linear combinations of the equations $F_{1}, \dots, F_{m}$, where each coefficient is picked independently and uniformly from $B$.
  By \cref{lemma: input processing}, with probability at least $1 - 2 n d^{n} /\abs{B}$, the zero set $V_{1} \coloneqq \var{G_{1, 1}, \dots, G_{n, 1}}$ is a finite set and contains $V$ as a subset.
  The system $F_{1}|_{x_{0} = 0} = \cdots = F_{m}|_{x_{0} = 0} = 0$ is unsatisfiable.
  The system $G_{1, 1}|_{x_{0} = 0} = \cdots = G_{n, 1}|_{x_{0} = 0} = 0$ consists of random linear combinations of the $F_i|_{x_{0}=0}$, so by \cref{lemma: input processing}, this system has no roots with probability at least $1 - 2 n d^{n} / \abs{B}$.
  Equivalently, with this probability, every point of $V_{1}$ lies on the chart $x_{0} = 1$.

  Similarly, let $G_{1, 2}, \dots, G_{n, 2}$ be a set of random linear combinations of $F_{1}, \dots, F_{m}$ and let $V_{2}$ be their zero set.
  It again is a finite set that contains $V$ and lies on the chart $x_{0} = 1$.
  Further, we claim that if $V_1$ is a finite set, then $V_{1} \cap V_{2} = V$ with probability at least $1  - d^n / |B|$.
  If $\valpha \in V_{1} \setminus V$ is a point, then there is some $F_{j}$ such that $F_{j}(\valpha) \neq 0$, therefore with probability at least $1 - 1/\abs{B}$, we have $G_{1, 2}(\valpha) \neq 0$.
  The fact that $V_1$ is finite implies that $\abs{V_1} \le d^n$, so the claimed bound on the probability that $V_1 \cap V_2 = V$ follows by a union bound.

  We now sample $\beta_1, \ldots, \beta_n \in B$ independently and uniformly, and define $L \coloneqq x_{0} + \beta_{1} u x_{1} + \beta_{2} u x_{2} + \cdots + \beta_{n} u x_{n}$.
  By the Poisson formula (\cref{theorem: multivariate poisson formula}), the resultant $\res\br{G_{1, 1}, \dots, G_{n, 1}, L}$ can be factored as
  \[
    \res\br{G_{1, 1}, \dots, G_{n, 1}, L} = c \cdot \prod_{\valpha \in V_{1}} L(\valpha)^{m_{\valpha}},
  \]
  where $m_{\valpha}$ is the multiplicity of $\valpha$ in $V_{1}$ and $c \in \bF \setminus \bc{0}$.
  Simplifying, and using the fact that all roots lie on the affine chart $x_0 = 1$, we obtain
  \[
    \res\br{G_{1, 1}, \dots, G_{n, 1}, L} = c \cdot \prod_{\valpha \in V_{1}} \br{1 + u \sum_{i=1}^{n} \beta_{i} \alpha_{i}}^{m_{\valpha}}.
  \]
  Since the $\beta_{i}$ were picked randomly, the sum $\sum_{i=1}^n \beta_{i} \alpha_{i}$ takes distinct values for distinct roots with probability at least $1 - d^{2n} /\abs{B}$.
  Similarly, we have
  \[
    \res\br{G_{1, 2}, \dots, G_{n, 2}, L} = c' \cdot \prod_{\valpha \in V_{2}} \br{1 + u \sum_{i=1}^{n} \beta_{i} \alpha_{i}}^{n_{\valpha}},
  \]
  where $n_{\valpha}$ is the multiplicity of $\valpha$ in $V_2$ and $c' \in \bF \setminus \bc{0}$.
  It is now clear that if we take the GCD of these two resultants, then the number of distinct roots of the GCD in $\overline{\bF}$ is equal to the number of distinct points in $V_{1} \cap V_{2} = V$.
  Therefore we are done by setting $G_{0, 1} = G_{0, 2} = L$ and using a union bound to control the probabilities of all the required events.
\end{proof}

We now give our main reduction, showing that the resultant can be used to count the number of solutions to a zero-dimensional affine system.
The idea will be similar to that in \cref{lemma: hn counting projective to resultant}.
We will take random linear combinations of the given system to reduce to the case when the number of variables and polynomials are the same.
Here, we have to use the generalized characteristic polynomial instead of the resultant itself to compute the sizes of zero sets of these random systems.

\begin{proposition} \label{lemma: hn counting to resultant}
  Suppose that \textsf{\labelcref{item:assA}} and \textsf{\labelcref{item:assB}} hold, and let $f_{1}, \dots, f_{m} \in \bF\bs{x_{1}, \dots, x_{n}}$ be polynomials of degree at most $d$ and height at most $h$.
  Suppose the coefficient subfield of $\bF$ has at least $100 n d^{2n}$ elements.
  Suppose $\var{f_{1}, \dots, f_{m}}$ is a finite nonempty set.
  There is a polynomial-time Monte Carlo algorithm with success probability $2/3$ that takes as inputs $f_{1}, \dots, f_{m}$, and produces two sets of polynomials $G_{0, 1}, \dots, G_{n, 1}$ and $G_{0, 2}, \dots, G_{n, 2}$ with coefficients in $\bF\bs{u, t}$ with the following properties.
  \begin{itemize}
    \item Each $G_{i, j}$ is homogeneous in $x_{0}, \dots, x_{n}$ of degree at most $d$ and has degree at most one in $u$ and $t$.
    \item Each $G_{i, j}$ has height at most $h \cdot \br{n \log d}^{c}$ for a universal constant $c$.
    \item The size of $\var{f_{1}, \dots, f_{m}}$ is exactly the number of distinct roots (in $\overline{\bF}$) of
      \[
        \gcd\br{\mathrm{TP}_{u} \Tt_{t} \res\br{G_{0, 1}, G_{1, 1}, \dots, G_{n, 1}}, \mathrm{TP}_{u} \Tt_{t} \res\br{G_{0, 2}, G_{1, 2}, \dots, G_{n, 2}}},
      \]
      where these resultants are computed by regarding the $G_{i,j}$ as polynomials in $\vx$ with coefficients in $\bF(t,u)$.
  \end{itemize}
\end{proposition}
\begin{proof}
  If the coefficient subfield of $\bF$ is $\bQ$, let $B$ be the set of natural numbers $\bc{1, \dots, 100 n d^{2n}}$.
  If the coefficient subfield of $\bF$ is finite, let $B$ be the set of nonzero elements of the coefficient subfield.
  Define $V \coloneqq \var{f_{1}, \dots, f_{m}}$.
  Let $g_{1, 1}, \dots, g_{n, 1}$ and $g_{1, 2}, \dots, g_{n, 2}$ be two random sets of linear combinations of $f_{1}, \dots, f_{m}$, with coefficients picked from $B$.
  Let $V_{1} \coloneqq \var{g_{1, 1}, \dots, g_{n, 1}}$ and $V_{2} \coloneqq \var{g_{1, 2}, \dots, g_{n, 2}}$.
  By Bézout's inequality and \cref{lemma: input processing}, with probability at least $1 - 6 n d^{n}/\abs{B}$, the zero sets $V_{1}$ and $V_{2}$ are finite and satisfy $V_1 \cap V_2 = V$ (refer to the proof of \cref{lemma: hn counting projective to resultant} for details).

  We now define $G_{i, j} \coloneqq g_{i, j}^{h} + t \cdot x_{i}^{d_{i}}$, where $d_{i}$ is the degree of $g_{i,j}$ and $g_{i,j}^h$ is the homogenization of $g_{i,j}$ with respect to a new variable $x_{0}$.
  Finally, define $L \coloneqq x_{0} + \beta_{1} u x_{1} + \cdots + \beta_{n} u x_{n}$ where $\beta_{1}, \dots, \beta_{n}$ are sampled independently and uniformly at random from $B$.
  Consider now
  \[
    \br{R_{1}}_{0}^{*} \coloneqq \Tt_{t} \res\br{G_{1, 1}, \dots, G_{n, 1}, L}.
  \]
  By \cref{lemma: uv resultant lemma}, we can deduce that $(R_1)^*_0$ factors into a product of linear forms, one for each point in the limiting set $\br{V_{1}}_{0}^{*}$.
  Every point of $V_{1}$ is in this this set, and further these are the only points of $(V_1)^*_0$ that lie on the chart $x_{0} = 1$.
  The remaining factors that correspond to points in $\br{V_{1}}_{0}^{*}$ on the hyperplane $x_{0} = 0$ are just the polynomial $u$.
  Therefore, if we consider $\mathrm{TP}_{u} \br{R_{1}}_{0}^{*}$, every factor is of the form $\br{1 + u \sum_{i=1}^{n} \beta_{i} \alpha_{i}}$ of some $\valpha \in V_{1}$, and for every $\valpha \in V_{1}$ there is a factor of this form.
  Note that we make no claim about the multiplicities.
  With probability at least $1 - d^{2n}/\abs{B}$, the coefficient of $u$ corresponding to each point, namely $\sum_{i=1}^{n} \beta_{i} \alpha_{i}$, is distinct.
  We can do the same operation with $g_{1, 2}, \dots, g_{n, 2}$.
  Then the number of distinct roots of $\gcd\br{\mathrm{TP}_{u} \br{R_{1}}_{0}^{*}, \mathrm{TP}_{u} \br{R_{2}}_{0}^{*}}$ in $\overline{\bF}$ is clearly the size of $V_{1} \cap V_{2}$ with probability at least $1 - 2 d^{2n} / \abs{B}$ over the choice of $\beta_{1}, \dots, \beta_{n}$.
  Hence we are done by setting $G_{0, 1} = G_{0, 2} = L$.
\end{proof}

%% file: sections/resultant.tex
In this section, we design a polylogtime-uniform family of constant-depth arithmetic circuits for the multivariate resultant $\res_{d_{0}, \dots, d_{n}}$.
Our circuits will be indexed by $n, d_{0}, \dots, d_{n}$, and our base field is $\bQ$.
The following subsection will establish some notation and explain the high level idea of our construction.
The subsequent subsections will contain the details.

\subsection{Notation and proof overview} \label{subsection: notation and proof idea}

Let $n$ be a positive integer, and let $d_{0}, \dots, d_{n}$ be positive integers.
Let $\vu_{0}, \dots, \vu_{n}$ be $n+1$ sets of variables, with $\vu_{i} = \br{u_{i, \valpha}}$ for all multi-indices $\valpha = \br{\alpha_{0}, \dots, \alpha_{n}}$ such that $\abs{\valpha} = d_{i}$.
In other words, $\vu_{i}$ consists of one variable for each monomial in $n+1$ variables of degree $d_{i}$.
We let $U$ denote the sum of the sizes of these sets of variables, so $U = \sum_{i=0}^{n} \binom{n + d_{i}}{d_{i}}$.
We use $S$ to denote the polynomial ring $\bQ\bs{\vu_{0},  \dots, \vu_{n}}$, and we write $\bK$ for the field of fractions of $S$.
Recall that the multivariate resultant (\cref{subsection: multivariate resultant}) is an element of the ring $S$ with integer coefficients.

For $j = 0, \dots, n$, we define polynomials $F_{j} \in S\bs{\vx}$ as
\[
  F_{j} = \sum_{\valpha \in \vu_{j}} u_{j, \valpha} x_{0}^{\alpha_{0}} \cdots x_{n}^{\alpha_{n}}.
\]
Any set of forms $P_{0}, \dots, P_{n}$ of degrees $d_{0}, \dots, d_{n}$ can be obtained by specializing $F_{0}, \dots, F_{n}$.
We wish to use the Poisson formula (\cref{theorem: multivariate poisson formula}) to compute the multivariate resultant.
The formula states that the following identity holds.
\[
  \res_{d_0, \ldots, d_n}(F_0, \ldots, F_n) = \res_{d_0, \ldots, d_{n-1}}(\overline{F}_0, \ldots, \overline{F}_{n-1})^{d_n} \cdot \prod_{\valpha \in \Var(f_0, \ldots, f_{n-1})} f_n(\valpha)^{m(\valpha)},
\]
where $\overline{F}_{j} = F_{j}|_{x_{n} = 0}$ and $f_{j} = F_{j}|_{x_{n} = 1}$.
We will apply this formula recursively to compute the first term, therefore we will have to study the polynomials $F_{j}|_{x_{i+1} = 0, \dots, x_{n} = 0}$ and $F_{j}|_{x_{i} = 1, x_{i+1} = 0, \dots, x_{n} = 0}$ for every $i$.
We give these polynomials and some related objects names to be able to refer to them easily.

For any $i, j \in \bc{0,1,\ldots,n}$, we let $\vu_{i,j} \subseteq \vu_i$ denote those variables $u_{i, \valpha}$ with $\alpha_{j+1} = \cdots = \alpha_n = 0$.
In other words, these variables $\vu_{i,j}$ correspond to monomials that only depend on $x_{0}, \dots, x_{j}$.
We let $F_{i,j}$ denote the polynomial
\[
  F_{i, j} = \sum_{\valpha \in \vu_{i, j}} u_{i, \valpha} x_{0}^{\alpha_{0}} \cdots x_{n}^{\alpha_{n}}.
\]
The polynomial $F_{i, j}$ is a generic degree-$d_{i}$ form in the variables $x_{0}, \dots, x_{j}$.
We use $f_{i, j}$ to denote the specialization $F_{i, j}(x_{0}, \dots, x_{j-1}, 1)$.

Let $\vF_{j}$ denote the vector of polynomials $\br{F_{0, j}, \dots, F_{j, j}}$.
This is a vector of $j+1$ generic polynomials in $j+1$ variables of degrees $d_{0}, \dots, d_{j}$.
We write $U_{j}$ for the number of coefficients in this system, so $U_{j} = \sum_{i=0}^{j} \binom{d_{i} + j}{j}$.

Let $\res_{j} \in \bQ\bs{\vu_{0, j}, \dots, \vu_{j, j}}$ denote the resultant $\res_{d_0,\dots,d_j}(\br{F_{0, j}, \dots, F_{j, j}})$ (the degrees $d_0,\dots,d_j$ are fixed throughout).
The polynomial $\res_{j}$ depends on the degrees $d_{0}, \dots, d_{j}$, but we suppress this from notation to avoid clutter.
Given a set of forms $P_{0}, \dots, P_{j} \in \bQ\bs{x_{0}, \dots, x_{j}}$ of degrees $d_{0}, \dots, d_{j}$, we use $\res_{j}\br{\vP}$ to denote the evaluation of $\res_{j}$ at the coefficients of $P_{0}, \dots, P_{j}$.

In the notation above, the second term in the Poisson formula involves the roots of the polynomials $f_{0, n}, \dots, f_{n-1, n}$.
These roots lie in the algebraic closure of $\bK$.
However, it is computationally difficult to describe and perform arithmetic with elements of $\overline{\bK}$.
We will instead work with roots that lie in a power series ring.
These are easier to manipulate, since we usually only require computing them to some bounded precision.

The roots of $f_{0, n}, \dots, f_{n-1, n}$ do not admit power series representations in the ring $\bQ\bsd{\vu_{0}, \dots, \vu_n}$.
This is already apparent even in the case when $n = 1$, and $d_{0} = 2$.
In this case we want the roots of a generic bivariate polynomial $u_{0, (2, 0)}x_{0}^{2} + u_{0, (1, 1)} x_{0} + u_{0, (0, 2)}$.
The expression for the roots involves the square root of the discriminant $u_{0, (1, 1)}^{2} - 4 u_{0, (2, 0)} u_{0, (0, 2)}$, but this element is not a square in the ring $\bQ\bsd{\vu_{0}, \dots, \vu_n}$.

We will instead find power series roots of $f_{0, n}, \dots, f_{n-1, n}$ using a homotopy.
We introduce a new variable $t$, and for each $i = 0, \dots, n$ we define polynomials
\[
  H_{i, n} := \br{1-t} G_{i, n} + t F_{i, n},
\]
where $G_{i, n} \in \bQ\bs{x_{0}, \dots, x_{n}}$ are yet to be defined.
We have $H_{i, n}|_{t=1} = F_{i, n}$, and $H_{i, n}|_{t=0} = G_{i, n}$.
The polynomials $H_{0, n}, \dots, H_{n, n}$ define a homotopy between the roots of $F_{0, n}, \dots, F_{n, n}$ and the roots of $G_{0, n}, \dots, G_{n,n}$.
The system of equations defined by the polynomials $G_{0, n}, \dots, G_{n, n}$ is called the \emph{initial system} the homotopy.
We define $h_{i, n} := H_{i, n}|_{x_{n} = 1}$ and $g_{i, n} := G_{i}|_{x_{n} = 1}$.
The forms $G_{i, n}$ will be chosen such that computing the roots of the equations $g_{0, n}, \dots, g_{n-1, n}$ can be easily performed by constant-free uniform constant-depth circuits.
We will also ensure that all roots of the system are simple, and that the Jacobian at each root can also be computed by such circuits.

A consequence of the simplicity of the roots of $g_{0, n}, \dots, g_{n-1, n}$ is that $h_{0, n}, \dots, h_{n-1, n}$ admit power series roots in $t$, whose constant terms are exactly the roots of $g_{0, n}, \dots, g_{n-1, n}$.
Hensel lifting and Newton iteration give constructive proofs of the existence of these roots.
In fact a Newton iteration with linear rate of convergence gives a constructive proof of the existence of roots of the above system in $S\bsd{t}$, that is, roots that are power series in one variable with coefficients are polynomial in the remaining variables.
While Newton iteration and Hensel lifting give constructive proofs, following either construction only gives us uniform circuits of polylogarithmic depth for these roots.
This is insufficient for our applications.
We will instead compute these roots using an explicit implicit function theorem (\cite[Proposition~20.3]{AiYu83}).

More generally, to handle the recursive terms in the Poisson formula, we define polynomials $G_{i, j}, H_{i,j}$ for all $1 \leq j \leq n$ and $0 \leq i \leq j$ by specializing $G_{0, n}, \dots, G_{n, n}$ and $H_{0, n}, \dots, H_{n, n}$ respectively.
In each case, $G_{0, j}, \dots, G_{j, j}$ will be the initial system for $H_{0, j}, \dots, H_{j, j}$.
We will apply the Poisson formula and the above mentioned explicit formulas for the roots to compute the resultant of $H_{0, n}, \dots, H_{n, n}$.
This computation will be carried out in $S\bsd{t}$ up to a certain precision.
This precision will be chosen such that truncating and setting $t = 1$ will allow us to recover $\res_{n}\br{\vF_{n}}$.

The next subsection describes the initial systems, and shows that their roots, and the Jacobian evaluated at the roots can be computed by uniform constant-depth circuits.
The subsection after that carries out the rest of the discussion above.

We point out that the algorithm in \cite{JS07} already combines homotopy techniques with the Poisson formula to compute resultants.
The homotopy in that article uses starting systems with a similar structure as the polynomials $g_{i,j}$ mentioned above, but bypasses the introduction of the new variable $t$.
If we let $\vnu_0,\dots,\vnu_j$ be the coefficients of these polynomials, their algorithm uses Newton iteration to produce a straight-line program that computes rational function approximations to the roots of $f_{0,j},\dots,f_{j-1,j}$ in $\bQ[[\vu_0-\vnu_0,\dots,\vu_{j-1}-\vnu_{j-1}]]$, injects them into the product formula, and eliminates divisions.
The explicit implicit function theorem in~\cite{AiYu83} also applies for the homotopy in~\cite{JS07}, but in this context, it would result in circuits of double exponential size.

\subsection{The initial systems} \label{subsection: initial system}

We now describe the initial systems, that is, the polynomials $G_{i, j}$.
The following construction is from~\cite{HeJeSaSo02} (and is also briefly mentioned in~\cite{MoSoWa95}), in the more general case of multi-homogeneous systems.

\begin{definition}[Initial System] \label{definition: initial system}
  For any integer $m$ and for $j \in \bc{0,1,\ldots,n}$, we define the linear form $L_{j, m}$ as
  \[
    L_{j, m}(\vx) \coloneqq x_{0} + m x_{1} + \cdots + m^{j-1} x_{j-1} + m^{j} x_{j}.
  \]
  We use $\ell_{j, m}$ to denote the specialization of $L_{j, m}$ at $x_{j} = 1$.

  Define subsets of integers $A_{0}, \dots, A_{n}$ as $A_{0} \coloneqq \bc{1, \dots, d_{0}}$, $A_{1} \coloneqq \bc{d_{0} + 1, \dots, d_{0} + d_{1}}$, and so on.
  The set $A_{i}$ has size $d_{i}$, and all of these sets are pairwise disjoint.
  For $i \in \bc{0,1,\ldots,n}$, we define the form $G_{i, j} \in \bQ\bs{\vx}$ as
  \[
    G_{i, j}(\vx) \coloneqq \prod_{k \in A_{i}} L_{j, k}(\vx).
  \]
  We use $\vG_{j}$ to denote the vector of polynomials $(G_{0, j}, \dots, G_{j, j})$.
  Finally, $g_{i, j}$ denotes the specialization of $G_{i,j}$ to $x_{j} = 1$, and $\vg_{j}$ denotes the vector of polynomials $(g_{0, j}, \dots, g_{j-1, j})$.
\end{definition}

Observe that each $G_{i, j}$ is a form of degree $d_{i}$ in the variables $x_{0}, \dots, x_{j}$.
The reason why we use $\vg_{j}$ to denote $g_{0, j}, \dots, g_{j-1, j}$ (and not the more natural choice of $g_{0, j}, \dots, g_{j, j}$) is that when applying the Poisson formula to $H_{0, j}, \dots, H_{j, j}$, the term corresponding to the product of roots will only involve solving the system $H_{0, j}|_{x_{j} = 1}, \dots, H_{j-1, j}|_{x_{j} = 1}$.

The following lemma captures some basic properties about the zeroes of the systems $\vg_{j}$.
Define $e_{0} \coloneqq 0$.
For $i = 1, \dots, n$ define $e_{i} \coloneqq d_{0} + \cdots + d_{i-1}$.

\begin{lemma} \label{lemma: initial system roots}
  For each $j \in \bc{0,1,\ldots,n}$, let $B_{j} \coloneqq [d_0] \times [d_1] \times \cdots \times [d_{j-1}]$.
  For each $\vc \in B_{j}$, the linear system $\ell_{j, c_{0}} = \ell_{j, e_{1} + c_{1}} = \cdots = \ell_{j, e_{j-1} + c_{j-1}} = 0$ has a single common solution, which we denote by $\vr_{\vc}$.
  The set of such $\vr_{\vc}$ are exactly the set of common zeroes of $\vg_{j}$.
\end{lemma}
\begin{proof}
  For each choice of $\vc$, the matrix corresponding to the linear system is a Vandermonde matrix, and is therefore invertible.
  This shows that a unique solution exists.
  Further, for any fixed $\vc$, if $\ell$ is a linear form of the type $\ell_{j, m}$ for any $m$ that is not among $e_{0} + c_{0}, \dots, e_{j-1} + c_{j-1}$, then the system of linear equations
  \[
    \ell, \ell_{j, e_{0} + c_{0}}, \dots, \ell_{j, e_{j-1} + c_{j-1}}
  \]
  is unsatisfiable.
  This can be seen by homogenizing the system and observing that the corresponding matrix is again a Vandermonde matrix, so the zero vector is the only solution.
  These two facts also imply that the set of common zeroes of $\vg_{j}$ are exactly $\vr_{\vc}$.
\end{proof}

Next, we compute the Jacobian of the system $\vg_j$ at each of its solutions $\vr_{\vc}$.

\begin{lemma} \label{lemma: initial system jacobian}
  Let the notion be as in the statement of \cref{lemma: initial system roots}.
  For $\vc \in B_{j}$, and $i \leq j$, define $\kappa_{\vc, i}$ as
  \[
    \kappa_{\vc, i} \coloneqq \frac{g_{i, j}}{\ell_{j, e_{i} + c_{i}}}\br{\vr_{\vc}}.
  \]
  Then the Jacobian of $\vg_{j}$ at $\vr_{\vc}$, denoted $\cJ_{\vc}$, is an invertible matrix and factors as 
  \[
    \cJ_{\vc} = \diag\br{\kappa_{\vc, 0}, \dots, \kappa_{\vc, j-1}} \cdot \vand{e_{0} + c_{0}, \dots, e_{j-1} + c_{j-1}}.
  \]
\end{lemma}
\begin{proof}
  The partial derivative of $g_{i, j}$ with respect to $x_{a}$ is given by
  \[
    \partial_{a} g_{i,j} \br{\vx} = \sum_{k'=1}^{d_{i}} \prod_{k \neq k'} \ell_{j, e_{i} + k}(\vx) \cdot \partial_{a} \ell_{j, e_{i} + k'}\br{\vx} = \sum_{k'=1}^{d_{i}} \prod_{k \neq k'} \ell_{j, e_{i} + k}(\vx) \cdot \br{e_{i} + k'}^{a}.
  \]
  When evaluated at $\vr_{\vc}$, the only term that survives in the above summation is the one where $\ell_{j, e_{i} + c_{i}}$ is omitted, since $\ell_{j, e_{i} + c_{i}}$ vanishes at $\vr_{\vc}$.
  In this summand, the first product term when evaluated at $\vr_{\vc}$ is exactly $\kappa_{\vc, i}$.
  This shows the claimed factorization for the Jacobian.
  The proof of \cref{lemma: initial system roots} shows that the $\kappa_{\vc, i}$ are nonzero, since $\vr_{\vc}$ is not a root of any linear polynomial that is a factor of $g_{i, j}$, other than $\ell_{j, e_{i} + c_{i}}$.
  The fact that the Jacobian is invertible now follows from the factorization.
\end{proof}

For each $j$, the solution $\vr_{\vc}$, the Jacobian $\cJ_{\vc}$, and its inverse $\cJ_{\vc}^{-1}$ are all rational functions of $\vc$.
The following lemma describes polylogtime-uniform, constant-free, constant-depth circuits that take as input $\vc$ and compute these rational functions.

\begin{lemma} \label{lemma: uniformity of initial system}
  There exists a polylogtime-uniform family of constant-free circuits $\cC^{\initialRoot}$ indexed by $j, d_{0}, \dots, d_{j-1}$ with the following properties.
  \begin{itemize}
    \item 
      The circuit $C^{\initialRoot}_{j, d_{0}, \dots, d_{j-1}}$ has size polynomial in $j + d_{0} +  \cdots + d_{j-1}$, and depth bounded by a universal constant.
    \item
      The circuit $C^{\initialRoot}_{j, d_{0}, \dots, d_{j-1}}$has $j$ input gates, labeled by the variables $z_{0}, \dots, z_{j-1}$.
      The circuit has $j + 2j^{2}$ output gates.
    \item
      On input $(c_{0}, \dots, c_{j-1}) \in [d_0] \times \cdots \times [d_{j-1}]$, this circuit computes $\vr_{\vc}$, the entries of the Jacobian $\cJ_{\vc}$, and the entries of its inverse $\cJ_{\vc}^{-1}$.
  \end{itemize}
\end{lemma}

\begin{proof}
  We first describe the circuit $C = C^{\initialRoot}_{j, d_{0}, \dots, d_{j-1}}$.
  Recall that $e_k$ denotes the sum $d_0 + \cdots + d_{k-1}$.
  The circuit $C$ has gates that compute all the integers between $1$ and $e_j$ by repeatedly adding the constant $1$.
  It has $+$ gates that compute the polynomials $e_{0} + z_{0}, \dots, e_{j-1} + z_{j-1}$ by adding $z_j$ to the integer $e_{j-1}$.
  The circuit has gates that compute the polynomials $-(e_0 + z_0)^j, \ldots, -(e_{j-1} + z_{j-1})^j$.
  This will require introducing more copies of the gates that compute $e_{0} + z_{0}$.
  The gates for these computations will be named $(1, *)$ for some $*$ (which itself will be a tuple), analogous to how gates performing similar tasks were named in the proof of \cref{lemma: special vand inverses are uniform}.
  The circuit has a copy of the circuit $C'$ for the inverse of a general Vandermonde matrix of size $j \times j$ (\cref{lemma: vandermonde inverse uniform}), whose inputs are connected to the gates computing the above polynomials.
  These gates will be named $(2, v)$ where $v$ is the name of the gate in the circuit $C'$.
  The output of this subcircuit computes the inverse of the Vandermonde matrix $\vand{e_0 + z_0, \ldots, e_{j-1} + z_{j-1}}$.
  The circuit then multiplies the inverse of the Vandermonde matrix with the vector $(-(e_0 + z_0)^j, \ldots, -(e_{j-1} + z_{j-1})^j)$.
  Denote the entries of the resulting vector by $\vr_{\vz}$.
  These gates will be named $(3, *)$.

  Next, for each $i \in \bc{0, \dots, j-1}$, the circuit has gates that compute the polynomials
  \[
    \sum_{k'=1}^{d_{i}} \prod_{k \neq k'} \ell_{j, e_{i} + k}(x_{0}, \dots, x_{j-1}),
  \]
  and evaluates them at $\vr_{\vz}$.
  Call these evaluations $\kappa_{\vz, i}$, the gates will be named $(4, *)$.
  The circuit then has gates that compute the matrix product 
  \[
    \cJ_{\vz} \coloneqq \diag\br{\kappa_{\vz, 1}, \dots, \kappa_{\vz, j-1}} \cdot \vand{e_{0} + z_{0}, \dots, e_{j-1} + z_{j-1}}
  \]
  Finally, the circuit has gates that compute the inverse of $\cJ_{\vz}$ using the factorization above and the entries of $\vand{e_0+z_0,\ldots,e_{j-1}+z_{j-1}}^{-1}$ computed earlier.
  These gates can also be named as above.
  In all cases, a natural numbering of the predecessor of each gates is clear from the description.

  We argue that evaluating this circuit on $\vc$ computes the root $\vr_{\vc}$, the Jacobian $\cJ_{\vc}$ and its inverse correctly.
  It follows from the definition of $\vr_{\vc}$ that $\vr_{\vz}(\vc) = \vr_{\vc}$.
  To see that $\kappa_{\vz, i}(\vc) = \kappa_{\vc, i}$, we observe that when evaluated at $\vr_{\vc}$, the only summand in the expression for $\kappa_{\vz, i}$ that does not vanish is the one where $\ell_{j, e_{i} + c_{i}}$ is omitted.
  From this and \cref{lemma: initial system jacobian}, it follows that the inverse and Jacobian are computed correctly.

  The claims on the size and depth of the circuit follow from the above construction and the bounds on the sizes of the subcircuits used.
  Uniformity follows by a simulation argument, as in the proof of \cref{lemma: polynomial interpolation uniform}, together with the above description.
  For the gates in the copy of $C'$, that is, gates named $(2, v)$ for some $v$, since $v$ can be deduced from this name, the Turing machine can simulate the machine that decides the direct connection language of the circuit family from \cref{lemma: vandermonde inverse uniform} whenever required.
  For all other gates, the above description of the circuit can be used to design the Turing machine, similar to the proof of \cref{lemma: special vand inverses are uniform}.
\end{proof}

\subsection{The resultant of the system defining the homotopy}

We can now describe the system that defines the homotopy.

\begin{definition} \label{definition: homotopy polynomials}
  For any $i, j \in \bc{0,1,\ldots,n}$, let $H_{i, j}(\vx,t)$ denote the polynomial
  \[
    H_{i,j}(\vx, t) \coloneqq (1-t) G_{i,j}(\vx) + t F_{i, j}(\vx),
  \]
  which we regard as an element of $\bK(t)[\vx]$.
  For any such $i$ and $j$, we use $h_{i, j}$ to denote the specialization of $H_{i, j}$ at $x_{j} = 1$.
  Let $\vH_{j}$ denote the vector of polynomials $\br{H_{0, j}, \dots, H_{j, j}}$ and $\vh_{j}$ denote the vector of polynomials $\br{h_{0,j}, \ldots, h_{j-1,j}}$.
\end{definition}

Specializing $H_{i,j}$ and $h_{i,j}$ to $t = 0$ gives us $G_{i, j}$ and $g_{i, j}$ respectively, and specializing to $t = 1$ gives $F_{i, j}$ and $f_{i, j}$.
Further, for $j \in [n]$, setting $x_j$ to zero in $H_{0,j},\dots,H_{j-1,j}$ gives us $H_{0,j-1},\dots,H_{j-1,j-1}$.
Via the Poisson formula, we can derive an expression for $\res_{n}\br{\vH_{n}}$.
\begin{lemma}
  \label{lemma: poisson formula homotopy}
  For each $j \in \bc{0,1,\ldots,n}$, let $V_{j} \subseteq \overline{\bK\br{t}}^{j}$ denote the set of common zeroes of $h_{0, j}, \dots, h_{j-1, j}$.
  We have
  \begin{multline*}
    \res_{n}\br{\vH_{n}} = \br{1 + t \br{u_{0,\br{d_{0}, 0, \dots, 0}} - 1}}^{\prod_{i=1}^{n} d_{i}} \times \br{\prod_{\vrho \in V_{1}} h_{1, 1}(\vrho)}^{\prod_{i=2}^{n} d_{i}} \\ \times \br{\prod_{\vrho \in V_{2}} h_{2, 2}(\vrho)}^{\prod_{i=3}^{n} d_{i}} \times\cdots \times \br{\prod_{\vrho \in V_{n}} h_{n, n}(\vrho)}.
  \end{multline*}
\end{lemma}
\begin{proof}
  We have $F_{0, 0} = u_{0,\br{d_{0}, 0, \dots, 0}}x_{0}^{d_{0}}$ and $G_{0, 0} = x_{0}^{d_{0}}$.
  Therefore, $\res\br{H_{0, 0}} = 1 + t \br{u_{0,\br{d_{0}, 0, \dots, 0}} - 1}$, using the fact that the resultant of a single form in one variable is just the coefficient (\cite[Example~2.1]{Jouanolou91}).
  For the systems $\vH_{j}$, the Poisson formula (\cref{theorem: multivariate poisson formula}) applies, since the specialization to $t=1$ results in generic forms.
  Unrolling the recurrence of the Poisson formula gives us the above expression.
\end{proof}

As described above, the sets $V_{j}$ are finite sets of points with coordinates in $\overline{\bK(t)}$.
In this form, they are not amenable to manipulation by arithmetic circuits and Turing machines.
Define $\bA \coloneqq \bQ\bs{\vu_{0}, \dots, \vu_{n}}\bsd{t}$.
Consider the $h_{i, j}$ as polynomials in the ring $\bA\bs{x_{0}, \dots, x_{j-1}}$.
The above observations, combined with the discussion in \cref{subsection: initial system}, allows us to apply Hensel's lemma to this system of equations.

\begin{lemma}
  \label{lemma: power series roots}
  Let $j \in [n]$.
  As elements of $\bA\bs{x_{0}, \dots, x_{j-1}}$, the set of polynomials $\vh_{j}$ admits $d_{0} \cdots d_{j-1}$ roots in $\bA$.
  For each $\vr_{\vc} \in B_{j}$ as defined in \cref{lemma: initial system roots}, there is a root $\vrho_{\vc} \in \bA^{j}$ of $\vh_j$ that satisfies $\vrho_{\vc} \pmod t = \vr_{\vc}$.
  The set of polynomials has no other roots in any extension of $\bA$.
\end{lemma}
\begin{proof}
  The first statement follows from a multivariate version of Hensel's lemma, see for example \cite[Exercise~7.26]{eisenbud2013commutative}.
  The final statement follows from Bézout's inequality.
\end{proof}

We will use an explicit version of the implicit function theorem to compute the $\vrho_{\vc}$ to the required precision.
To apply this theorem, we have to perform some shifting and scaling to ensure that the root $\vrho_{\vc}$ we are computing has no constant term, and that the Jacobian at the origin after shifting is the identity.

\begin{definition}
  \label{definition: shifted scaled system}
  Let $j \in [n]$ and $\vc \in B_j$.
  Define the shifted and scaled system $\tilde{\vh}_{\vc}$ as
  \[
    \tilde{\vh}_{\vc}(\vx) \coloneqq \cJ_{\vc}^{-1}\vh_{j}\br{\vx + \vr_{\vc}}.
  \]
  This system has $\vrho_{\vc} - \vr_{\vc}$ as a root, which has a vanishing constant term in every coordinate.
  The Jacobian at zero of this system is the identity matrix.
  Let $\Delta_{\vc}$ denote the Jacobian determinant of this system.
  Note that we hide $j$ in the notation $\tilde{\vh}_{\vc}$ for brevity, as it will always be clear from context.
\end{definition}

Each polynomial in $\tilde{\vh}_{\vc}$ lies in $S\bs{t, x_{0}, \dots, x_{j-1}}$.
We now quote an explicit version of the implicit function theorem, specialized to the system described in \cref{definition: shifted scaled system}.

\begin{lemma}[Implicit Function Theorem {\cite[Proposition~20.3]{AiYu83}}]
  \label{lemma: explicit implicit function for roots}
  Let $\tilde{\vrho}_{\vc} = \br{\tilde{\rho}_{\vc,0}(t), \ldots, \tilde{\rho}_{\vc, j-1}(t)}$ be the root of $\tilde{\vh}_{\vc}$ with vanishing constant term in every coordinate.
  Write $\tilde{\rho}_{\vc, i}(t) = \sum_{N \ge 0} \tilde{\rho}_{\vc, i, N} t^N$ with coefficients $\tilde{\rho}_{\vc, i, N} \in S$.
  Then the coefficient $\tilde{\rho}_{\vc, i, N}$ is given by
  \[
    \tilde{\rho}_{\vc,i,N} =
    \sum_{\substack{(b_0,\dots,b_{j-1}) \in \bN^j \\ |\vb| \le 2N}} (-1)^{\abs{\vb}}\, {\rm coeff}(x_i(\tilde h_{\bm c,0}(\vx)-x_0)^{b_0}
    \cdots (\tilde h_{\bm c, j-1}(\vx)-x_{j-1})^{b_{j-1}} \Delta_{\vc}, t^N x_0^{b_0} \cdots x_{j-1}^{b_{j-1}}).
  \]
\end{lemma}
\begin{proof}
  The result we cite from~\cite[Proposition~20.3]{AiYu83} is complex-analytic, but deriving the formal statement above from it is straightforward.

  In what follows, $j$ and $\vc$ are fixed.
  \cref{lemma: power series roots} pointed out that the coefficients of $\tilde{\rho}_{\vc}$ are polynomials in $\vu_{0}, \dots, \vu_{n}$, that is, $\tilde{\rho}_{\vc,i,N}$ is in $\mathbb{Q}[\vu_{0}, \dots, \vu_{n}]$ for all $i,N$.
  We use $Q_{i,N}$ to denote the polynomials appearing as right-hand sides in the statement of \cref{lemma: explicit implicit function for roots}, so our claim is that $\tilde{\rho}_{\vc,i,N}=Q_{i,N}$ for all $i$ and all $N$ (these polynomials actually only involve $\vu_{0}, \dots, \vu_{j-1}$, but this has no bearing on the proof).

  Choose $\vnu_0,\dots,\vnu_{n}$ with entries in $\bC$ and let $\tilde{\vh}_{\vc,\vnu}$ be the polynomials in $\bC[t][x_0,\dots,x_{j-1}]$ obtained by evaluating $\vu$ at $\vnu$ in  $\tilde{\vh}_{\vc}$.
  Through this evaluation, we see that the unique formal power series root with vanishing constant term to $\tilde{\vh}_{\vc,\vnu}=0$ is $\tilde{\rho}_{\vc}(\vnu,t)$, whose coefficients are $\tilde{\rho}_{\vc,i,N}(\vnu)$, for $i=0,\dots,j-1$ and $N \ge 0$.
  On the other hand, applying \cite[Proposition~20.3]{AiYu83} to $\tilde{\vh}_{\vc,\vnu}$ tells us that the $i$th coordinate of the unique analytic root to these equations with vanishing constant term is
  \begin{align*}
    &
    \sum_{N \ge 0}
    \sum_{\substack{\abs{\vB} \in \bN^j \\ |\vb| \le 2N}} (-1)^{\abs{\vb}}\, {\rm coeff}(x_i(\tilde h_{\bm c,\vnu,0}(\vx)-x_0)^{b_0}
    \cdots (\tilde h_{\bm c,\vnu, j-1}(\vx)-x_{j-1})^{b_{j-1}} \Delta_{\vc,\vnu}, t^N x_0^{b_0} \cdots x_{j-1}^{b_{j-1}}) t^N \\
    &= \sum_{N \ge 0} Q_{i,N}(\vnu) t^N,
  \end{align*}
  where $\Delta_{\vc,\vnu}$ is the Jacobian determinant of $\tilde {\vh}_{\vc,\vnu}$ with respect to $x_0,\dots,x_{j-1}$  (we note that \cite{AiYu83} uses the notation $D^{\alpha, \beta}_{w, z}$ to denote the partial derivative $\partial^{\alpha}_{w} \partial^{\beta}_{z}$).

  Since an analytic root is also a root over the ring of formal power series, we deduce that $\tilde{\rho}_{\vc,i,N}(\vnu)=Q_{i,N}(\vnu)$ for all $i$, $N$, and $\vnu$, so that $\tilde{\rho}_{\vc,i,N}=Q_{i,N}$ for all $i$ and $N$.
\end{proof}

The root $\vrho_{\vc}$ of $\vh_{j}$ is a shift of $\tilde{\vrho}_{\vc}$ by $\vr_{\vc}$.
If $\rho_{\vc, i, N}$ are the coefficients of $t$ in the coordinates of $\vrho_{\vc}$ then $\rho_{\vc, i, N} = \tilde{\rho}_{\vc, i, N}$ for $N \geq 1$, and $\rho_{\vc, 0, N} = \vr_{\vc}$.
Using \cref{lemma: explicit implicit function for roots}, we can thus compute $\vrho_{\vc}$ up to any precision of our choice.
To compute the resultant, we will require the coefficients $\rho_{\vc, i, N}$ for $N$ up to $P \coloneqq \sum_{i=0}^{n} \prod_{j \neq i} d_{j}$, which is the degree of $\res_{n}\br{\vH_{n}}$ in the variable $t$ (\cref{lemma: resultant degree}).
We now describe the circuits carrying out the above computations.
That there exist constant-depth circuits that do the above computations follows directly from the above formula.
The part that takes some work is establishing that there exist uniform circuit families that implement the expression in \cref{lemma: explicit implicit function for roots}.

We start by showing that the shifted and scaled system, along with its Jacobian determinant, can be computed efficiently.
\begin{lemma} \label{lemma: uniformity of shifted initial system}
  There exists a polylogtime-uniform family of constant-free circuits $\cC^{\scaledSystem}$ indexed by $j, d_{0}, \dots, d_{j-1}$ with the following properties.
  \begin{itemize}
    \item 
    The circuit $C^{\scaledSystem}_{j, d_{0}, \dots, d_{j-1}}$ has size polynomial in $U_j$ and $\exp\br{\poly\br{j}}$, and depth bounded by a universal constant.
    \item
      The circuit $C^{\scaledSystem}_{j, d_{0}, \dots, d_{j-1}}$ has $U_{j} + 2j + 1$ input gates, labeled by variables $\vu_{0}, \dots, \vu_{j}$, $x_{0}, \dots, x_{j-1}$, $z_{0}, \dots, z_{j-1}, t$.
      The circuit has $j + 1$ output gates.
      The circuit is division-free with respect to $\vu_{0}, \dots, \vu_{j}, x_{0}, \dots, x_{j-1}, t$.
    \item
      When $z_{0}, \dots, z_{c-1}$ are specialized to $(c_{0}, \dots, c_{j-1}) \in B_{j}$, the circuit computes the polynomials $\tilde{\vh}_{\vc}\br{x_{0}, \dots, x_{j-1}, t}$ and the Jacobian determinant $\Delta_{\vc}\br{x_{0}, \dots, x_{j-1}, t}$.
  \end{itemize}
\end{lemma}

\begin{proof}
  The construction of $\cC^{\scaledSystem}$ requires us to define auxiliary circuit families $\cC'$ and $\cC''$.
  We start with the construction of $\cC'$.

  The family $\cC'$ is also indexed by $j, d_{0}, \dots, d_{j-1}$ and takes the same input variables as $\cC$.
  We describe $C' = C'_{j, d_{0}, \dots, d_{j-1}}$.
  It contains as a subcircuit a copy of the circuit $C^{\initialRoot}_{j, d_{0}, \dots, d_{j-1}}$, constructed in \cref{lemma: uniformity of initial system}, that takes as input the variables $\vz$.
  Let $\vr_{\vz}$, $\cJ_{\vz}$, and $\cJ^{-1}_{\vz}$ denote the outputs of this subcircuit.
  The circuit $C'$ also contains a subcircuit that computes the polynomials $\vh_{j}$.
  Computing the polynomials $\vh_{j}$ amounts to computing the polynomials $\vf_{j}$ and $\vg_j$, and can be done using the definition of $\vg_j$.
  The subcircuit that computes $\vh_j$ is given as input $\vx + \vr_{\vz}$, so it computes $\vh_j(\vx + \vr_{\vz})$.
  Next, $C'$ computes the matrix-vector product of $\cJ_{\vz}^{-1}$ and $\vh_{j}(\vx + \vr_{\vz})$.
  Denote the resulting output by $\tilde{\vh}_{\vz}$.
  Note that if we specialize $\vz$ to $\vc \in B_{j}$, then $\tilde{\vh}_{\vz}$ is exactly $\tilde{\vh}_{\vc}$.
  It is straightforward to see that the family $\cC'$ described above is polylogtime-uniform, using a simulation argument as in the proof of \cref{lemma: polynomial interpolation uniform}.
  The size of $C'$ is polynomial in $U_{j}$, and the degree in $\vx$ of each output is bounded by $d_{0} + \cdots + d_{j-1}$.
  The depth of $C'$ is bounded by a universal constant.
  Further, the denominators in all division gates are free of any variable in $\vx$.

  We now construct $\cC''$.
  The family $\cC''$ is also indexed by $j, d_{0}, \dots, d_{j-1}$ and takes the same input variables as $\cC$.
  It is obtained by applying \cref{lemma: one at a time polynomial interpolation uniform} to the family $\cC'$, with distinguished variables $x_{0}, \dots, x_{j-1}$.
  The statements at the end of the previous paragraph show that $\cC'$ meets all the required assumptions for this interpolation result to apply.
  The resulting circuit $C'' = C''_{j, d_{0}, \dots, d_{j-1}}$ has size polynomial in $U_{j}$, and computes the coefficients of every output of $C'_{j, d_{0}, \dots, d_{j-1}}$ in each of the variables $x_{0}, \dots, x_{j-1}$.
  This circuit family is also polylogtime-uniform and consists of constant-depth circuits.

  We now describe the circuit $C = C^{\scaledSystem}_{n, j, d_{0}, \dots, d_{j-1}}$.
  The circuit $C$ contains a copy of the circuit $C' = C'_{j, d_{0}, \dots, d_{j-1}}$, whose outputs are also outputs of $C$.
  These gates will be named $(1, v)$ where $v$ is the name of the gate within $C'$.
  Next $C$ contains a copy of $C'' = C''_{j, d_{0}, \dots, d_{j-1}}$.
  These gates will be named $(2, v)$ where $v$ is the name of the gate within the circuit obtained from interpolation.

  Using the coefficients of $x_{k}$ in $\tilde{\vh}_{\vz}$ computed in the subcircuit $C''$, the circuit $C$ then computes the entries of the Jacobian matrix of $\tilde{\vh}_{\vz}$.
  The partial derivatives with respect to $x_{k}$ are just linear combinations, weighted by powers of $x_k$, of the coefficients of $\tilde{\vh}_{\vz}$ viewed as a polynomial in $x_{k}$.
  Finally, the circuit $C$ computes the determinant of this Jacobian using the trivial depth-two circuit for determinants of $j \times j$ matrices, that is, by expanding it out into $j!$ summands.
  The gates for these computations will be named $(3, *)$.

  We now bound the size of $C$.
  The size of the subcircuits $C', C''$, and the circuitry to compute the partial derivatives can be bounded by $\poly\br{U_{j}}$.
  To compute the determinant, we use a depth-two circuit that has size polynomial in $j!$, which can be bounded by $\exp\br{\poly\br{j}}$.
  The size of $C$ is therefore polynomial in $U_{j}$ and $\exp\br{\poly\br{j}}$ as claimed.
  The proof of uniformity is a simulation argument, just as the previous proofs.
  As before, our naming scheme allows us to deduce the name of a gate within each of the subcircuits, which can be used to simulate the Turing machine that decides the direct connection language of each subcircuit whenever required.
\end{proof}

Using \cref{lemma: uniformity of shifted initial system}, we now show that the truncations of the power series roots can be computed by uniform constant-depth circuits.

\begin{lemma} \label{lemma: uniformity of explicit implicit function theorem}
  There exists a polylogtime-uniform family of constant-free circuits $\cC^{\powerseriesRoot}$ indexed by parameters $n, j, d_{0}, \dots, d_{n}$ with the following properties.
  \begin{itemize}
    \item 
      The circuit $C^{\powerseriesRoot}_{n, j, d_{0}, \dots, d_{n}}$ has size bounded by $P^{\poly\br{n}}$, where $P \coloneqq \sum_{i=0}^{n} \prod_{j \neq i} d_{j}$, and depth bounded by a universal constant.
    \item
      The circuit $C^{\powerseriesRoot}_{n, j, d_{0}, \dots, d_{n}}$ has $U_{j} + j+1$ input gates, labeled by variables $\vu_{0}, \dots, \vu_{j-1}, z_{0}, \dots, z_{j-1}, t$.
      The circuit has $j$ output gates.
      The circuit is division-free with respect to $\vu_{0}, \dots, \vu_{j}, t$.
    \item
      When the inputs $z_{0}, \dots, z_{j-1}$ are specialized to $c_{0}, \dots, c_{j-1}$, the circuit computes the truncated roots $\sum_{N \leq P} \rho_{\vc, i, N} t^{N}$ for each $i \leq j-1$.
  \end{itemize}
\end{lemma}
\begin{proof}

  We start by defining an auxiliary circuit family $\cC'$.
  This family will be indexed by the parameters $n, j, d_{0}, \dots, d_{n}, k, b_{0}, \dots, b_{j-1}$.
  The input variables will be $\vu_{0}, \dots, \vu_{j}, x_{0}, \dots, x_{j-1}, z_{0}, \dots, z_{j-1}, t$.

  If $k \geq j$ or if $b_{0} + \cdots + b_{j-1} > 2P$, then the corresponding circuit in $\cC'$ is just the empty circuit.
  These conditions can be checked in time that is polynomial in the description of the index.
  We describe the circuit in the case where the indices satisfy $k < j$ and $b_{0} + \cdots b_{j-1} \leq 2 P$.
  
  The circuit $C'$ contains a copy of $C^{\scaledSystem}_{j, d_{0}, \dots, d_{j-1}}$ from \cref{lemma: uniformity of shifted initial system}.
  This computes $\tilde{\vh}_{\vz}$ and its Jacobian determinant $\Delta_{\vz}$.
  As usual, these gates are named $(1, v)$, where $v$ is the name of the corresponding gate within $C^{\textrm{scaled}}_{j, d_0, \ldots, d_{j-1}}$.
  The circuit then has $b_{0}$ copies of a sum gate, each of which computes $\tilde{h}_{\vz, 0}(\vx) - x_{0}$.
  These are named $(2, 0, i)$ for $i \leq b_{0}$.
  Similarly, for each $e < j$, there are $b_{e}$ copies of $\tilde{h}_{\vz, t}(\vx) - x_{e}$ computed by sum gates named $(2, e, i)$ for $i \leq b_{e}$.
  For each $e < j$, there is a product gate $(3, e)$ that multiplies these copies to compute $\br{h_{\vz, e}(\vx) - x_{e}}^{b_{e}}$.
  Finally, there is a multiplication gate that has inputs $(3, e)$ for all $e$, the variable $x_{k}$, and the Jacobian determinant $\Delta_{\vz}$ that was computed by the copy of $C^{\scaledSystem}_{j, d_{0}, \dots, d_{j-1}}$.
  The circuit therefore computes the polynomial
  \[
    x_{k}(\tilde h_{\vz ,0}(\vx)-x_0)^{b_0} \cdots (\tilde h_{\vz, j-1}(\vx)-x_{j-1})^{b_{j-1}} \Delta_{\vz}.
  \]
  Uniformity of this circuit family follows from a simulation argument, since the gate names allow us to recover the gate names within the copy of $C^{\scaledSystem}_{j, d_{0}, \dots, d_{j-1}}$.
  No division gate involves division by a polynomial that depends on $\vx$ or $t$.
  The size of the circuit $C'$ is polynomial in $P$ (note that $U_{j}$ is itself polynomial in $P$).
  The degree in $\vx$ and $t$ of the polynomial computed by $C'$ can likewise be bounded by $\poly\br{P}$.

  Using the family $\cC'$, we define another auxiliary circuit family $\cC''$ with the same indices.
  The family $\cC''$ is obtained by applying multivariate interpolation (\cref{lemma: multivariate polynomial interpolation uniform}) to $\cC'$ to interpolate out the coefficients of $t, x_{0}, \dots, x_{j-1}$.
  The statements at the end of the previous paragraph show that the assumptions required for the interpolation hold.
  To summarize the construction so far, we have a polylogtime-uniform family $\cC''$ such that $C'' = C''_{n, j, d_{0}, \dots, d_{n}, k, b_{0}, \dots, b_{j-1}}$ has input variables $\vu_{0}, \dots, \vu_{j-1}, z_{0}, \dots, z_{j-1}$, and computes the coefficients of 
  \[
    x_{k}(\tilde h_{\vz ,0}(\vx)-x_0)^{b_0} \cdots (\tilde h_{\vz, j-1}(\vx)-x_{j-1})^{b_{j-1}} \Delta_{\vz}
  \]
  as a polynomial in $t, x_{0}, \dots, x_{j-1}$.
  The size of $C''$ is bounded by $P^{\poly\br{n}}$, and the depth is a universal constant.

  With this family in hand, we return to the construction of $C^{\powerseriesRoot}_{n, j, d_{0}, \dots, d_{n}}$.
  Fix indices $n, j, d_{0}, \dots, d_{n}$.
  We describe the circuit $C = C^{\powerseriesRoot}_{n, j, d_{0}, \dots, d_{n}}$.
  For every vector $\vb \in \bN^{j}$ with $\abs{\vb} \leq 2P$ and for every $k < j$, the circuit $C$ contains a copy of $C''_{n, j, d_{0}, \dots, d_{n}, k, b_{0}, \dots, b_{j-1}}$.
  We will call this subcircuit $C_{\vb, k}$ in the rest of this argument for brevity.
  The names of each gate in $C_{\vb, k}$ are of the form $(1, \vb, k, v)$ where $v$ is the name of the corresponding gate within the subcircuit.
  For each $\vb, k$, and each $N \leq P$, the circuit $C$ has a product gate that computes
  \[
    (-1)^{\abs{\vb}} \cdot {\rm coeff}(x_i(\tilde h_{\vz,0}-x_0)^{b_0} \cdots (\tilde h_{\vz, j-1}-x_{j-1})^{b_{j-1}} \Delta_{\vz}, t^N x_0^{b_0} \cdots x_{j-1}^{b_{j-1}})
  \]
  by taking as input a constant gate with constant $-1$ and the coefficient from the copy of $C_{\vb, k}$.
  These product gates are named $(2, \vb, k, N)$.
  For each $1 \leq N \leq P$ and each $k < j$, the circuit $C$ contains a sum gate that uses the above gates to compute
  \[
    \rho_{\vz,k,N} \coloneqq
    \sum_{(b_0,\dots,b_{j-1}),\ \abs{\vb} \le 2N} (-1)^{\abs{\vb}} \cdot {\rm coeff}(x_k(\tilde h_{\vz,0}-x_0)^{b_0}
    \cdots (\tilde h_{\vz, j-1}-x_{j-1})^{b_{j-1}} \Delta_{\vz}, t^N x_0^{b_0} \cdots x_{j-1}^{b_{j-1}}).
  \]
  These gates are named $(3, N, k)$.
  The predecessors of the summation gates are numbered in lexicographic order based on $\vb$.
  All that remains is to compute the constant terms $\rho_{\vz, k, 0}$, which are simply $\vr_{\vz}$.
  For this, the circuit $C$ contains a copy of $C^{\initialRoot}_{j, d_{0}, \dots, d_{j-1}}$, whose gates are named $(4, v)$ where $v$ is the name of the corresponding gate within the subcircuit.
  With all the coefficients $\rho_{\vz, i, N}$ computed as above, the circuit has product gates that compute $\rho_{\vz, k, N} t^{N}$ for each $N \leq P$ and $k \leq j-1$, and addition gates that add these to compute $\sum_{N \leq P} \rho_{\vz, k, N} t^{N}$ for each $k \leq j-1$.
  These gates have names of the form $(5, *)$.

  We now bound the size of the circuit.
  The total number of vectors $\vb$ with $\abs{\vb} \leq 2P$ is bounded by $P^{\poly\br{n}}$.
  The size of each copy $C_{\vb, k}$ is bounded by $P^{\poly\br{n}}$, and consequently the size of the whole circuit is bounded by $P^{\poly\br{n}}$.
  The fact that the depth is a universal constant is straightforward.

  We now argue that the circuit family is uniform.
  Let $T_{C}$ be the name of the machine that we will construct to decide the direct connection language of $\cC^{\powerseriesRoot}$.
  Let $T_{C'}$ and $T_{C'"}$ be the machines that decide the direct connection languages of $\cC'$ and $\cC''$, respectively.
  Suppose $(n, j, d_{0}, \dots, d_{n}, a, p, a')$ is an input to $T_{C}$.
  The bound $P$ can be computed from the index in polynomial time.
  Suppose $a$ has a name of the form $(1, \vb, k, v)$.
  The key observation is that the machine $T_{C}$ can simulate the machine $T_{C''}$ on inputs of the form $(n, j, d_{0}, \dots, d_{n}, k, b_{0}, \dots, b_{j-1}, v, p, a'')$, where $a''$ is either a gate type, or the name of a gate within $C_{\vb, k}$.
  This will allow $T_{C}$ to use $T_{C'}$ to decide if $(n, j, d_{0}, \dots, d_{n}, a, p, a')$ is a YES instance.
  As before, if $v$ is an input or output gate of the subcircuit, then $T_{C}$ uses the description in the construction above to check if the gate is wired correctly.
  The rest of the argument, and the argument when $a$ is of the form $(2, *), (3, *), (4, *)$ or $(5, *)$ is the usual simulation argument that we have repeated before.
\end{proof}

With this, we can compute $\res_{n}\br{\vF_{n}}$ via a polylogtime-uniform family of constant-depth circuits.

\begin{lemma} \label{lemma: uniformity of homotopy resultant}
  There exists a polylogtime-uniform family of weakly division-free, constant-free circuits $\cC^{\res}$ indexed by parameters $n, d_{0}, \dots, d_{n}$ with the following properties.
  \begin{itemize}
    \item 
      The circuit $C^{\res}_{n, d_{0}, \dots, d_{n}}$ has size bounded by $P^{\poly\br{n}}$, where $P = \sum_{i=0}^{n} \prod_{j \neq i} d_{j}$, and depth bounded by a universal constant.
    \item
      The circuit $C^{\res}_{n, d_{0}, \dots, d_{n}}$ has $U \coloneqq \sum_{i=0}^{n} \binom{n + d_{i}}{n}$ input gates, labeled by $\vu_{0} \dots, \vu_{n}$.
      The circuit has a single output gate.
    \item
      The circuit computes the multivariate resultant $\res_{n}\br{\vF_{n}}$ in the variables $\vu_{0}, \dots, \vu_{n}$.
  \end{itemize}
\end{lemma}
\begin{proof}
  We start by constructing an auxiliary family $\cC'$ indexed by $n, d_{0}, \dots, d_{n}$.
  Fix indices $n, d_{0}, \dots, d_{n}$.
  The inputs to $C' = C'_{n, d_{0}, \dots, d_{n}}$ will be variables $\vu_{0}, \dots, \vu_{n}, t$.
  
  We now describe $C'$.
  For each $j \in \bs{n}$, the circuit $C'$ has $d_{0} \cdots d_{j-1}$ copies of the circuit $C^{\powerseriesRoot}_{n, j, d_{0}, \dots, d_{n}}$ from \cref{lemma: uniformity of explicit implicit function theorem} as sub-circuits.
  The names of the gates are $(1, j, i, v)$, where $i \leq d_{0} \cdots d_{j-1}$ and $v$ is the name of the corresponding gate within the subcircuit.
  In addition, the circuit $C'$ has gates that compute each of the integers $1, \ldots, \sum_{i=1}^{n} d_{i}$.

  The elements of $B_{j} = \bs{d_{0}} \times \cdots \times \bs{d_{j-1}}$ are tuples of integers of length $j$, and $B_{j}$ itself has size $d_{0} \cdots d_{j-1}$.
  Given $d_{0}, \dots, d_{n}, j, n$, and $i$ in binary, computing the $i\ts{th}$ element of $B_{j}$ in lexicographic order can be done in polynomial time.
  In $C'$, the $i\ts{th}$ copy of $C^{\powerseriesRoot}_{n, j, d_{0}, \dots, d_{n}}$ is evaluated at the $i\ts{th}$ element of $B_{j}$ using the integers $1, \dots, \sum d_{i}$ computed above.
  By \cref{lemma: uniformity of explicit implicit function theorem}, each of the output gates of the copies computes approximations (in $t$) of the roots of the system $\vh_{j}$.

  The circuit $C'$ also has copies of a subcircuit that computes the polynomials $h_{1, 1}, \dots, h_{n, n}$, as in the proof of \cref{lemma: uniformity of shifted initial system}.
  Specifically, it has one copy of $h_{j, j}$ for each element of $B_{j}$.
  To the input gates of the $i\ts{th}$ copy of $h_{j, j}$, we wire the approximate roots constructed using the $i\ts{th}$ element of $B_{j}$.
  Consequently, for each $\vc \in B_{j}$, the circuit computes $h_{j, j}(\vrho_{\vc}')$, where $\vrho_{\vc}' = \vrho_{\vc} \pmod{t^{P}}$.
  Finally, the circuit $C'$ has gates that use these evaluations to implement the Poisson formula (\cref{lemma: poisson formula homotopy}).

  To summarize, the circuit $C'_{n, d_{0}, \dots, d_{n}}$ implements the Poisson formula for $\res_{n}\br{\vH_{n}}$, except it uses approximations of the roots of $\vh_{j}$ in the ring $\bA\bsd{t}$.
  The order of approximation is $P = \sum_{i=0}^{n} \prod_{j \neq i} d_{j}$, and the output has degree $P^{\poly\br{n}}$.
  The fact that $C'$ has size $P^{\poly\br{n}}$, constant depth, and the family $\cC'$ is polylogtime-uniform follows by the same argument as in previous constructions.
  No division gate in $C'$ involves division by a polynomial that depends on $t$.
  Further, the degree in $t$ of the polynomial computed by $C'$ is bounded by $P^{\poly\br{n}}$.

  Applying \cref{lemma: polynomial interpolation uniform} to the family $\cC'$ with distinguished variable $t$ results in a circuit family $\cC''$ that computes the coefficients in $t$ of the above approximation to $\res_{n}\br{\vH_{n}}$.
  The family $\cC''$ is polylogtime-uniform and consists of circuits of size $P^{\poly\br{n}}$.

  We now return to the construction of $C^{\res}_{n, d_{0}, \dots, d_{n}}$.
  This circuit has a copy of $C''_{n, d_{0}, \dots, d_{n}}$.
  It has one additional summation gate that sums the outputs of $C''_{n, d_{0}, \dots, d_{n}}$ corresponding to the coefficients of $t^{i}$ for $i \leq P$.
  The computed polynomial is the same as the one obtained by truncating the approximate computation of the resultant and evaluating at $t = 1$.
  By the Poisson formula, the degree bounds on the resultant, and the construction of the homotopy, this polynomial is exactly $\res_{n}\br{\vF_{n}}$.
  The claims on the size, depth, and uniformity are again straightforward.
\end{proof}

In \cref{subsection: hn to resultant,subsection: counting to resultant}, we saw examples of how the resultant can be used to solve problems beyond satisfiability of square homogeneous systems.
In these applications, we often have to compute the resultant of a set of polynomials whose coefficients are themselves polynomials.
If these coefficients come from the ring $\bQ\bs{w_{1}, \dots, w_{k}}$, then the resultant is itself a polynomial in $\bQ\bs{w_{1}, \dots, w_{k}}$, and can be computed by simply substituting these polynomial coefficients in the circuit for $\res_{n}\br{\vF_{n}}$ designed above.

If the coefficients are polynomials in $w_{1}, \dots, w_{k}$ of degree at most $D$, then the resulting computation can also be seen as taking $U \cdot \binom{k + D}{k}$ inputs, one corresponding to each monomial in the $\vw$ variables in each coefficient of $\vF_{n}$, and having at most $\br{D n d^{n}}^{O(k)}$ outputs, where $d \coloneqq \max_i d_{i}$.
The outputs correspond to the coefficients of the $\vw$ variables in the computed resultant.
The following corollary states that this computation can be carried out by polylogtime-uniform circuits of constant depth.
It is an easy corollary of the above resultant computation and the uniformity of interpolation.

\begin{corollary} \label{corollary uniformity of composed resultant}
  There exists a polylogtime-uniform family of weakly division-free, constant-free circuits $\cC^{\composedResultant}$ indexed by parameters $n, d_{0}, \dots, d_{n}, k, D$ with the following properties.
  \begin{itemize}
    \item 
      The circuit $C^{\composedResultant}_{n, d_{0}, \dots, d_{n}, k, D}$ has size bounded by $\br{PD}^{\poly\br{n, k}}$, where $P = \sum_{i=0}^{n} \prod_{j \neq i} d_{j}$, and depth bounded by a universal constant.
    \item
      The circuit has $U \cdot \binom{k + D}{k}$ input variables, which we denote by $u_{i, \valpha, \vbeta}$ (where $\vbeta \in \bN^{k}$ satisfies $\abs{\vbeta} \leq D$).
      The circuit has $\binom{D D' + k}{k}$ outputs, where $D'$ is the degree of $\res_{n}\br{\vF_{n}}$.
    \item
      When evaluated at $\gamma_{i, \valpha, \vbeta}$, the output gates compute the coefficients of the monomials in $\vw$ of the resultant of the polynomials
      \[
        H_{i} = \sum_{\valpha} \sum_{\vbeta} \gamma_{i, \valpha, \vbeta} \vx^{\valpha} \vw^{\vbeta}
      \]
      with respect to the variables $\vx$.
  \end{itemize}
\end{corollary}
\begin{proof}
  There exist polylogtime-uniform constant-depth circuits that compute the polynomials $\sum_{\vbeta} u_{i, \valpha, \vbeta} \vw^{\vbeta}$ for each $i$ and $\valpha$.
  We then compose these with the circuit for $\res_{n}\br{\vF_{n}}$ obtained from \cref{lemma: uniformity of homotopy resultant}.
  The resulting circuit is polylogtime-uniform and computes a polynomial in the variables $\vw$ and $u_{i, \valpha, \vbeta}$.
  The degree in the $\vw$ variables is at most $D D'$.
  We use \cref{lemma: multivariate polynomial interpolation uniform} to interpolate out the coefficients of monomials in $\vw$.
  This resulting circuit satisfies the depth, size, and uniformity requirements that we claim.
\end{proof}

%% file: sections/circuits_to_counting.tex
In this section, we prove that Hilbert's Nullstellensatz can be decided in the counting hierarchy.
To do this, we prove a general transfer theorem that translates uniform constructions of arithmetic circuits into algorithms that can be run on Turing machines.

To every family of polynomials $\cF = (f_\vn)_{\vn}$ with rational coefficients, there is an associated family of boolean functions $\hat{\cF} = (\hat{f}_{\vn, h})_{\vn, h}$, where $\hat{f}_{\vn, h}$ corresponds to evaluation of $f_{\vn}$ on tuples of rational numbers represented by pairs of integers of height at most $h$.
Our goal will be to show that if the family of polynomials $\cF$ can be computed by polylogtime-uniform arithmetic circuits of constant depth and exponential size, then this associated family of boolean functions can be computed in the counting hierarchy $\CH$.
Once we have this in hand, an immediate corollary of \cref{lemma: uniformity of homotopy resultant} will be an $\FP^\CH$ algorithm to evaluate the multivariate resultant.
Applying the reductions of \cref{lemma: hn to resultant,lemma: hn counting to resultant} will yield $\CH$ and $\FP^\CH$ algorithms for the decision and counting versions of Hilbert's Nullstellensatz, respectively.
\roberttodo[inline]{I cut a lot of material from this introduction to section 6. Some of the details are not precise, but I think the current version communicates the overall idea clearly without getting bogged down in the details.}

\subsection{From uniform arithmetic circuits to the counting hierarchy} 

Our proof of the transfer theorem proceeds in two steps.
We first use the polylogtime-uniform arithmetic circuits for the family $\cF$ to construct polylogtime-uniform threshold circuits of similar size and depth for the boolean functions $\hat{f}_{\vn, h}$.
The fact that arithmetic circuits can be simulated efficiently by threshold circuits is a straightforward consequence of the fact that iterated addition and iterated multiplication are in logtime-uniform $\TC^0$ (items 1 and 2 of \cref{theorem: integer ops in tc0}, respectively).
The uniformity of the resulting threshold circuits is a consequence of the logtime-uniformity of iterated addition and multiplication combined with the assumed uniformity of the arithmetic circuits that compute $\cF$.

We then show that this family of threshold circuits for $\hat{f}_{\vn, h}$ can be evaluated in $\CH$.
Even though these threshold circuits are exponentially large, the fact that they are uniform means that a polynomial-time Turing machine can use nested majority quantifiers to simulate them.
The number of majority quantifiers used in this simulation corresponds to the depth of the threshold circuit that computes $\hat{f}_{\vn, h}$.
Since we build threshold circuits of constant depth for $\hat{f}_{\vn, h}$, the resulting evaluation procedure uses a bounded number of majority quantifiers, and hence lies in the counting hierarchy.

Neither of these ideas are fundamentally new to our work.
\textcite{AAD00} observed that arithmetic circuits can be transformed into threshold circuits that evaluate the corresponding polynomial on $\bc{0,1}$-valued inputs, and that this transformation preserves the size and depth of the circuit.
\textcite{AKRRV01} proved a result similar to ours, showing that exponentially-large constant-depth arithmetic circuits can be converted to $\CH$ algorithms.
This second result uses threshold circuits as an intermediate representation in the same manner as our work.

Although these ideas are present in prior work, we are not aware of any reference that starts with a uniform family of arithmetic circuits for a polynomial family $\cF$ and concludes that the corresponding boolean functions can be computed in $\CH$.
Because of this, we provide complete details for the results of this section.

We start by converting arithmetic circuits into threshold circuits by implementing each arithmetic gate as a $\TC^0$ sub-circuit using \cref{theorem: integer ops in tc0}.
Because the circuits provided by \cref{theorem: integer ops in tc0} are uniform, this procedure preserves the uniformity present in the family of arithmetic circuits.
The families of boolean functions we construct throughout this subsection treat their inputs as sequences of rational numbers, represented as pairs of integers.
To make these functions total rather than partial, we adopt the convention that if any such pair of integers is of the form $(a, 0)$, the boolean function evaluates to zero.
In what follows, we tacitly assume that inputs are well-formed, since this can be decided in polynomial time as a preprocessing step.

\begin{lemma} \label{lemma: arithmetic to threshold}
  Let $\mathcal{C} = (C_\vn)_{\vn}$ be a polylogtime-uniform family of weakly division-free, constant-free arithmetic circuits over $\bQ$ of size $s_{\vn}$ and depth $\Delta \in \bN$.
  Let $f_{\vn} \in \bQ[x_1, \ldots, x_{m_{\vn}}]$ be the polynomial computed by $C_\vn$.
  There exists a family of boolean functions $\hat{f}_{\vn, h}$ indexed by $\vn$ and a natural number $h$ with the following properties.
  \begin{itemize}
    \item For each $\vn$ and $h$, the function $\hat{f}_{\vn, h}$ has $2 m_{\vn} \br{h+1}$ inputs, which are interpreted as $m_{\vn}$ rational numbers represented by a pair of integers of height $h$.
    Further, the function has $2 h \cdot \br{s_{\vn} \log s_{\vn}}^{\Delta} + 2$ outputs, which are interpreted as a single rational number represented by a pair of integers of height $h \cdot \br{s_{\vn} \log s_{\vn}}^{\Delta}$.
  \item If $(a_{1}, b_{1}), \dots, (a_{m_{\vn}}, b_{m_{\vn}})$ are representations of rational numbers of height at most $h$, then $\hat{f}_{\vn, h}\br{(a_{1}, b_{1}), \dots, (a_{m_{\vn}}, b_{m_{\vn}})}$ is a representation of $f_{\vn}(a_{1} / b_{1}, \dots, a_{m_{\vn}} / b_{m_{\vn}})$.
  \item There exists a polylogtime-uniform family $\cD = (D_{\vn, h})_{\vn, h}$ of threshold circuits of depth $O(\Delta)$ and size $\poly\br{h, \br{s_{\vn} \log s_{\vn}}^{\Delta}}$ that compute the functions $\hat{f}_{\vn, h}$.
  \end{itemize}
\end{lemma}

\begin{proof}
  We will construct circuits for $\hat{f}_{\vn, h}$ from the circuit $C_{\vn}$ by replacing each arithmetic gate with a threshold circuit that implements iterated addition and iterated multiplication of rational numbers, using items 1 and 2 of \cref{theorem: integer ops in tc0}.
  To do this, we first need to bound the heights of the integers used to represent intermediate values in $C_{\vn}$, as this will determine the size of the threshold circuits we need to correctly simulate the operations done in $C_{\vn}$.
  For each gate $v$ in the circuit $C_{\vn}$, we will compute a rational number $a_v / b_v$ that represents the value of the gate on the given input.
  We set $C = C_{\vn}$ and $s = s_{\vn}$ in the next part of this discussion for brevity.

  For a gate $v$, the \emph{depth} of $v$ is the length of the longest path starting from $v$ and ending at either an input gate or a constant gate.
  For example, input gates and constant gates have depth $0$, while a gate that only has constant gates as input has depth $1$.
  We claim that if $v$ has depth $\delta$ in $C$, then whenever the inputs to $C$ are represented by pairs of integers of height at most $h$, the values of $a_{v}$ and $b_{v}$ can be represented by integers of height bounded by $h \cdot (s \log s)^\delta$.
  We prove this by induction on $\delta$.
  When $\delta = 0$, the gate $v$ is either an input to the circuit or one of the constants $0, +1, -1$.
  If $v$ is a constant, then clearly the heights of $a_v$ and $b_v$ are at most $1$.
  If $v$ is an input, then by assumption, the numerator $a_v$ and denominator $b_v$ both have height at most $h$.

  When $\delta \ge 1$, we proceed by case analysis depending on the operation labeling the gate $v$.
  \begin{itemize}
    \item 
      Suppose $v = v_1 + \cdots + v_k$ is an addition gate.
      Because the circuit $C$ has size bounded by $s$, we know $k \le s$.
      Each $v_i$ is represented by a quotient $a_i / b_i$ where, by induction, the integers $a_i$ and $b_i$ have height at most $ h \cdot (s \log s)^{\delta - 1}$.
      We represent the value of $v$ as $a/b$, where
      \begin{align*}
        a &\coloneqq \br{\frac{a_1}{b_1} + \cdots + \frac{a_k}{b_k}} b_1 \cdots b_k = a_1 b_2 \cdots b_k + b_1 a_2 b_3 \cdots b_k + \ \cdots\  + b_1 \cdots b_{k-1} a_k \\
        b &\coloneqq b_1 \cdots b_k.
      \end{align*}
      Every product appearing in the definitions of $a$ and $b$ above is a product of $k$ numbers, each of which have height $h \cdot (s \log s)^{\delta - 1}$.
      This means that each such product---and in particular, the integer $b$---has height $k h \cdot (s \log s)^{\delta - 1}$.
      The integer $a$ is a sum of $k$ such terms, so has height $k \log (k) h \cdot (s \log s)^{\delta - 1}$.
      The fact that $k \le s$ implies that $a$ and $b$ have height at most $h \cdot (s \log s)^{\delta}$, as claimed.
    \item 
      Suppose $v = v_1 \times \cdots \times v_k$ is a product gate.
      As in the previous case, we know $k \le s$.
      Each $v_i$ is represented by a quotient $a_i / b_i$ where, by induction, the integers $a_i$ and $b_i$ have height at most $h \cdot (s \log s)^{\delta - 1}$.
      We represent the value of $v$ as $a / b$, where
      \begin{align*}
        a &\coloneqq a_1 \cdots a_k \\
        b &\coloneqq b_1 \cdots b_k.
      \end{align*}
      The same analysis as in the case where $v$ is an addition gate shows that $a$ and $b$ have height at most $h \cdot (s \log s)^{\delta}$.
    \item 
      Finally, suppose $v = v_1 / v_2$ is a division gate.
      The values of $v_1$ and $v_2$ are represented by the quotients $a_1 / b_1$ and $a_2 / b_2$, respectively.
      By induction, each $a_i$ and $b_i$ has height at most $h \cdot (s \log s)^{\delta - 1}$.
      We represent $v$ as the quotient $a/b$, where
      \begin{align*}
        a &\coloneqq a_1 b_2 \\
        b &\coloneqq a_2 b_1.
      \end{align*}
      It is clear that both $a$ and $b$ have height at most $2 h \cdot (s \log s)^{\delta - 1} \le h \cdot (s \log s)^\delta$, as claimed.
  \end{itemize}

  To construct the circuit for the boolean function $\hat{f}_{\vn, h}$, we replace each of the arithmetic gates in $C$ with a threshold circuit that implements the corresponding operation, using the $\TC^0$ circuits provided by \cref{theorem: integer ops in tc0}.
  We also replace each of the constant gates by a set of $1$ and $0$ gates that encode representations of the respective constant.
  As the analysis above shows, it suffices to use threshold circuits that implement iterated addition and iterated multiplication of at most $s$ integers of height $h \cdot (s \log s)^{\Delta}$, since this is the maximum height required to represent the value of any gate in the circuit when evaluated on an input of height at most $h$.
  These operations can be performed with threshold circuits of constant depth and size $h^{O(1)} (s \log s)^{O(\Delta)}$.
  The fact that the arithmetic circuit $C$ is weakly division-free ensures that as long as all the inputs are well-formed, no division by zero takes place when evaluating the circuit in the above manner.

  Formally, the threshold circuits provided by \cref{theorem: integer ops in tc0} perform iterated addition or multiplication of $h \cdot (s \log s)^{\Delta}$ integers of height $h \cdot (s \log s)^{\Delta}$, even though gates in $C$ have fan-in bounded by $s$.
  To deal with this, we hardwire the superfluous inputs of these threshold subcircuits to the constant $0$ for iterated addition and the constant $1$ for iterated multiplication.
  Similarly, these threshold circuits will have more than $h \cdot (s \log s)^{\Delta}$ bits of output.
  We ignore these extra bits of output, since we know they will not affect the result of our simulation.

  Since the circuit $C$ has size $s$ and depth $\Delta$, it is clear that this procedure results in a threshold circuit of size $h^{O(1)} (s \log s)^{O(\Delta)}$ and depth $O(\Delta)$ that correctly computes a binary representation of the output of the circuit $C$ on inputs of height $h$.
  We denote the resulting threshold circuit by $D_{\vn, h}$ and define the boolean function $\hat{f}_{\vn, h}$ to be the function computed by $D_{\vn, h}$.

  It remains to bound the uniformity of the circuit family $\cD = \br{D_{\vn, h}}_{\vn, h}$.
  The names of gates in $D_{\vn, h}$ that are part of threshold circuits simulating arithmetic will be a tuple $(1, v, u)$, where $v$ is the name of the gate in $C_\vn$ that is being simulated and $u$ is the name of a gate within the threshold subcircuit that simulates the arithmetic of $v$.
  The other gates in $D_{\vn, h}$, namely those that compute encodings of the integers $1$ and $0$ to be used in the extra inputs to the arithmetic subcircuits, have names of the form $(2, *)$.
  Let $T_{C}$, $T_+$, and $T_{\times}$ be Turing machines that decide the direct connection languages of $\cC$, threshold circuits for iterated addition, and threshold circuits for iterated multiplication, respectively.
  Denote by $T_{D}$ the Turing machine that we will design to decide the direct connection language of $\cD$.

  Let $(\vn, h, a, p, b)$ be an input to $T_{D}$.
  Suppose $a = (1, v_{a}, u_{a})$.
  If $p = \epsilon$, then $T_{D}$ must accept if $b$ represents the type of $a$.
  To decide this, the machine $T_{D}$ first computes the type of the gate $v_{a}$ using $T_{C}$.
  If $v_a$ is a summation gate, then $T_{D}$ simulates $T_+$ to compute the type of $u_{a}$ and accepts if $b$ matches this type.
  If $v_a$ is a product gate, then $T_D$ simulates either $T_+$ or $T_{\times}$, as appropriate, to compute the type of $u_a$ and accept or reject accordingly.
  If $v_a$ is a division gate, then $T_D$ simulates $T_{\times}$ to compute the type of $u_a$.
  If $v_{a}$ is an input gate of the subcircuit but $a$ is not an input gate of $D_{\vn, h}$, then we adopt the (arbitrary) convention that the type of $a$ is OR.

  If $p \neq \epsilon$, then the machine $T_D$ branches based on $b$.
  \begin{itemize}
    \item
      $T_{D}$ first checks if $b$ is of the form $(1, v_{b}, u_{b})$.
      In this case, it branches further as follows.
      \begin{itemize}
        \item
          Suppose $v_{a} = v_{b}$.
          In this case, $a$ and $b$ are part of the same subcircuit implementing arithmetic.
          The machine computes the type of the gate $v_{a}$ as in the above paragraph, and uses the corresponding Turing machine to check if $u_{b}$ is a predecessor of $u_{a}$.
          If so, it accepts, and if not then it rejects.
        \item
          Suppose $v_{a} \neq v_{b}$.
          In this case, the machine checks if $u_{a}$ is an input gate in the threshold subcircuit, and rejects if it is not.
          If it is, then $T_D$ computes integers $p_{1}$ and $p_{2}$ such that $u_{a}$ is the $p_{1}\ts{th}$ bit of the $p_{2}\ts{th}$ rational number input to the threshold subcircuit.
          It then computes the arity of $v_{a}$ and rejects if the arity is less than $p_{2}$.
          If the arity is more than $p_{2}$, the machine uses $T_{C}$ to check if $v_{b}$ is the $p_{2}\ts{th}$ input to $v_{a}$ in $C_{\vn}$.
          It also checks if $u_{b}$ is the $p_{1}\ts{th}$ output bit of its corresponding threshold subcircuit.
          If all of these checks pass, it accepts the string.
          If any of them fail, it rejects.
      \end{itemize}
      \item
        Suppose $b$ has name $(2, *)$.
        In this case, as above, the machine $T_D$ checks that $a$ is an input bit.
        It then computes integers $p_{1}$ and $p_{2}$ such that $u_{a}$ is the $p_{1}\ts{th}$ bit of the $p_{2}\ts{th}$ integer input to the threshold circuit.
        The machine $T_D$ verifies that the arity of $v_a$ is at most $p_2$.
        Finally, $T_D$ computes the type of $a$ and, based on whether $a$ is a summation or multiplication gate, checks that $b$ is the correct bit in the encoding of the $0$ or $1$ that is used for the extra inputs.
        If all these checks pass, $T_D$ accepts its input.

  \end{itemize}
  The case when $a = (2, *)$ can be handled in a similar manner.
  All computations take time polylogarithmic in the size of the final circuit and polynomial in the binary representation of the index.
  This shows that the circuit family $\cD$ is polylogtime-uniform.
\end{proof}

While \cref{lemma: arithmetic to threshold} is stated for single output circuit families, it can be easily extended to circuit families with multiple outputs.
In this case, the boolean function $f_{\vn, h}$ computes a tuple of binary representations of natural numbers, one for each output gate of $C_{\vn}$.
The proof itself requires no changes.

Next, we show that uniform threshold circuits of constant depth and exponential size can be evaluated in the counting hierarchy.
Of course, the number of output gates in such a circuit may itself be exponential, so it may not be possible to write down the complete output in polynomial time.
We instead show that given a multi-index $\vn$, an input to the $\vn\ts{th}$ circuit in the family, and an index $i$, we can compute the $i\ts{th}$ bit of the corresponding output in $\CH$, a problem we formalize below.

\begin{definition}
  \label{definition: language of function family}
  Let $\cF \coloneqq \br{f_{\vn} : \bc{0,1}^{m'_{\vn}} \to \bc{0,1}^{m_{\vn}}}_{\vn}$ be a multi-indexed family of boolean functions.
  The \emph{language $L_{\cF}$ corresponding to $\cF$} is defined as 
  \[
    L_{\cF} \coloneqq \setbuild{(i, b, x, \vn)}{f_{\vn}(x)_{i} = b}.
  \]
  If $\cD$ is a family of circuits, then we abuse notation and use $L_{\cD}$ to denote the language corresponding to the family of functions computed by $\cD$.
\end{definition}

Next, we show that if $\cD$ is a polylogtime-uniform family of exponentially-large threshold circuits, then the language $L_{\cD}$ corresponding to $\cD$ lies in the counting hierarchy.

\begin{lemma}[{see, e.g., \cite[Proof of Theorem 4.1]{ABKM09}}] \label{lemma: threshold to ch}
  Let $\cD = \br{D_{\vn}}_{\vn}$ be a polylogtime-uniform family of threshold circuits of size $s_{\vn}$ and depth $\Delta$, where $D_{\vn}$ has $m_{\vn}$ inputs.
  Let $L_{\cD}$ be the language corresponding to the circuit family $\cD$ as defined in \cref{definition: language of function family}.
  Suppose $s_{\vn} \le \exp(\poly(m_{\vn}))$.
  Then $L_{\cD} \in \CH$.
\end{lemma}

\begin{proof}
  For an integer $\delta \in \bN$, consider the auxiliary language
  \[
    L_{\delta} \coloneqq \setbuild{(g, b, x, \vn)}{\text{$g$ is a gate of depth at most $\delta$ in the circuit $D_{\vn}$ and evaluates to $b$ on input $x$}},
  \]
  where the gate name $g$ and index $\vn$ into the circuit family are provided in binary.
  We will prove by induction on $\delta$ that $L_{\delta} \in \CH$.
  
  As a first step, we show that given $(g, b, x, \vn)$ as input, we can check if $x$ has length $m_{\vn}$ in time polynomial in the length of the input.
  Recall that, by the definition of uniformity, there is a polynomially-bounded increasing function $T: \bN \to \bN$ and a Turing machine $M$ such that given input $\vn$, the machine $M$ runs in time at most $T(\log\br{s_{\vn}} + N)$ and computes $m_{\vn}$, where $N = \sum_{i} \log n_{i}$.
  To decide if $x$ has length $m_{\vn}$, we are only allowed time polynomial in $\abs{g} + \abs{x} + N$, rather than $\log(s_{\vn}) + N$.
  Let $Q$ be a polynomially-bounded increasing function such that $\log\br{s_{\vn}} \leq Q\br{m_{\vn}}$.
  To check if $x$ has the correct length, we simulate $M$ on input $\vn$ for at most $T\br{Q\br{\abs{x}} + N}$ steps.
  If this simulation terminates, we obtain $m_{\vn}$ and can check if $\abs{x} = m_{\vn}$.
  If this simulation does not terminate, then we know that $x$ is too short to be a valid input to $C_{\vn}$.
  When $x$ has the correct length, the fact that $\log\br{s_{\vn}} \leq Q\br{m_{\vn}}$ also implies that the name of every gate in $C_{\vn}$ has length bounded by a polynomial function of the length of the input $(g, b, x, \vn)$.
  In the rest of this proof, all our machines will perform the above check first and reject if it fails.
  
  We now consider the language $L_{0}$.
  We receive $(g, b, x, \vn)$ as input and must decide if $g$ is a gate labeled by a constant or if $g$ is an input gate of $D_\vn$ and, if so, whether $g$ evaluates to $b$ on input $x$.
  By the assumed size and uniformity of the circuit family $\cD$, we can decide if $g$ is an input gate in polynomial time (in the length of $(g, b, x, \vn)$) and, if so, determine which input bit $x_i$ labels $g$.
  In this case we accept $(g, b, x, \vn)$ if and only if $b = x_i$, where $x_i$ is the input bit that labels $g$.
  If $g$ is not an input gate, we can check if it has type $0$, $+1$, or $-1$, and accept if $b$ matches the type.
  We reject all other inputs.
  This shows that $L_0 \in \P \subseteq \CH$.

  When $\delta \ge 1$, we will show that the language $L_\delta$ can be decided in $\PP$ with oracle access to $L_{\delta - 1}$.
  As $L_{\delta - 1} \in \CH$ by induction, we can conclude that $L_\delta \in \CH$.
  Let $(g, b, x, \vn)$ be the input to $L_{\delta}$.
  By the assumed size and uniformity of $\mathcal{D}$, we can determine the type of the gate $g$ in time polynomial in the length of $(g, b, x, \vn)$.
  Without loss of generality, we may assume that $\mathcal{D}$ consists only of majority and negation gates, since the case of input and constant gates can be handled as above.
  \begin{itemize}
    \item
      Suppose $g$ is a negation gate.
      We first nondeterministically guess a gate $h$ and verify, using the uniformity of $\mathcal{D}$, whether $h$ is the child of $g$.
      The fact that the names of all gates are bounded by a polynomial function of the length of the input $(g, b, x, \vn)$ allows us to perform this step.
      If $h$ is not the child of $g$, then we nondeterministically branch once more, accepting in one path and rejecting in the other, so that the branches corresponding to non-children of $g$ contribute the same number of accepting and rejecting paths.
      If $h$ is the child of $g$, then we use oracle access to $L_{\delta - 1}$ to check if $(h, 1-b, x, \vn) \in L_{\delta-1}$, accepting $(g, b, x, \vn)$ if and only if $(h, 1-b, x, \vn) \in L_{\delta-1}$.
      The majority of computation paths are accepting exactly when $g$ evaluates to $b$ on input $x$, so we can decide if $(g, b, x, \vn) \in L_{\delta}$ using a $\PP^{L_{\delta-1}}$ algorithm.
    \item
      Suppose $g$ is a majority gate.
      As in the previous case, we nondeterministically guess a gate $h$ and verify if $h$ is a child of $g$, again doing this using the uniformity of $\mathcal{D}$.
      If $h$ is not a child of $g$, we branch into one accepting and one rejecting path so the non-children of $g$ contribute an equal number of accepting and rejecting paths.
      If $h$ is a child of $g$, we use the $L_{\delta-1}$ oracle to decide if $h$ evaluates to $b$ on input $x$ and accept if and only if this is the case.
      The majority of computation paths are accepting exactly when the majority of the children of $g$ evaluate to $b$ on input $x$, so we can decide if $(g, b, x, \vn) \in L_{\delta}$ using a $\PP^{L_{\delta-1}}$ algorithm.
  \end{itemize}
  In both cases, we can decide if $(g, b, x, \vn) \in L_{\delta}$ in $\PP^{L_{\delta-1}}$ as claimed.
  As a consequence, we obtain $L_{\Delta} \in \CH$.

  We now consider $L_{\cD}$.
  Given an input $(i, b, x, \vn)$, the uniformity of $\cD$ also allows us to compute the name of the output gate corresponding to the index $i$.
  We can then decide if the string is a YES instance using an oracle to $L_{\Delta}$.
  This proves $L_{\cD} \in \CH$.
\end{proof}

By combining the preceding lemmas, we conclude that if a family of polynomials $\cF$ can be computed by a uniform family of arithmetic circuits of exponential size and constant depth, and if the corresponding family of boolean functions is $\hat{\cF}$, then $L_{\hat{\cF}}$ is in $\CH$.

\begin{corollary} \label{corollary: arithmetic to ch}
  Let $\cC = \br{C_{\vn}}_{\vn}$ be a polylogtime-uniform family of weakly division-free, constant-free, constant-depth arithmetic circuits of size $s_{\vn}$, where $C_{\vn}$ has $m_{\vn}$ inputs.
  Let $\cD$ denote the circuit family constructed in \cref{lemma: arithmetic to threshold} that computes boolean functions corresponding to the polynomials computed by $\cC$.
  Suppose that $s_{\vn} \le \exp(\poly(m_{\vn}))$.
  Then $L_{\cD} \in \CH$.
\end{corollary}

\begin{proof}
  Combining the assumptions on the size of the circuits in $\cC$ with the bounds on the size of $\cD$ guaranteed by \cref{lemma: arithmetic to threshold}, we see that $\cD$ satisfies the assumptions of \cref{lemma: threshold to ch}.
  The result follows by an application of \cref{lemma: threshold to ch} to $\cD$.
\end{proof}

We need one final simulation result.
Suppose we have a family of functions $\br{f_{n}}_{n}$ that map $n$ bits to $2^{n}$ bits.
Suppose we also have a polylogtime-uniform family $\cD = (D_n)_n$ of polynomial-sized $\TC^{0}$ circuits where $D_{n}$ takes $n$ bits as input.
In this setting, we can define a family of functions obtained by composing $D_{2^{n}}$ with $f_{n}$ for each $n$.
The following lemma states that the language corresponding to this composed function family can be decided in $\CH^{L_{\cF}}$.

\begin{lemma} \label{lemma: threshold with oracle to ch}
  Let $\cD = \br{D_{n}}_{n}$ be a polylogtime-uniform family of polynomial-size threshold circuits of depth $\Delta$, where $D_{n}$ has $n$ inputs.
  Let $\cF = \br{f_{\vn} : \bc{0,1}^{m'_{\vn}} \to \bc{0,1}^{m_{\vn}}}_{\vn}$ be a family of boolean functions.
  Suppose $m_{\vn} \le \exp(\poly\br{m'_{\vn}})$ and that a binary representation of $m_{\vn}$ can be computed in polynomial time given a unary representation of $m'_{\vn}$ and a binary representation of $\vn$.
  Let $\cD \circ \cF = (D_{m_{\vn}} \circ f_{\vn})_{\vn}$ denote the family of composed functions.
  Then $L_{\cD \circ \cF} \in \CH^{L_{\cF}}$.
  In particular, if $L_{\cF} \in \CH$, then $L_{\cD \circ \cF} \in \CH$.
\end{lemma}

\begin{proof}
  We proceed as in the proof of \cref{lemma: threshold to ch}.
  For $\delta \in \bN$, consider the auxiliary language
  \[
    L_{\delta} \coloneqq \setbuild{(g, b, x, \vn)}{\text{$g$ is a gate of depth at most $\delta$ in $D_{m_{\vn}}$ and evaluates to $b$ on input $f_{\vn}(x)$}}
  \]
  By induction on $\delta$, we will show that $L_\delta \in \CH^{L_\cF}$.

  As a first step, we query $L_{\cF}$ on the inputs $(1, 1, x, \vn)$ and $(1, 0, x, \vn)$.
  The length of $x$ is equal to $m'_{\vn}$ if and only if exactly one of these two strings is in $L_{\cF}$, so these queries allow us to ensure that $x$ has the correct length to be a valid input to $f_{\vn}$.
  Whenever $x$ has the correct length, we have $m_{\vn} \le \exp(\poly(\abs{x}))$ by assumption.
  In the rest of this argument, all our machines will perform the above check.

  We now consider the case when $\delta = 0$.
  We receive $(g, b, x, \vn)$ as input and must decide if $g$ is an input gate of $D_{m_{\vn}}$ and, if so, whether $g$ evaluates to $b$ on input $f_{\vn}(x)$.
  Because the circuit family $\cD$ is polylogtime-uniform, and by the observation at the end of the previous paragraph, we can decide if $g$ is an input gate in time polynomial in the length of the input $(g, b, x, \vn)$.
  If $g$ is the $p\ts{th}$ input gate, then we use an oracle call to $L_{\cF}$ to decide if $(p, b, x, \vn) \in L_{\cF}$ and accept if and only if this oracle call returns YES.
  If $g$ is a gate of type $0$ or $1$, then we accept or reject depending on $b$.
  This shows that $L_{\delta} \in \P^{L_{\cF}} \subseteq \CH^{L_{\cF}}$.

  For $\delta \ge 1$, the same argument as in the proof of \cref{lemma: threshold to ch} shows that $L_\delta$ can be decided in $\PP$ with oracle access to $L_{\delta-1}$.
  By induction, we have $L_{\delta-1} \in \CH^{L_\cF}$, so this implies that $L_\delta \in \PP^{\CH^{L_\cF}} = \CH^{L_\cF}$ as desired.

  In particular, by taking $\delta = \Delta$ and using uniformity of $\cD$ as in the proof of \cref{lemma: threshold to ch}, we obtain $L_{\cD \circ \cF} \in \CH^{L_\cF}$.
  If in addition $L_\cF \in \CH$, then we have $L_{\cD \circ \cF} \in \CH^{\CH} = \CH$ as claimed.
\end{proof}

A frequent use case of \cref{lemma: threshold with oracle to ch} will be to compose two uniform families of threshold circuits, where the inner family is of exponential size and the outer family is of polynomial size.
By \cref{lemma: threshold to ch}, the output gates of the inner family can be evaluated in $\CH$.
Combining this with \cref{lemma: threshold with oracle to ch} above shows that the composition can also be evaluated in $\CH$.

We will also encounter the following extension of the above situation.
We have a family of circuits of exponential size, with multiple outputs.
These outputs will be naturally grouped together, and we want to compose each group of outputs with a polynomial sized circuit.
For example, the first circuit family computes the coefficients of a polynomial with coefficients in $\bZ$, and we want to reduce each coefficient modulo a prime $p$.
In this setting, we can derive the same conclusion as \cref{lemma: threshold with oracle to ch}, under the same assumptions.
Essentially the same proof as that of \cref{lemma: threshold with oracle to ch} applies in this setting, therefore we omit the proof.

\begin{lemma} \label{lemma: threshold with oracle to ch multioutput}
  Let $\cD = \br{D_{n}}_{n}$ be a polylogtime-uniform family of polynomial-size threshold circuits of depth $\Delta$, where $D_{n}$ has $n$ inputs.
  Let $\cF = \br{f_{\vn} : \bc{0,1}^{m'_{\vn}} \to \bc{0,1}^{m^{(1)}_{\vn} \times m^{(2)}_{\vn}}}_{\vn}$ be a family of boolean functions.
  Suppose $m_{\vn} \le \exp(\poly(m'_{\vn}))$ and that binary representations of $m^{(1)}_{\vn}$ and $m^{(2)}_{\vn}$ can be computed in polynomial time from a unary representation of $m'_{\vn}$ and a binary representation of $\vn$.
  Let $\cD \circ \cF$ denote the family of functions obtained by composing each of the $m^{(2)}_{\vn}$ sets of outputs of $f_{\vn}$ with the function computed by $D_{m^{(1)}_{\vn}}$.
  Then $L_{\cD \circ \cF} \in \CH^{L_{\cF}}$.
  In particular, if $L_{\cF} \in \CH$, then $L_{\cD \circ \cF} \in \CH$.
\end{lemma}

\subsection{Computing the resultant and deciding the Nullstellensatz} \label{subsection: resultant in ch}

Applying \cref{corollary: arithmetic to ch} to the uniform circuit family for the multivariate resultant obtained in \cref{lemma: uniformity of homotopy resultant}, we conclude that the multivariate resultant over $\bZ$ can be evaluated in the counting hierarchy.
Using this algorithm for the resultant over $\bZ$ as a starting point, we can likewise conclude that the resultant over other domains, such as $\bF_p$, can be computed in the counting hierarchy.
This uses the fact that the resultant over $\bZ$ has integer coefficients and that its image modulo $p$ is precisely the resultant over $\bF_p$ (see \cref{definition: multivariate resultant all fields}).
Because of this, we can compute the resultant over $\bF_p$ by first lifting the input to polynomials with integer coefficients, then computing the resultant over the integers, and finally reducing the result modulo $p$.
Similar ideas allow us to compute resultants over $\bF_{p^a}$ and $\bF_{p^a}[y_1,\ldots,y_k]$.

Recall that in \cref{subsection: finite field arithmetic} we discussed various integral domains and how elements of these domains can be represented in binary.
If $R$ is any such domain, and if $f \in R\bs{x_{1}, \dots, x_{n}}$ is a polynomial, then we can define boolean functions $\hat{f}_{h}$ corresponding to evaluation of $f$ on inputs of height at most $h$, analogous to the previous section where we did this for $R = \bZ$ and $R = \bQ$.

\begin{theorem} \label{theorem: resultant in ch}
  Let $R$ be one of the rings $\bZ, \bZ\bs{y_1,\ldots,y_k}$, $\bF_p$, $\bF_p[y_1,\ldots,y_k]$, $\bF_{p^a}$, or $\bF_{p^a}[y_1,\ldots,y_k]$, where $p$ is a prime number.
  Let $\cF_{R}$ be the family of boolean functions corresponding to the resultant over $R$.
  Then $L_{\cF_{R}} \in \CH$.
\end{theorem}

\begin{proof}
  We proceed by case analysis depending on the choice of the ring $R$.
  \begin{itemize}
    \item
      $R = \bZ$:
      Apply \cref{corollary: arithmetic to ch} to the circuits of \cref{lemma: uniformity of homotopy resultant}.
      The statement of \cref{lemma: uniformity of homotopy resultant} provides bounds on the size and number of inputs of the circuits, and these are seen to satisfy the requirements of \cref{corollary: arithmetic to ch}.
      The resulting threshold circuits treat their inputs as rational numbers represented by pairs of integers.
      These circuits can be evaluated at the given integer inputs by representing the integer $a$ with the pair $(a, 1)$.
      The output of these circuits is the value of the resultant, represented as a rational number $(a,b)$.

      Because the resultant is a polynomial with integer coefficients, its evaluation at an integer-valued point is itself an integer.
      This implies that $a/b$ is an integer, so $b$ necessarily divides $a$.
      To compute the binary representation of $a/b$, we compose the family of threshold circuits for the resultant with a circuit for integer division (item 3 of \cref{theorem: integer ops in tc0}) to obtain a circuit family that computes the integer representation of the resultant.
      The conclusion of \cref{lemma: arithmetic to threshold} provides bounds on the bit complexity of $a$ and $b$, and therefore the requirements to invoke \cref{lemma: threshold with oracle to ch} with the circuits for division are satisfied.
      By \cref{lemma: threshold with oracle to ch}, because the language corresponding to a rational representation of the resultant is in $\CH$, so is the language corresponding to the integer representation of the resultant.
    \item
      $R = \bZ[y_1, \ldots, y_k]$:
      \cref{corollary uniformity of composed resultant} constructs a polylogtime-uniform family of arithmetic circuits that take as input polynomials $F_0, \ldots, F_n \in \bQ[y_1,\ldots,y_k][x_0,\ldots,x_n]$ and outputs the coefficient of the resultant $\res(F_0,\ldots,F_n)$, which is itself a polynomial in $\bQ[y_1,\ldots,y_k]$.
      If the polynomials $F_i$ have degree at most $d$ in the $\vx$ variables and degree at most $D$ in the $\vy$ variables, then these circuits have size bounded by $(n d D)^{\poly(n,k)}$, which is at most singly-exponential in terms of the number of inputs.
      Applying the multi-output analogue of \cref{corollary: arithmetic to ch} to these arithmetic circuits allows us to decide the language corresponding to the resultant over $\bQ[y_1, \ldots, y_k]$ in $\CH$.
      As above, \cref{corollary uniformity of composed resultant} provides bounds on size and number of inputs of the circuits, which are seen to satisfy the requirements in \cref{corollary: arithmetic to ch}.
      The resulting threshold circuits can be evaluated at inputs from $\bZ[y_1, \ldots, y_k]$ by representing the integer $a$ with $(a, 1)$ as in the previous case.

      Similar to the integer case, when the input $F_0, \ldots, F_n$ are elements of $\bZ[y_1,\ldots,y_k][x_0,\ldots,x_n]$, the resultant $\res(F_0, \ldots, F_n)$ is a polynomial with integer coefficients, rather than rational coefficients.
      Therefore, we can use \cref{lemma: threshold with oracle to ch multioutput} to compose the family of threshold circuits for the resultant with circuits for integer division (item 3 of \cref{theorem: integer ops in tc0}) applied to each coefficient to obtain integer representations of the coefficients of $\res(F_0, \ldots, F_n)$.
      By \cref{lemma: threshold with oracle to ch multioutput}, since the language corresponding to the resultant over $\bQ\bs{y_{1}, \dots, y_{k}}$ is in $\CH$, so is the language corresponding to these composed functions that compute the integer resultant.
    \item
      $R = \bF_p$:
      Denote the input polynomials by $F_0, \ldots, F_n \in \bF_p[\vx]$.
      We can lift these to a collection of polynomials $\widehat{F}_0, \ldots, \widehat{F}_n \in \bZ[\vx]$ with integer coefficients by lifting each element of $\bF_p$ to an integer between $0$ and $p-1$ inclusive.
      Using the result from the case $R = \bZ$, we can decide the value of any bit of $\res(\widehat{F}_0, \ldots, \widehat{F}_n)$ in $\CH$.
      The resultant $\res(\widehat{F}_0, \ldots, \widehat{F}_n)$ is an integer.
      The fact that the resultant over $\bF_p$ is the image of the resultant over $\bZ$ under the map $\bZ \to \bF_p$ (see \cref{definition: multivariate resultant all fields}) means that
      \[
        \res(F_0, \ldots, F_n) = \res(\widehat{F}_0, \ldots, \widehat{F}_n) \bmod{p}.
      \]
      Thus, it remains to compute the remainder of $\res(\widehat{F}_0, \ldots, \widehat{F}_n)$ modulo $p$.
      By item 3 of \cref{theorem: integer ops in tc0}, there are logtime-uniform $\TC^0$ circuits for integer division with remainder.
      In particular, there are logtime-uniform $\TC^{0}$ circuits that compute the remainder of an integer when divided by $p$.
      We can compose these with the circuits that compute the resultant of $\widehat{F}_{0}, \dots, \widehat{F}_{n}$ using \cref{lemma: threshold with oracle to ch}.
      Since the language corresponding to the latter is in $\CH$, so is the language corresponding to the composed functions, which is exactly the resultant over $\bF_{p}$.
    \item
      $R = \bF_p[y_1,\ldots,y_k]$:
      We reduce to the case of $R = \bZ[y_1,\ldots,y_k]$ in a manner completely analogous to the case of $R = \bF_p$.
      Let $F_0, \ldots, F_n \in \bF_p[y_1,\ldots,y_k][\vx]$ be the input polynomials, and let $\widehat{F}_0,\ldots,\widehat{F}_n \in \bZ[\vy][\vx]$ be their lifts to the integers.
      Then we have
      \[
        \res(F_0, \ldots, F_n) = \res(\widehat{F}_0, \ldots, \widehat{F}_n) \bmod{p},
      \]
      so we can compute $\res(F_0, \ldots, F_n)$ by first computing $\res(\widehat{F}_0, \ldots, \widehat{F}_n)$ and then reducing each coefficient in this resultant modulo $p$ as in the case of $R = \bF_p$ above, using \cref{lemma: threshold with oracle to ch multioutput}.
    \item
      $R = \bF_{p^a}$:
      Recall that $\bF_{p^a}$ is isomorphic to $\bF_p[z]/(g(z))$, where $g \in \bF_p[z]$ is a degree-$a$ irreducible polynomial.
      Let $F_0, \ldots, F_n \in \bF_{p^a}[\vx]$ be the input polynomials, and let $\widehat{F}_0, \ldots, \widehat{F}_n \in \bZ[z][\vx]$ be their lifts to $\bZ[z]$.
      Let $\hat{g}$ be a monic lift of $g(z)$ to $\bZ\bs{z}$.
      As in previous cases, by the universality of the resultant (see \cref{definition: multivariate resultant all fields}), we have
      \[
        \res(F_0, \ldots, F_n) = \res(\widehat{F}_0, \ldots, \widehat{F}_n) \bmod \ideal{p, \hat{g}(z)}.
      \]
      By the case $R = \bZ[z]$, we can compute $\res(\widehat{F}_0, \ldots, \widehat{F}_n)$ in $\CH$, which is an element of $\bZ[z]$.
      We use the circuit for pseudodivision over $\bZ$ (item 4 of \cref{theorem: integer ops in tc0}) combined with \cref{lemma: threshold with oracle to ch} to pseudodivide this polynomial by $\hat{g}$.
      Since $\hat{g}$ is monic, the pseudoquotient is the actual quotient.
      We then reduce the coefficients modulo $p$ using circuits for integer division as in previous cases to obtain the resultant $\res(F_0, \ldots, F_n)$.
    \item
      $R = \bF_{p^a}[y_1,\ldots,y_k]$:
      We reduce to the case of $R = \bZ\bs{\vy}$ as above by lifting the inputs.
      We then compute the resultant, pseudodivide each coefficient of $\vy$ by a lift of $\hat{g}$, and then reduce every coefficient modulo $p$.
      \qedhere
  \end{itemize}
\end{proof}

By invoking the reduction of \cref{lemma: hn to resultant} from deciding the Nullstellensatz to computing the resultant, we obtain a $\CH$ procedure to decide satisfiability of systems of polynomial equations.
To handle the case where the inputs have coefficients in a number field, we lift to the rational numbers using the following lemma.

\begin{lemma}
  \label{lemma: number field to rationals}
  Let $\bK \coloneqq \bQ\bs{\alpha}$ be a number field and let $g \in \bQ\bs{z}$ be the minimal polynomial of $\alpha$.
  Let $f_{1}, \dots, f_{m} \in \bK\bs{x_{1}, \dots, x_{n}}$ be a set of polynomials.
  Define $f'_{1}, \dots, f'_{m} \in \bQ\bs{x_{1}, \dots, x_{n}, z}$ to be the polynomials obtained from $f_{1}, \dots, f_{m}$ by writing each $\bK$-coefficient as a polynomial in $\bQ\bs{z}$ of degree less than $\deg g$.
  Then $\var{f_{1}, \dots, f_{m}} \neq \emptyset$ if and only if $\var{f'_{1}, \cdots, f'_{m}, g} \neq \emptyset$.
  Further, if the former variety is zero-dimensional and has $r$ points, then the latter variety is zero-dimensional and has $r \cdot \deg g$ points.
\end{lemma}
\begin{proof}
  The rings $\bK\bs{\vx} / \ideal{f_{1}, \dots, f_{m}}$ and $\bQ\bs{\vx, z} / \ideal{f'_{1}, \dots, f'_{m}, g}$ are isomorphic.
  By the Nullstellensatz, $\var{f_{1}, \cdots, f_{m}} \neq \emptyset$ if and only if $\bK\bs{\vx} / \ideal{f_{1}, \dots, f_{m}}$ is not the zero ring.
  Similarly, $\var{f'_{1}, \cdots, f'_{m}, g} \neq \emptyset$ if and only if $\bQ\bs{\vx, z} / \ideal{f'_{1}, \dots, f'_{m}, g}$ is not the zero ring.
  The first statement follows directly from these three facts.

  Assume now that $\var{f_{1}, \dots, f_{m}}$ is a zero-dimensional variety.
  The size of this variety, which we denote by $r$, is exactly the dimension of $\bK\bs{\vx} / \radideal{f_{1}, \dots, f_{m}}$ as a $\bK$-vector space.
  Let $\bL$ be a splitting field of $g$.
  Let $\alpha_{1}, \dots, \alpha_{e}$ be the roots of $g$ in $\bL$, where $e = \deg g$ and $\alpha = \alpha_{1}$.
  The number of points in $\var{f'_{1}, \dots, f'_{m}, g}$ with last coordinate $\alpha_{1}$ is exactly $r$ by definition.

  Let $\sigma_{1}, \dots, \sigma_{e}$ be automorphisms of $\bL$ such that $\sigma_{i}(\alpha) = \alpha_{i}$ for every $i$.
  Define $\bK_{i} :=  \sigma_{i}(\bK)$.
  The isomorphism $\sigma_{i}$ induces an isomorphism $\bK\bs{\vx} / \radideal{f_{1}, \dots, f_{m}} \cong \bK_{i}\bs{\vx} / \radideal{\sigma_{i}(f_{1}), \dots, \sigma_{i}(f_{m})}$.
  Combined with the statements in the previous paragraph, this isomorphism shows that the number of points with last coordinate $\alpha_{i}$ in $\var{f'_{1}, \dots, f'_{m}, g}$ is $r$ for every $i$.
  This completes the proof of the second statement.
\end{proof}

We now combine \cref{lemma: hn to resultant} with \cref{theorem: resultant in ch} to prove that Hilbert's Nullstellensatz can be decided in $\CH$.

\begin{theorem} \label{theorem: nullstellensatz in ch}
  Let $R$ be one of the rings $\bZ$, $\bZ[y_1, \ldots, y_k]$, a number field $\bK$, $\bK\bs{y_{1}, \dots, y_{k}}$, $\bF_p$, $\bF_p[y_1,\ldots,y_k]$, $\bF_{p^a}$, or $\bF_{p^a}[y_1,\ldots,y_k]$, where $p$ is a prime number.
  Then Hilbert's Nullstellensatz over $R$ can be decided in $\CH$.
\end{theorem}

\begin{proof}
  Let $f_1, \ldots, f_m \in R[x_1, \ldots, x_n]$ be the polynomials whose satisfiability we must decide.
  Let $d$ and $h$ be bounds on the degrees and heights, respectively, of the $f_i$.
  If $R$ is either $\bK$ or $\bK\bs{y_{1}, \dots, y_{k}}$ then we use \cref{lemma: number field to rationals} to reduce to the case when $R = \bQ$ or $\bQ\bs{y_{1}, \dots, y_{k}}$, respectively, after which we pass to a common denominator, further reducing to the case $R = \bZ$ or $R = \bZ\bs{y_1, \ldots, y_k}$.
  If $R$ is one of the rings $\bF_{p}, \bF_{p^{a}}, \bF_{p}\bs{y_{1}, \dots, y_{k}}, \bF_{p^{a}}\bs{y_{1}, \dots, y_{k}}$ and if $p$ (or $p^{a}$) is smaller than $15 n d^{n}$, then we use \cref{lemma: passing to bigger finite field} to pass to an extension $\bF_{p^{b}}$ or $\bF_{p^{b}}\bs{y_{1}, \dots, y_{k}}$ such that $p^{b} \geq 15 n d^{n}$.
  The degree of the extension required is polynomial in the input size, therefore doing so only increases the heights of the inputs by a polynomial factor.
  Similarly, arithmetic in the extension can be simulated efficiently using arithmetic in the original ring $R$.
  Further, satisfiability of the system is unchanged by passing to extensions.

  Apply the randomized algorithm of \cref{lemma: hn to resultant} to the input $f_1, \ldots, f_m$ and denote the resulting polynomials by $G_{i,j} \in R[t,w,u,x_0,\ldots,x_n]$ where $i \in \bc{0,1,\ldots,n}$ and $j \in [n]$.
  To decide if the system $f_1 = \cdots = f_m = 0$ is satisfiable, we must decide if there is a $j \in [n]$ such that $(\Tt_{w} \Tt_{t} \res(G_{0,j}, \ldots, G_{n,j}))(0) = 0$.

  This second task can be solved in $\CH$ as follows.
  We iterate over all choices of $j \in [n]$.
  For each $j \in [n]$, we compute $R_j(t,w,u) \coloneqq \res(G_{0,j}, \ldots, G_{n,j})$ in $\CH$ using the algorithm of \cref{theorem: resultant in ch}.
  The resultant $R_j$ is computed with respect to the $\vx$ variables and is a polynomial in $R[t,w,u]$.
  Our task is to decide if $\Tt_w \Tt_t R_j$ has a zero constant term.
  We do this in $\CH$ as follows: we first nondeterministically guess the exponents of the trailing monomials of $R_j$ with respect to $t$ and $w$, and then verify that (1) the corresponding polynomial in $u$ has a zero constant term, and (2) all smaller monomials in $t$ and $w$ have a coefficient of zero.
  The first verification task can be done in $\CH$, since we can compute the coefficients of $R_j$ in $\CH$.
  The second can be done in $\coNP$ with a $\CH$ oracle, since we must verify that all smaller monomials have a zero coefficient, and the coefficient of any single monomial can be computed in $\CH$.
  Since $\coNP^{\CH} = \CH$, this second verification task can likewise be performed in $\CH$.

  Thus, we have a randomized reduction from the task of deciding satisfiability of $f_1 = \cdots = f_m = 0$ to a problem that can be solved in $\CH$.
  This reduction succeeds with probability at least $2/3$, so Hilbert's Nullstellensatz is in $\BPP^{\CH}$.
  Since $\BPP \subseteq \PP \subseteq \CH$, we can decide Hilbert's Nullstellensatz in $\CH$ as claimed.
\end{proof}

\subsection{Counting solutions in zero-dimensional systems}

In \cref{subsection: counting to resultant}, we saw that to count solutions to zero-dimensional systems of equations, it is sufficient to compute multivariate resultants, univariate GCD's, and count the number of distinct roots of a univariate polynomial in the algebraic closure of its coefficient field.
The results of \cref{subsection: resultant in ch} show that we can compute multivariate resultants in $\CH$.
In this section, we develop the additional machinery necessary to compute univariate GCD's and the number of distinct roots of a univariate polynomial.
For our application to counting roots of zero-dimensional systems of equations, we need to perform these operations on polynomials of exponentially-large degree.
The coefficients of these polynomials will be computable in $\CH$, and our goal is to compute the coefficients of the GCD and the number of distinct roots in $\CH$.
To do this, it suffices, by \cref{lemma: threshold with oracle to ch,lemma: threshold with oracle to ch multioutput}, to show that the GCD and number of distinct roots can be computed in polylogtime-uniform $\TC^0$.

We start by observing that the resultant of two polynomials can be computed by polylogtime-uniform $\TC^0$ circuits.
This follows from \cref{lemma: arithmetic to threshold} applied to the circuits we construct for the multivariate resultant of two polynomials.

\begin{lemma} \label{lemma: resultant in tc0}
  Let $R$ be one of the rings $\bZ, \bZ\bs{y_{1}, \dots, y_{k}}, \bQ, \bQ[y_1, \ldots, y_k]$\roberttodo{Do we need the rational case?}, $\bF_p$, $\bF_p[y_1,\ldots,y_k]$, $\bF_{p^a}$, or $\bF_{p^a}[y_1,\ldots,y_k]$, where $p$ is a prime number.
  Then the resultant of two univariate polynomials with coefficients in $R$ can be computed in polylogtime-uniform $\TC^0$.
\end{lemma}

\begin{proof}
  The circuits we construct in \cref{lemma: uniformity of homotopy resultant} and \cref{corollary uniformity of composed resultant} for the multivariate resultant are polynomial sized when $n = 2$.
  Applying \cref{lemma: arithmetic to threshold} to these circuits immediately gives us polylogtime-uniform $\TC^{0}$ circuits over $\bQ$ and $\bQ\bs{\vy}$ respectively.
  For the remaining rings, we use the same arguments as in the proof of \cref{theorem: resultant in ch}, first lifting the input to $\bZ$ or $\bZ[y_1,\ldots,y_k]$, then computing the resultant, and finally projecting back down to the desired ring $R$.
\end{proof}

We note that constant-depth algebraic circuits for the resultant of two polynomials have been constructed prior to our work.
\textcite{AW24} construct circuits for such resultants over $\bQ$ using constant depth versions of the Girard--Newton identities, and \textcite{BKRRSS25b} construct circuits over any field using Lagrange inversion.
Neither of these works addressed the uniformity of their construction, but it is not hard to deduce uniformity of their constructions (in the case of the latter work, only when the field is $\bQ$) using the results in \cref{section: arithmetic circuits}.

\begin{remark} \label{remark: resultant is tc0-complete}
  It is not hard to show that the resultant of two univariate polynomials is complete for polylogtime-uniform $\TC^0$.
  As \cref{lemma: resultant in tc0} shows, the resultant of two polynomials can be computed in polylogtime-uniform $\TC^0$.
  To show that the resultant is $\TC^0$-hard, we reduce from the known $\TC^0$-complete problem of integer powering: given an $n$-bit integer $a$ and an $O(\log n)$-bit integer $b$, compute $a^b$.
  \textcite{HAB02} showed that powering is complete for logtime-uniform $\TC^0$.
  To reduce integer powering to the resultant, observe that if we regard the constant $1$ as a polynomial of degree $b$, then
  \[
    \res(a, 1) = a^b.
  \]
  In particular, to compute $a^b$, it suffices to compute the resultant of two univariate polynomials of degree at most $n^{O(1)}$.
  This proves that the resultant is hard for logtime-uniform $\TC^0$.
\end{remark}

We now show how resultant computations can be used to compute greatest common divisors of univariate polynomials.
The following reduction is from \cite{BKRRSS25b}.

\begin{lemma}[{\cite[Section~1.2]{BKRRSS25b}}]
  \label{lemma: gcd via filtering}
  Let $\bF$ be a field.
  If $f$ and $g$ are univariate polynomials in $\bF\bs{x}$, then
  \[
    \gcd\br{f, g}(y) = \frac{\Tt_{z} \res_{x}\br{z \cdot (y - x) + f(x), f(x) + u \cdot g(x)}}{\Tt_{z} \res_{x}\br{z + f(x), f(x) + u \cdot g(x)}}.
  \]
\end{lemma}
\begin{proof}
  For any univariate polynomials $a, b, c \in \bF\bs{x}$, the Poisson formula implies that, up to a sign, we have
  \[
    \Tt_{z} \res_{x}\br{zb + c, a} = a_{0}^{\max\br{\deg b, \deg c}} \br{\prod_{\alpha \in \var{a} \setminus \var{c}} c(\alpha)^{m(\alpha)}} \times \br{\prod_{\alpha \in \var{a} \cap \var{c}} b(\alpha)^{m(\alpha)}},
  \]
  where $z$ is a new variable, $a_0$ is the leading coefficient of $a$, and $m(\alpha)$ is the multiplicity of $\alpha$ as a root of $a$.
  As a consequence, we have, up to a sign,
  \[
    \prod_{\alpha \in \var{a} \cap \var{c}} b(\alpha)^{m(\alpha)} = a_{0}^{\max\br{0, \deg b - \deg c}} \frac{\Tt_{z} \res_{x}\br{zb + c, a}}{\Tt_{z} \res_{x}\br{z + c, a}}.
  \]

  Now consider the polynomials $f, g \in \bF\bs{x}$ and let $u$ be a new variable.
  As a polynomial with coefficients in $\overline{\bF\br{u}}$, the polynomial $f(x) + u g(x)$ factors into a product of linear forms.
  All common factors of $f(x)$ and $f(x) + u g(x)$ are necessarily factors of $\gcd(f,g)$.
  Moreover, these common factors have the same multiplicity in $f(x) + u g(x)$ as they do in $\gcd(f, g)$.
  The result now follows from the discussion above invoked in the field $\bF\br{u, y}$, with $a(x) = f(x) + u \cdot g(x)$, $b(x) = (y - x)$, and $c(x) = f(x)$.
\end{proof}

By combining \cref{lemma: resultant in tc0} with \cref{lemma: gcd via filtering}, we obtain a polylogtime-uniform family of $\TC^0$ circuits to compute the GCD of two polynomials.

\begin{lemma} \label{lemma: uniform GCD}
  Let $R$ be one of the rings $\bZ$, $\bZ[y_1, \ldots, y_k]$, $\bF_p$, $\bF_p[y_1,\ldots,y_k]$, $\bF_{p^a}$, or $\bF_{p^a}[y_1,\ldots,y_k]$, where $p$ is a prime number.
  Then the greatest common divisor in the fraction field of $R$ of two univariate polynomials with coefficients in $R$ can be computed in polylogtime-uniform $\TC^0$.
  The coefficients of the GCD will be represented as $b_{i} / c$ with $b_{i}, c \in R$.
\end{lemma}

\begin{proof}
  Suppose $f$ and $g$ are the inputs.
  We compute resultants $r_{1} \coloneqq \res_{x}\br{z \cdot (y - x) + f, f + u \cdot g}$ and $r_{2} \coloneqq \res_{x}\br{z + f, f + u \cdot g}$ by invoking \cref{lemma: resultant in tc0} for the ring $R\bs{z, u, y}$.
  From $r_{1}, r_{2}$, the trailing terms $a_{1} \coloneqq \Tt_{z} r_{1}$ and $a_{2} \coloneqq \Tt_{z} r_{2}$ can be computed by inspection of the coefficients as follows.
  For each $0 \le i \le \deg(a_1)$, we first determine if the coefficient of $z^i$ in $r_1$ corresponds to $\Tt_{z} r_1$ by checking that some monomial $z^i u^j y^k$ has a nonzero coefficient in $r_1$ and that all monomials of the form $z^{i'} u^j y^k$ for $i' < i$ have coefficients of zero.
  We then compute the bits of the coefficient of the monomial $u^j y^k$ in $\Tt_{z} r_1$: the $\ell$\ts{th} bit of this coefficient is $1$ exactly when there is some $i$ such that (1) the coefficient of $z^i$ corresponds to $\Tt_z r_1$, and (2) the $\ell$\ts{th} bit of the coefficient of $z^i u^j y^k$ in $r_1$ is a $1$.
  This computation can be carried out by polynomial-size constant-depth boolean circuits built from AND and OR gates in a straightforward manner.
  An entirely analogous computation produces the coefficients of $a_2 = \Tt_z r_2$.
  
  We have $a_{1} \in R\bs{u, y}$ and $a_{2} \in R\bs{u}$.
  By \cref{lemma: gcd via filtering}, we know that $a_{2}$ divides $a_{1}$, that the quotient is independent of $u$, and that the quotient is the GCD of $f$ and $g$.
  Suppose $a_{1} = \sum_i a_{1, i} y^{i}$ with $a_{1, i} \in R\bs{u}$.
  The only way for $a_{1} / a_{2}$ to be independent of $u$ is for each $a_{1, i}$ to be a multiple of $a_{2}$ in the fraction field of $R$.
  This multiple can therefore be read off of just the leading terms; that is, $a_{1, i} / a_{2}$ is simply the leading coefficient of $a_{1, i}$ divided by that of $a_{2}$.
  Setting $c$ to be the leading coefficient of $a_{2}$ and $b_{i}$ to be the leading coefficient of $a_{1, i}$ therefore gives us the GCD in the required form.
\end{proof}

Next, we show how resultant computations can be used to compute factors of a univariate polynomial $f$ that correspond to roots that occur with multiplicity at least $k$.
This computation is inspired by the squarefree decomposition algorithm in \cite{AW24}.
For a univariate polynomial $f$ with coefficients in a field $\bF$, we define $f_{>k}$ as
\[
  f_{>k}(x) \coloneqq \prod_{\alpha \in \var{f}, m(\alpha) > k} (y - \alpha)^{m(\alpha)},
\]
the product of all factors of multiplicity greater than $k$.

\begin{lemma}
  \label{lemma: aw multiplicity more than k part}
  Let $\bF$ be a field.
  If $f$ is a polynomial in $\bF\bs{x}$, then we have
  \[
    f_{>k} = c' \cdot \frac{\Tt_{z} \res_{x}\br{z \cdot (y - x) + \Hasse\br{f} + v \Hasse^{2}\br{f} + \cdots + v^{k-1} \Hasse^{k}\br{f}, f}}{\Tt_{z} \res_{x}\br{z + \Hasse\br{f} + v \Hasse^{2}\br{f} + \cdots + v^{k-1} \Hasse^{k}\br{f}, f}},
  \]
  where $\Hasse^{i}\br{f}$ is the $i^{th}$ Hasse derivative of $f$ and $c' \in \bF \setminus \bc{0}$.
\end{lemma}
\begin{proof}
  We invoke the equation
  \[
    \prod_{\alpha \in \var{a} \cap \var{c}} b(\alpha)^{m(\alpha)} = a_{0}^{\max\br{0, \deg b - \deg c}} \frac{\Tt_{z} \res_{x}\br{zb + c, a}}{\Tt_{z} \res_{x}\br{z + c, a}},
  \]
  with $a(x) = f(x)$, $b(x) = (y - x)$, and $c(x) = \Hasse\br{f} + v \Hasse^{2}\br{f} + \cdots + v^{k-1} \Hasse^{k}\br{f}$, where $v$ is a new variable.
  Any $\alpha \in \var{f} \cap \var{c}$ is necessarily a root of $\Hasse\br{f}, \ldots, \Hasse^{k}\br{f}$.
  The roots of $f$ that are also roots of these Hasse derivatives are exactly those roots of $f$ that occur with multiplicity more than $k$.
  Thus, $f_{>k}(y)$ is precisely $\prod_{\alpha \in \var{f} \cap \var{c}} (y-\alpha)^{m(\alpha)}$ as claimed.
\end{proof}

We also define $f_{=k}$ as $f_{>k-1}(x) / f_{>k}$.
If $f_{=k}$ has degree $r$, then $f(x)$ has exactly $r / k$ distinct roots that occur with multiplicity exactly $k$.
Observe that the denominator in the expression in \cref{lemma: aw multiplicity more than k part} is independent of $y$, so to compute the $y$-degree of $f_{>k}$, it suffices to compute the numerator.
Using this observation, we design a polylogtime-uniform family of $\TC^0$ circuits to compute the number of distinct roots of a univariate polynomial.

\begin{lemma} \label{lemma: uniform root count}
  Let $R$ be one of the rings $\bZ$, $\bZ[y_1, \ldots, y_k]$, $\bF_p$, $\bF_p[y_1,\ldots,y_k]$, $\bF_{p^a}$, or $\bF_{p^a}[y_1,\ldots,y_k]$, where $p$ is a prime number.
  Let $\bK$ be the algebraic closure of the field of fractions of $R$.
  Then the number of distinct roots in $\bK$ of a univariate polynomial with coefficients in $R$ can be computed in polylogtime-uniform $\TC^0$.
\end{lemma}

\begin{proof}
  Let $f \in R\bs{x}$ be the input polynomial.
  First, we compute all the Hasse derivatives $\Hasse^{i}\br{f}$ for $1 \leq i \leq \deg f$ in parallel.
  This involves only integer arithmetic.
  Then we compute the resultants $r_{k} \coloneqq \Tt_{z} \res_{x}\br{z \cdot (y - x) + \Hasse\br{f} + v \Hasse^{2}\br{f} + \cdots + v^{k-1} \Hasse^{k}\br{f}, f}$ using the circuit from \cref{lemma: resultant in tc0} for the ring $R\bs{z, y, v}$.
  By \cref{lemma: aw multiplicity more than k part} and the discussion following it, the $y$-degree of $r_{k}$ is exactly the $y$-degree of $f_{>k}$.
  Once these degrees are computed for all $k$ up to $\deg f$, we can compute
  \[
    \deg f - \deg f_{>1} + \frac{\deg f_{>1} - \deg f_{>2}}{2} + \cdots + \frac{\deg f_{>d-1}}{d}
  \]
  using circuits for integer arithmetic, and this is the desired number of distinct roots of $f$ in $\bK$.
\end{proof}

Now that we have polylogtime-uniform $\TC^0$ algorithms to compute GCD's and the number of roots, we are ready to state the main result of this subsection.
Using \cref{lemma: hn counting to resultant}, we obtain a $\CH$ algorithm to count the number of solutions to a zero-dimensional system of polynomial equations.

\begin{theorem} \label{theorem: counting nullstellensatz in ch}
  Let $R$ be one of the rings $\bZ$, $\bZ[y_1, \ldots, y_k]$, a number field $\bK$, $\bK[y_1, \ldots, y_k]$, $\bF_p$, $\bF_p[y_1,\ldots,y_k]$, $\bF_{p^a}$, or $\bF_{p^a}[y_1,\ldots,y_k]$, where $p$ is a prime number.
  Let $\bL$ be the algebraic closure of the field of fractions of $R$.\roberttodo{Changed from $\bK$ to $\bL$ so that $\bK$ is available for a number field}
  There is a $\FP^{\CH}$ algorithm that counts the number of solutions in $\bL^n$ to zero-dimensional systems of polynomial equations with coefficients in $R$.
\end{theorem}

\begin{proof}
  Let $f_1, \ldots, f_m \in R[x_1, \ldots, x_n]$ be the polynomials whose roots we must count.
  Let $d$ and $h$ be bounds on the degrees and heights, respectively, of the $f_i$.
  If $R$ is either $\bK$ or $\bK\bs{y_{1}, \dots, y_{k}}$, then we use \cref{lemma: number field to rationals} to reduce to the case when $R = \bQ$ or $\bQ\bs{y_{1}, \dots, y_{k}}$, respectively, and then subsequently pass to a common denominator to reduce to the case where $R = \bZ$ or $R = \bZ[y_1, \ldots, y_k]$, since scaling by field elements does not change the number of roots.
  If $R$ is one of the rings $\bF_{p}$, $\bF_{p^{a}}$, $\bF_{p}\bs{y_{1}, \dots, y_{k}}$, or $\bF_{p^{a}}\bs{y_{1}, \dots, y_{k}}$, and if $p$ (or $p^{a}$) is smaller than $100 n d^{2n}$, then we use \cref{lemma: passing to bigger finite field} to pass to an extension $\bF_{p^{b}}$ or $\bF_{p^{b}}\bs{y_{1}, \dots, y_{k}}$ such that $p^{b} \geq 100 n d^{2n}$.
  The degree of the extension required is polynomial in the input size, therefore doing so only increases the heights of the inputs by a polynomial factor.
  Similarly, arithmetic in the extension can be simulated efficiently using arithmetic in the original ring $R$.
  Further, the number of roots is unchanged by passing to extensions.

  Apply the randomized algorithm of \cref{lemma: hn counting to resultant} to the input $f_1, \ldots, f_m$ and denote the resulting polynomials by $G_{i,j} \in R[t,u,x_0,\ldots,x_n]$ where $i \in \bc{0,1,\ldots,n}$ and $j \in \bc{1, 2}$.
  To count the number of roots of $f_1 = \cdots = f_m = 0$, we must count the number of distinct roots of the GCD of $\mathrm{TP}_{u} \Tt_{t} \res(G_{0,1}, \ldots, G_{n,1})$ and $\mathrm{TP}_{u} \Tt_{t} \res(G_{0,2}, \ldots, G_{n,2})$.

  This task can be solved in $\CH$ as follows.
  For $j = 1, 2$, we compute $R_j(t,u) \coloneqq \res(G_{0,j}, \ldots, G_{n,j})$ in $\CH$ using the algorithm of \cref{theorem: resultant in ch}.
  The resultant $R_j$ is computed with respect to the $\vx$ variables and is a polynomial in $R[t,u]$.
  We first nondeterministically guess the exponents of the trailing monomials of $R_j$ with respect to $t$, and then verify  all smaller monomials in $t$ have a coefficient of zero.
  This can be done in $\coNP$ with a $\CH$ oracle, since coefficient in $u$ of the coefficients of any single monomial in $t$ can be computed in $\CH$.
  Since $\coNP^{\CH} = \CH$, this verification task can be performed in $\CH$.
  
  We then nondeterministically guess the exponent of the trailing monomial of $u$ in $\Tt_{t} R_{j}$.
  This again can be done in $\CH$.
  We then have access to the polynomials $\mathrm{TP}_{u} \Tt_{t} R_{j}$ in the form of $\CH$ oracles for the coefficients of each monomial.
  Using \cref{lemma: threshold with oracle to ch} and \cref{lemma: uniform GCD}, we can obtain $\CH$ oracles to the coefficients of the GCD of these two polynomials.
  These coefficients will be represented as ratios with a common denominator, but we can ignore the denominator since scaling will not affect the next step.
  Finally, using \cref{lemma: threshold with oracle to ch} once more with \cref{lemma: uniform root count}, we can count the number of distinct roots, which is the same as the number of roots of the original system.

  Thus, we have a randomized reduction from the task of counting the number of roots of $f_1 = \cdots = f_m = 0$ to a problem that can be solved in $\FP^{\CH}$.
  This reduction succeeds with probability at least $2/3$, so the counting version of Hilbert's Nullstellensatz is in the functional version of $\BPP^{\CH}$, which in turn lies in $\FP^{\CH}$.
\end{proof}

%% file: sections/applications.tex
In this section, we give a few examples of problems that can be reduced to Hilbert's Nullstellensatz.
These problems were previously only known to be in $\PSPACE$, and our results show that they can in fact be solved in $\CH$.
This list is far from exhaustive.
In this section, we let $\bF$ denote any field for which our results hold, namely $\bQ$, a number field $\bK$, a finite field $\bF_{q}$, or a function field over one of these fields.
In the function field case we assume that the inputs are restricted to polynomials instead of arbitrary rational functions.

The first application is to deciding radical ideal membership in the algebraic closure.
Equivalently, this is the problem of testing if a given polynomial vanishes on a given variety.

\roberttodo[inline]{I upgraded these results from lemmas to corollaries}

\begin{corollary}
  \label{lemma: application radical membership}
  Given polynomials $g, f_{1}, \dots, f_{m} \in \bF\bs{x_{1}, \dots, x_{n}}$, deciding if 
  \[
    g \in \radideal{\ideal{f_{1}, \dots, f_{m}} \overline{\bF}\bs{x_{1}, \dots, x_{n}}}
  \]
  can be done in $\CH$.
\end{corollary}
\begin{proof}
  Consider the system of equations $f_{1} = \cdots = f_{m} = 1 - zg = 0$, where $z$ is a new variable.
  By the Nullstellensatz, the condition that $g \in \radideal{\ideal{f_{1}, \dots, f_{m}} \overline{\bF}\bs{x_{1}, \dots, x_{n}}}$ is equivalent to the condition that this system is not satisfiable.
  This is usually called the Rabinowitsch trick.
  Satisfiability can be decided in $\CH$, and therefore so can membership in the ideal $\radideal{\ideal{f_1, \ldots, f_m} \overline{\bF}\bs{x_1, \ldots, x_n}}$.
\end{proof}

The second problem is that of computing the dimension of an algebraic variety given its defining equations.
To make this a decision problem, our algorithms will also take as input a number and will decide if the dimension equals this number.
A similar reduction between computing the dimension of a variety and Hilbert's Nullstellensatz was given by \textcite{koiran1997randomized}.
\begin{corollary}
  \label{lemma: application dimension}
  Given polynomials $f_{1}, \dots, f_{m} \in \bF\bs{x_{1}, \dots, x_{n}}$ and an integer $r$,  deciding if 
  \[
    dim \var{f_{1}, \dots, f_{m}} = r
  \]
  can be done in $\CH$.
\end{corollary}
\begin{proof}
  We design a randomized Turing reduction to satisfiability. 
  Let $d \coloneqq \max \deg f_{i}$ and let $B \subseteq \bF$ be a subset of size at least $100 n^{2} d^{n}$.
  If $\bF$ is not big enough to pick this subset, we pass to a sufficiently-large field extension using \cref{lemma: passing to bigger finite field}.
  Let $\ell_{1}, \dots, \ell_{n}$ be random linear polynomials whose coefficients are picked uniformly and independently from $B$.
  For any integer $k$, if $\dim \var{f_{1}, \dots, f_{m}} \geq k$, then with probability at least $1 - 2 n d^{n} / \abs{B}$, the system $f_{1} = \cdots = f_{m} = \ell_{1} = \cdots = \ell_{k} = 0$ is satisfiable (\cref{lemma: intersect dimension drop}).
  Conversely, if $\dim \var{f_{1}, \dots, f_{m}} < k$, then with the same probability, the system $f_{1} = \cdots = f_{m} = \ell_{1} = \cdots = \ell_{k} = 0$ is unsatisfiable.
  Therefore, based on the satisfiability of these systems for all $k \in [n]$, we can compute the dimension of $\var{f_{1}, \dots, f_{m}}$ and decide if this equals $r$.
\end{proof}

The final application we discuss is the computation of the rank of a tensor over $\overline{\bF}$.
We are given a tensor $T \in \bF^{d_1} \otimes \cdots \otimes \bF^{d_k}$ and a natural number $r \in \bN$ as input, and we must decide if $T$ the rank of $T$ is equal to $r$.
It is straightforward to write down a system of polynomial equations that is solvable if and only if $T$ has rank at most $r$.
Since \cref{theorem: nullstellensatz in ch} allows us to decide the solvability of this system of equations over the algebraic closure $\overline{\bF}$ in $\CH$, we can likewise decide in $\CH$ whether $T$ has rank $r$ when viewed as a tensor over the algebraic closure $\overline{\bF}$.

\begin{corollary}
  \label{lemma: application tensor rank}
  Given an integer $r$ and a tensor $T \in \bF^{d_{1}} \tensor \cdots \tensor \bF^{d_{k}}$ explicitly as a list of entries, deciding if the rank of $T$ as an element of $\overline{\bF}^{d_{1}} \tensor \cdots \tensor \overline{\bF}^{d_{k}}$ is exactly $r$ is in $\CH$.
\end{corollary}
\begin{proof}
  Every tensor in $\bF^{d_1} \otimes \cdots \otimes \bF^{d_k}$ has rank at most $d_1 \cdots d_k$, so if $r > d_1 \cdots d_k$, we immediately reject the input.
  Otherwise, let $x_{i, j, p}$ be a set of variables where $i \in [k]$, $j \in [d_i]$, and $p \in [r]$.
  For each index $(j_1, \ldots, j_k) \in [d_1] \times \cdots \times [d_k]$, we write the polynomial equation
  \[
    T_{j_1, \ldots, j_k} = \sum_{p=1}^r \prod_{i=1}^k x_{i, j_i, p}.
  \]
  Taken together, these equations express that the tensor $T$ can be written as the sum of $r$ rank-one tensors, where the $p$\ts{th} such tensor is the outer product of the vectors $\vx_{1, \bullet, p} \otimes \cdots \otimes \vx_{k, \bullet, p}$ for $\vx_{i, \bullet, p} = (x_{i, 1, p}, \ldots, x_{i, d_i, p})$.
  It follows that this system of equations is satisfiable over $\overline{\bF}$ if and only if $T$ has rank at most $r$ as a tensor over $\overline{\bF}$.

  This is a collection of at most $d_1 \cdots d_k$ equations in at most $k \cdot (d_1 \cdots d_k)^2$ variables.
  Since the tensor $T$ is given as a list of $d_1 \cdots d_k$ field elements, we can write down this system of equations in time that is polynomial in the size of the input.
  Applying the algorithm of \cref{theorem: nullstellensatz in ch} allow us to decide if $T$ has rank at most $r$ in $\CH$.
  To decide if the rank of $T$ is exactly $r$, we use the same procedure to test if $T$ has rank at most $r - 1$, again in $\CH$.
\end{proof}